\documentclass[superscriptaddress,
 amsmath,amssymb,
 aps,
 reprint,
pra
]{revtex4-2}

\usepackage{graphicx}
\usepackage{dcolumn}
\usepackage{bm}
\usepackage{physics}
\usepackage{amsfonts, amsmath}
\usepackage{dsfont}
\usepackage{amsthm}
\usepackage{bm}
\usepackage{multirow}
\usepackage{hyperref}
\usepackage{bbold}
\usepackage{color}
\usepackage{lipsum,babel}
\usepackage[normalem]{ulem}
\usepackage{makecell}
\usepackage{tabularx}
\usepackage{nicematrix}
\usepackage{mathtools}

\newcommand{\T}[0]{{\text{T}}}

\newcommand{\haf}{\text{haf}}
\newcommand{\lhaf}{\text{lhaf}}
\newcommand{\per}{\text{Per}}
\newcommand{\poly}{\text{poly}}

\newtheorem{theorem}{Theorem}

\newtheorem{definition}{Definition}
\newtheorem{lemma}{Lemma}

\begin{document}
\definecolor{navy}{RGB}{46,72,102}
\definecolor{pink}{RGB}{219,48,122}
\definecolor{grey}{RGB}{184,184,184}
\definecolor{yellow}{RGB}{255,192,0}
\definecolor{grey1}{RGB}{217,217,217}
\definecolor{grey2}{RGB}{166,166,166}
\definecolor{grey3}{RGB}{89,89,89}
\definecolor{red}{RGB}{255,0,0}

\preprint{APS/123-QED}

\title{Classical algorithms for measurement-adaptive Gaussian circuits}
\author{Changhun Oh}
\email{changhun0218@gmail.com}
\affiliation{Department of Physics, Korea Advanced Institute of Science and Technology, Daejeon 34141, Korea}
\author{Youngrong Lim}
\affiliation{Department of Physics, Chungbuk National University, Cheongju, Chungbuk 28644, Korea}

\begin{abstract}
Gaussian building blocks are essential for photonic quantum information processing, and universality can be practically achieved by equipping Gaussian circuits with adaptive measurement and feedforward.
The number of adaptive steps then provides a natural parameter for computational power.
Rather than assessing power only through sampling problems---the usual benchmark---we follow the ongoing shift toward tasks of practical relevance and study the quantum mean-value problem, i.e., estimating observable expectation values that underpin simulation and variational algorithms. 
More specifically, we analyze bosonic circuits with adaptivity and prove that when the number of adaptive measurements is small, the mean-value problem admits efficient classical algorithms even if a large amount of non-Gaussian resources are present in the input state, whereas less constrained regimes are computationally hard. 
This yields a task-level contrast with sampling, where non-Gaussian ingredients alone often induce hardness, and provides a clean complexity boundary parameterized by the number of adaptive measurement-and-feedforward steps between classical simulability and quantum advantage. Beyond the main result, we introduce classical techniques---including a generalization of Gurvits' second algorithm to arbitrary product inputs and Gaussian circuits---for computing the marginal quantities needed by our estimators, which may be of independent interest.
\end{abstract}

\maketitle


\section{Introduction}
Quantum computers are promised to provide computational power for solving various problems, such as integer factorization~\cite{shor1994algorithms} and time-dynamics simulation of quantum systems~\cite{lloyd1996universal}, which are believed to be intractable for any classical computer.
However, due to demanding experimental requirements, including a large number of qubits and small error rates, achieving quantum computational advantage for solving practical problems remains highly challenging.
Thus, although the ultimate goal is to build a universal quantum computer, many recent experiments have focused on demonstrating quantum advantage through sampling problems that are relatively easier to implement in experiments~\cite{bouland2019complexity, arute2019quantum, wu2021strong, morvan2023phase, zhong2020quantum, zhong2021phase, madsen2022quantum, deng2023gaussian, young2024atomic, liu2025robust} while many debates are going on regarding the precise boundary between classical simulability and genuine quantum advantage~\cite{neville2017classical,clifford2018classical,oszmaniec2018classical,garcia2019simulating,qi2020regimes,oh2021classical,quesada2022quadratic,bulmer2022boundary,oh2022classical,oh2023classical,liu2023simulating,oh2024classical, oh2025recent, oh2025classical}.

In bosonic systems, particularly in photonic platforms, the most prominent paradigm for quantum advantage is boson sampling due to a relatively more straightforward requirement, a linear-optical circuit~\cite{aaronson2011computational, hamilton2017gaussian, deshpande2021quantum}; consequently, numerous experiments have been conducted to demonstrate quantum advantage using boson sampling~\cite{zhong2020quantum, zhong2021phase, madsen2022quantum, deng2023gaussian, young2024atomic, liu2025robust}.
Despite its hardness results~\cite{aaronson2011computational, hamilton2017gaussian, deshpande2021quantum}, it is also widely recognized that universal quantum computation requires more than linear-optical circuits, such as non-Gaussian operations or measurement-and-feedforward operations~\cite{knill2001scheme, mari2012positive, weedbrook2012gaussian, rahimi2016sufficient, baragiola2019all}.

In theory, nonlinear gates such as Kerr interactions are sufficient for universality~\cite{lloyd1999quantum}.
In practice, however, such gates are extremely difficult to realize because photons rarely interact, and inevitable loss and noise further degrade the effect~\cite{silberhorn2001generation}.
Consequently, feasible bosonic gates in experiments are largely limited to Gaussian unitaries.
Since nonlinearity remains indispensable for universality, many experiments instead rely on non-Gaussian measurements combined with post-selection or feedforward~\cite{knill2001scheme, scheel2003measurement, arrazola2021quantum, bourassa2021blueprint, spagnolo2023non, larsen2025integrated, aghaee2025scaling}.
For example, heralding can generate single-photon states from two-mode squeezed vacua, while GKP states can be approximately prepared by post-selecting certain outcomes in Gaussian boson sampling circuits~\cite{larsen2025integrated, takase2023gottesman, aghaee2025scaling}.
Alternatively, non-Gaussian resources can be supplied entirely through the input states, provided that the measurement and feedforward operations are available, for the universal quantum computation~\cite{baragiola2019all}.
This has motivated recent experimental efforts to go beyond the boson-sampling framework by incorporating non-Gaussian resources, such as measurement-induced nonlinear effects, to unlock more versatile and powerful photonic processors.
Representative directions include nonlinear boson sampling~\cite{spagnolo2023non} and adaptive boson sampling~\cite{chabaud2021quantum, hoch2025quantum}.

Importantly, some of these recent works aim not only to demonstrate sampling-based quantum advantage but also to move towards more practical computational advantages, ultimately aligned with building a universal quantum computer.
From this perspective, it becomes increasingly critical to understand the complexity of simulating bosonic circuits for tasks beyond sampling. 
A particularly important computational task is the quantum mean-value problem because of its generality and practicality (e.g., Refs.~\cite{bravyi2021classical, kim2023evidence}).
More specifically, the quantum mean-value problem encompasses various practical problems, such as variational quantum algorithms and kernel-based machine learning applications~\cite{chabaud2021quantum, yin2024experimental, hoch2025quantum}.
Furthermore, the quantum mean-value problem for a universal quantum circuit is known to be BQP-complete~\cite{bernstein1993quantum, kitaev2002classical}; hence, this problem essentially captures the universal quantum computer's power, provided that the necessary resources are supplied.
This shift reflects a broader transition: extending the role of photonic devices from demonstrating quantum advantage in sampling problems to realizing practical advantages in tasks more relevant to real-world applications.
In this context, clarifying the boundary between classical and quantum computational power is essential.


In this work, we investigate the complexity of simulating bosonic circuits with measurement and feedforward operations, focusing on how the computational cost of the simulation depends on the number of measurements for feedforward operations (See Table~\ref{table:main} for the main results).
Among various notions of simulation~\cite{pashayan2020estimation, hangleiter2023computational}, we mainly analyze and compare two: sampling and the quantum mean-value problem.
While non-Gaussian resources generally render the sampling problem classically intractable~\cite{aaronson2011computational, hamilton2017gaussian, deshpande2021quantum, chabaud2023resources, hahn2024classical}, we show that the quantum mean-value problem can remain efficiently solvable classically even with highly non-Gaussian inputs, as long as the number of adaptive measurements is limited.
Hence, our main results characterize the classical costs of these tasks as functions of the number of measurements and feedforward.
Additionally, it suggests that the requirement for the hardness of the quantum mean-value problem is significantly more demanding.
We show this by presenting classical algorithms that can efficiently estimate the expectation value of a product observable when (i) the number of measurements used for feedforward is at most constant and (ii) the Hilbert-Schmidt norm of the observable is at most polynomially large.

In this work, we use the term \emph{measurement-adaptive Gaussian circuit} for the following class of photonic architectures.
We consider a collection of bosonic modes that undergo layers of Gaussian unitaries, possibly interleaved with $L$ layers of measurements on some of the modes.
The measurement outcomes are processed classically and can be used to choose the subsequent Gaussian unitaries (feedforward).
We allow general, possibly non-Gaussian, product input states on some modes (e.g., single-photon or GKP states), while the dynamical gates acting during the circuit are Gaussian.
The special case $L=0$ corresponds to a non-adaptive Gaussian circuit, whereas $L\ge 1$ yields a measurement-adaptive Gaussian circuit in our terminology.

More specifically, considering practical relevance, we analyze three circuit families in detail:
(i) Gaussian circuits without feedforward,
(ii) Gaussian circuits with photon-number-resolving detection for feedforward, and
(iii) Gaussian circuits with Gaussian measurements for feedforward.
It is worth emphasizing that if a sufficient number of photon-number measurements or Gaussian measurements for feedforward is provided, the circuits (ii) and (iii) become a universal quantum computer~\cite{knill2001scheme, baragiola2019all}.
Hence, the quantum mean-value problem becomes BQP-complete.
Therefore, understanding the complexity of the quantum mean-value problem is crucial for delineating regimes where such circuits remain classically simulable (for example, with a bounded amount of adaptivity) and regimes where they are powerful enough to enable universal quantum computation.
We also emphasize that our classical algorithms are efficient when the Hilbert-Schmidt norm of the observable is at most polynomially large, which is precisely the regime targeted by many proposals for \emph{practical} quantum advantage.
In photonic platforms, such proposals typically consider expectation values of few-mode or $k$-local observables, for example, local Hamiltonian terms, correlation functions, witnesses, or other figures of merit used as quantum-advantage benchmarks, whose Hilbert-Schmidt norms grow at most polynomially with the system size~(e.g., photonic variational quantum algorithms~(VQAs)~\cite{peruzzo2014variational,pappalardo2024photonic,hoch2024variational,facelli2024exact,baldazzi2025four,maring2024versatile}).
In these scenarios, additive-error estimation of such observables is the natural accuracy notion, so our results identify a regime where practically motivated mean-value tasks remain efficiently classically simulable, even though the same architectures become BQP-complete once unrestricted adaptivity is allowed.


To prove the efficiency of our classical algorithms, we exploit low-rank structures of matrices arising in these circuits, inspired by Gurvits' second algorithm, which enables us to compute the marginal probabilities of Fock-state boson sampling~\cite{aaronson2011computational, ivanov2020complexity}.
It is well-known that low-rankness often significantly reduces the computational complexity in evaluating functions such as the permanent, hafnian, and loop hafnian~\cite{barvinok1996two, bjorklund2019faster, oh2024quantum}.
Using these observations, we design efficient classical algorithms for simulating bosonic circuits with a limited number of measurements for feedforward.
In particular, by extending Gurvits' second algorithm to more general setups, our algorithms exploit the low-rank structure of matrices associated with Gaussian circuits, which enables us to approximate relevant quantities with significantly reduced computational costs.
We expect that our extension of Gurvits' second algorithm is of independent interest for the complexity analysis of bosonic systems.

Our paper is organized as follows. 
In Sec.~\ref{sec:prior}, we first compare our work to previously known results to clarify our contributions.
In Sec.~\ref{sec:pre}, we provide preliminaries for our main results.
In Sec.~\ref{sec:setup}, we provide the details of the problem setups.
In Sec.~\ref{sec:algo}, we provide the main results about classical algorithms that simulate measurement-adaptive Gaussian circuits. 
More specifically, in Sec.~\ref{sec:gaussian}, we consider Gaussian circuits without feedforward and provide the classical algorithms that solve the sampling and quantum mean-value problems associated with the circuits, respectively.
In Sec.~\ref{sec:adaptive_photon}, we introduce measurement-adaptive Gaussian circuits by photon-number-resolving detectors and provide the classical algorithms solving the sampling and quantum mean-value problems, respectively.
In Sec.~\ref{sec:adaptive_gaussian}, we consider measurement-adaptive Gaussian circuits by Gaussian measurements and provide the classical algorithms solving the sampling and quantum mean-value problems, respectively.
In Sec.~\ref{sec:marginal}, we provide the main techniques that are necessary for the classical algorithms.
In Sec.~\ref{sec:conclusion}, we conclude.

\begin{table*}[t]
\label{table:main}
\centering
\begin{tabular}{|m{4.5cm}|m{6.0cm}|m{6.0cm}|}
\hline
\centering & \makecell[c]{\bf{Sampling} \\ (Gaussian input and measurement)} & \centering \bf{Quantum mean-value problem} \tabularnewline
\hline
\makecell[c]{Gaussian circuits \\ without feedforward} & \centering Efficient~\cite{mari2012positive, weedbrook2012gaussian, rahimi2016sufficient} & \centering *Efficient~[Sec.~\ref{sec:gaussian}, Theorem~\ref{thm:gaussian}]\tabularnewline
\hline
\makecell[c]{Gaussian circuits +\\ 
photon number counting} & \centering *Efficient when $L=O(\log M)$ \\~[Sec.~\ref{sec:adaptive}, Theorem~\ref{thm:sampling_gaussian_photon}]
& \makecell[c]{ *Efficient when $L=O(1)$\\~[Sec.~\ref{sec:adaptive}, Theorem~\ref{thm:gaussian_photon}],  \\ BQP-complete for general $L$~\cite{knill2001scheme}} \tabularnewline
\hline
\makecell[c]{Gaussian circuits +\\ 
Gaussian measurement} & \centering Efficient~\cite{mari2012positive, weedbrook2012gaussian, rahimi2016sufficient} & \makecell[c]{*Efficient when $L=O(1)$ \\~ [Sec.~\ref{sec:adaptive_gaussian}, Theorem~\ref{thm:gaussian_gaussian}],  \\ 
BQP-complete for general $L$~\cite{baragiola2019all}} \tabularnewline
\hline
\end{tabular}
\caption{Comparison of sampling and quantum mean-value estimation complexities for different families of bosonic quantum circuits: Gaussian circuit without feedforward, Gaussian circuits augmented by photon number counting and Gaussian feedforward, and Gaussian circuits augmented by Gaussian measurement and Gaussian feedforward. Here, $M$ is the number of modes and $L$ is the number of measurement-and-adaptive operations. Asterisks represent the results from this work.}
\end{table*}

\section{Related previous studies}\label{sec:prior}
Conceptually, our study aims to find an analogous classical algorithm in bosonic systems to those for simulating quantum circuits composed mainly of Clifford gates with a few $T$-gates~\cite{aaronson2004improved, bravyi2016improved}.
The analogy is that, in qubit systems, circuits with stabilizer input states, Clifford operations, and computational basis measurements are efficiently simulable on classical computers by the Gottesman-Knill theorem~\cite{gottesman1998heisenberg}, whereas the addition of $T$-gates provides the necessary “magic” for universal quantum computation.
In bosonic systems, Gaussian input states, Gaussian circuits, and Gaussian measurements play the role of stabilizer-Clifford circuits, i.e., easy to classically simulate, while non-Gaussian gates correspond to $T$-gates, which make the circuit hard to classically simulate~\cite{mari2012positive, rahimi2016sufficient}.
Moreover, just as $T$-gates can be simulated via gadgets that rely on measurement and feedforward, we also consider measurement-induced non-Gaussian elements in bosonic circuits~\cite{knill2001scheme, scheel2003measurement}.

This analogy has motivated a line of work on the classical simulation of Gaussian circuits augmented by limited non-Gaussian resources~\cite{veitch2012negative, chabaud2021classical, bourassa2021fast, marshall2023simulation, chabaud2023resources, hahn2024classical, dias2024classical, calcluth2024sufficient, mele2025symplectic} while many of which particularly focuses on strong simulation, i.e., computing the output (marginal) probabilities~\cite{chabaud2021classical, chabaud2023resources, hahn2024classical, dias2024classical}.
For example, Chabaud et al. and Mele et al.~\cite{chabaud2023resources, mele2025symplectic} analyzed the role of non-Gaussian resources in continuous-variable circuits and identified which types of resources are necessary to achieve quantum computational advantage.
Chabaud et al.~\cite{chabaud2021classical}, Dias et al.~\cite{dias2024classical}, and Hahn et al.~\cite{hahn2024classical} developed classical algorithms for simulating Gaussian circuits with a small number of non-Gaussian gates or inputs,  primarily targeting sampling problems and their strong simulations.
Other works, such as Veitch et al.~\cite{veitch2012negative}, explored the role of Wigner function negativity as a resource for quantum computational power, while Bourassa et al.~\cite{bourassa2021fast} investigated fast classical simulation methods for restricted classes of photonic circuits.
More recently, Calcluth et al.~\cite{calcluth2024sufficient} provided sufficient conditions for efficient simulation of bosonic circuits under resource constraints.
Overall, these studies mainly focus on classical simulation in the context of sampling problems, with a focus on strong simulations, rather than the quantum mean-value problem that we emphasize here.

Meanwhile, there is also a rich literature on estimating output probabilities of bosonic circuits with additive error, which is a special case of the quantum mean-value problem.
In particular, quasi-probability-based methods have been studied for estimating output probability distributions~\cite{pashayan2015estimating, pashayan2020estimation, lim2023approximating}.
For example, Pashayan et al.~\cite{pashayan2015estimating, pashayan2020estimation} showed that such estimation techniques can be carried out classically with complexity depending on the negativity of the quasi-probability representation.
While these methods have been applied to several interesting systems, they have not been specifically developed for the class of bosonic circuits we consider in this work.
Additionally, whereas their algorithms concentrate on estimating output probabilities, our work addresses a more general problem: the quantum mean-value problem.
Furthermore, our method does not explicitly depend on the negativity of the quantum circuits, which is in stark contrast to the previous works, making our method more generally applicable.

Finally, our previous work initiated the study of the complexity of the quantum mean-value problem in non-adaptive linear-optical circuits~\cite{lim2025efficient}.
While the main motivation there was to investigate the practicality of boson-sampling-type architectures, passive linear optics alone is not believed to render a universal quantum computer.
More generally, Gaussian circuits by themselves are also non-universal; universality is recovered only when one supplements Gaussian or linear-optical unitaries with suitable non-Gaussian resources or measurement-based feedforward, as in the KLM scheme and Gaussian cluster-state architectures.
Motivated by this, in the present work, we extend this line of research from passive linear optics to a broader class of bosonic circuits, namely measurement-adaptive Gaussian circuits, where Gaussian unitaries are interleaved with a finite number of measurement-and-feedforward layers acting on possibly non-Gaussian resource states.
We emphasize that there are increasing attentions to the adaptive circuits due to their potential practical applications~\cite{chabaud2021quantum, chabaud2024phase, hoch2025quantum}.
From a technical standpoint, we employ a similar method for the main routine of the classical algorithms.
However, extending from linear-optical circuits to Gaussian circuits with measurement and feedforward requires a more challenging subroutine, which is computing marginal probabilities for general Gaussian circuits.
To do this, we extend Gurvits' second algorithm, which works only for Fock states and linear-optical circuits, to arbitrary product states and Gaussian circuits~\cite{ivanov2020complexity, aaronson2011computational}.
To this end, we also provide a classical algorithm for computing the overlap of a low-mode quantum state evolved through a linear-optical circuit with an arbitrary product state.

\section{Preliminary}\label{sec:pre}
In this section, we provide some preliminaries for the main results (more details can be found in Refs.~\cite{cahill1969density, ferraro2005gaussian, serafini2017quantum}).
We will mainly focus on $M$-mode bosonic systems throughout this work.
Here, the annihilation and creation operators for $i$th mode are denoted by $\hat{a}_i$ and $\hat{a}_i^\dagger$, respectively, and they satisfy the canonical commutation relation, $[\hat{a}_i,\hat{a}_j^\dagger]=\delta_{ij}$ and $[\hat{a}_i,\hat{a}_j]=[\hat{a}_i^\dagger,\hat{a}_j^\dagger]=0$.
In such systems, any $M$-mode operator $\hat{O}$ can be written as a linear combination of displacement operators:
\begin{align}\label{eq:expansion}
    \hat{O}
    =\frac{1}{\pi^M}\int d^{2M}\bm{\alpha} \chi_{\hat{O}}(\bm{\alpha})\hat{D}^\dagger(\bm{\alpha}),
\end{align}
where $\chi_{\hat{O}}(\bm{\alpha})\equiv\Tr[\hat{D}(\bm{\alpha})\hat{O}]$ is called the characteristic function of an operator $\hat{O}$, $\hat{D}(\bm{\alpha})$ is an $M$-mode displacement operator of an amplitude $\bm{\alpha}\in\mathbb{C}^M$ which is a tensor product of single-mode displacement operators, i.e., $\hat{D}(\bm{\alpha})=\otimes_{i=1}^M \hat{D}(\alpha_i)=\otimes_{i=1}^M e^{\alpha_i\hat{a}_i^\dagger-\alpha_i^*\hat{a}_i}$ with $\alpha_i\in \mathbb{C}$.
The displacement operator with zero amplitude is equal to the identity operator $\hat{D}(\bm{0})=\hat{\mathbb{1}}_M$, where $\hat{\mathbb{1}}_M$ denotes the identity operator for the $M$-mode system.
An important property of the displacement operator that we frequently use is the following twirling relation:
\begin{align}\label{eq:twirl}
    \frac{1}{\pi^M}\int d^{2M}\bm{\alpha}\hat{D}(\bm{\alpha})\hat{O}\hat{D}^\dagger(\bm{\alpha})
    =\Tr[\hat{O}]\hat{\mathbb{1}}_M.
\end{align}
Note that the convention of integration is $\int d^{2M}\bm{\alpha}=\int_{\mathbb{C}^M} \prod_{j=1}^M d^2\alpha_j$ and $\int_{\mathbb{C}} d^2\alpha_j=\int_{-\infty}^\infty d\Re\alpha_j d\Im\alpha_j$, where $\alpha_j=\Re\alpha_j+i\Im\alpha_j$ throughout this work.

\textit{Gaussian states, circuits, and measurements.---}
In this work, we frequently consider Gaussian states, Gaussian unitary circuits, and Gaussian measurements~\cite{ferraro2005gaussian, weedbrook2012gaussian, serafini2017quantum}.
Thus, we provide definitions for these objects.

First of all, Gaussian states are defined as the ones whose characteristic function (or equivalently, the Wigner function) is Gaussian.
The states that are not Gaussian are called non-Gaussian states.

Gaussian unitary circuits are those that map Gaussian states to Gaussian states.
Otherwise, the circuits are referred to as non-Gaussian.
An essential property of Gaussian unitary circuits is that by the Bloch-Messiah decomposition, any Gaussian unitary circuit can be written as a linear-optical circuit $\hat{U}$, the $M$-mode tensor product of single-mode squeezing gates $\hat{S}(\bm{r})=\otimes_{i=1}^M \hat{S}(r_i)$, and another linear-optical circuit $\hat{V}$, i.e., $\hat{G}=\hat{D}(\bm{\beta})\hat{U}\hat{S}(\bm{r})\hat{V}$, where $\bm{\beta}\in\mathbb{C}^M$ is the displacement vector and $\bm{r}\in \mathbb{R}^M$ is the squeezing parameter vector.
Here, a linear-optical circuit, a restricted family of Gaussian circuits, is defined as an operator that transforms bosonic annihilation operators as
\begin{align}
    \hat{U}^\dagger \hat{a}_i\hat{U}
    =\sum_{j=1}^M U_{ij}\hat{a}_j,
\end{align}
where $U$ is the associated $M\times M$ unitary matrix to the linear-optical circuit $\hat{U}$.
Another important property that is frequently exploited is that a Gaussian unitary circuit transforms a displacement operator into another displacement operator~\cite{ferraro2005gaussian, weedbrook2012gaussian, serafini2017quantum}
\begin{align}\label{eq:G_dis}
    e^{2i\Im(\bm{\alpha}^\dagger\bm{\beta})}\hat{D}^\dagger(\bm{\alpha}')
    =\hat{G}^\dagger \hat{D}^\dagger(\bm{\alpha})\hat{G},
\end{align}
where $\bm{\beta}$ is the displacement in $\hat{G}$'s Bloch-Messiah decomposition and $\bm{\alpha}'\in \mathbb{C}^M$ is a function of $\bm{\alpha}$ that depends on the Gaussian circuit $\hat{G}$.
More specifically, for a Gaussian unitary that transforms the annihilation operators as
\begin{align}
    \hat{G}^\dagger \hat{\bm{a}} \hat{G}=X\hat{\bm{a}}+Y\hat{\bm{a}}^\dagger+\bm{\beta},
\end{align}
where $\hat{\bm{a}}=(\hat{a}_1,\dots,\hat{a}_M)$, we have
\begin{align}\label{eq:gaussian_transform}
    \bm{\alpha}'=X^\dagger \bm{\alpha}-Y^\T\bm{\alpha}^*.
\end{align}
Gaussian measurements are the ones whose POVMs are given as~\cite{serafini2017quantum, weedbrook2012gaussian}
\begin{align}
    \hat{\Pi}(\bm{\alpha})
    =\frac{1}{\pi^M}\hat{D}(\bm{\alpha})\hat{\rho}_G\hat{D}^\dagger(\bm{\alpha}),
\end{align}
where $\hat{\rho}_G$ is a Gaussian state that defines the Gaussian measurement.
Gaussian measurements include homodyne, heterodyne, and general-dyne detections, which are frequently employed in current quantum optical circuits.
Without loss of generality, the Gaussian state $\hat{\rho}_G$ that defines the measurement can be assumed to be a pure Gaussian state $|\psi_G\rangle\langle \psi_G|$ because the Gaussian measurements represented by a mixed Gaussian state can be simulated by those with a pure Gaussian state.
Thus, we write the Gaussian measurements' POVMs as
\begin{align}
    \hat{\Pi}(\bm{\alpha})
    =\frac{1}{\pi^M}\hat{D}(\bm{\alpha})|\psi_G\rangle\langle \psi_G|\hat{D}^\dagger(\bm{\alpha}).
\end{align}
The completeness of POVMs is guaranteed by the twirling property of the displacement operator, Eq.~\eqref{eq:twirl}.

\textit{The median-of-means estimator.---}
Finally, since all the estimators we use for our classical algorithms to solve the quantum mean-value problem are based on the median-of-means estimator~\cite{jerrum1986random, lim2025efficient}, let us provide the relevant parameters.
Suppose we want to estimate the mean of a random variable $X(\bm{\alpha})$:
\begin{align}
    \mathbb{E}[X]=\int d^{2M}\bm{\alpha} p(\bm{\alpha})X(\bm{\alpha}),
\end{align}
where $p(\bm{\alpha})$ is a proper probability distribution, i.e., $\mathbb{E}[X]$ is the average of the random variable $X(\bm{\alpha})$ over $\bm{\alpha}\sim p(\bm{\alpha})$.
Let $\text{Var}(\bm{\alpha})$ be the variance of the random variable $X(\bm{\alpha})$.
Then, the median-of-means estimator is an estimator $\mu$ constructed by taking the median of means of subsamples, and, importantly, it satisfies
\begin{align}\label{eq:mom}
    \Pr[|\mu-\mathbb{E}[X]\geq \epsilon]\leq 1-\delta,
\end{align}
if we use $N$ samples, where $N\geq 544 \text{Var}(X)\epsilon^{-2}\log(1/\delta)$~\cite{jerrum1986random, lim2025efficient}.
Hence, to bound the sample complexity, it suffices to find an upper bound on the variance (or the second moment) of the random variable $X$.

\textit{Sampling by chain rule.---}
Finally, our classical algorithms are randomized algorithms based on sampling using the chain rule of conditional probabilities. 
More specifically, when we have a probability distribution $p(\bm{n})$, one way to sample $\bm{n}$ from the probability distribution is to use the chain rule of the conditional probability:
\begin{align}\label{eq:sampling}
    p(\bm{n})=p(n_1)p(n_2|n_1)\cdots p(n_M|n_1,\dots,n_{M-1}).
\end{align}
Hence, if one can compute the marginal probability, by sampling each random variable $n_i$ using the marginal probabilities, one can sample $\bm{n}$ from the full probability distribution $p(\bm{n})$~\footnote{We note that when marginal probabilities can be efficiently computed, sampling from each mode is efficient because the only scaling parameter, defining the efficiency, comes from the marginal probabilities. However, there may be additional costs in practice for sampling, especially for continuous variables (e.g., rejection sampling)}.
Therefore, the complexity of computing the marginal probabilities determines that of sampling.

\section{Problem setup}\label{sec:setup}
For computational tasks, we will examine two problems: (i) the sampling problem and (ii) the quantum mean-value problem.
We now provide precise definitions of these two problems.

The sampling problem is defined as simulating the output probability distribution of a given quantum circuit, which is the basis of the current quantum advantage experiments, such as random circuit sampling~\cite{bouland2019complexity, arute2019quantum, wu2021strong, morvan2023phase} and boson sampling~\cite{zhong2020quantum, zhong2021phase, madsen2022quantum, deng2023gaussian, young2024atomic, liu2025robust}.
Its definition is as follows:
\begin{definition}[Sampling problem]
    Given as input an efficient description of an $(M+L)$-mode product input state $|\psi\rangle$, a quantum circuit $\mathcal{E}$ acting on $(M+L)$ modes that performs measurements on $L$ modes with adaptive control, and a measurement on the remaining $M$ output modes described by a POVM $\{\Pi_{\bm{x}}\}$, generate output samples according to the probability distribution $p(\bm{x})\equiv \Tr[\Pi_{\bm{x}}\mathcal{E}(|\psi\rangle\langle \psi|)]$.
\end{definition}
In the definition, two parameters, the number of output modes $M$ and the number of modes that are used for measurement for adaptive control $L$, appear.
When $L=0$, there is no adaptive control in the circuit.
Our main focus in this work is to investigate for which scalings of $L$ as a function of $M$ the corresponding sampling problem admits an efficient classical simulation, i.e., can be solved in time polynomial in $M$, depending on the class of quantum circuits $\mathcal{E}$.

Note that in our definition, we choose the notion of sampling as exact sampling; hence, it is the hardest for classical computers (compared to approximate sampling).

For the quantum mean-value problem, the goal is to estimate the expectation value of a product operator $\hat{O}=\hat{O}_1\otimes \cdots \otimes \hat{O}_M$ for an output state of a quantum circuit, $|\psi_\text{out}\rangle$, with an additive error,
\begin{align}
    \langle \psi_\text{out}|\hat{O}|\psi_\text{out}\rangle.
\end{align}
For a finer complexity analysis, we consider a product structure $\hat{O}=\hat{O}_A\otimes \hat{\mathbb{1}}_B$.
Thus,
\begin{align}
    \langle \psi_\text{out}|\hat{O}_A\otimes \hat{\mathbb{1}}_B|\psi_\text{out}\rangle,
\end{align}
where $\hat{O}_A$ is the nontrivial part of the observable on subsystem $A$, still assumed to be in a product form, i.e., $\hat{O}_A=\prod_{i\in A}\hat{O}_i$.
We let $l$ denote the number of modes in the subsystem $A$.

\begin{definition}[Quantum mean-value problem]
    Given as input an efficient description of an $(M+L)$-mode product input state $|\psi\rangle$, a quantum circuit $\mathcal{E}$ acting on $(M+L)$ modes that performs measurements on $L$ modes with adaptive control, and an $M$-mode product observable $\hat{O}=\hat{O}_A\otimes \hat{\mathbb{1}}_B$, estimate the expectation value of $\hat{O}$ with respect to the output state of the circuit, $\Tr[\hat{O}\mathcal{E}(|\psi\rangle\langle \psi|)]$ to within additive error $\epsilon$ with success probability at least $1-\delta$.
\end{definition}
In this setting, the relevant size parameters are again the number of output modes $M$ and the number $L$ of modes used for adaptive control.
Our goal in this work is to identify regimes of $L$ and classes of observables $\hat{O}_A$ for which this quantum mean-value problem admits an efficient classical algorithm, namely one whose running time is polynomial in $M$, $1/\epsilon$, and $\log(1/\delta)$.


For quantum circuits, we will focus on the following three different families:
(i) Gaussian circuits without feedforward,
(ii) Gaussian circuits augmented by photon-number measurement and adaptive Gaussian unitary gates depending on the photon-number measurement outcomes,
(iii) Gaussian circuits augmented by Gaussian measurement and adaptive Gaussian unitary gates depending on the Gaussian measurement outcomes.
Here, the choice of the input states and measurements depends on the tasks.
The main reason that we investigate these three families is the feasibility of their implementations in practice.
One may notice that we assume that all gate elements, including adaptive operations, are composed of Gaussian gates, which are experimentally accessible.
Hence, for the hardness results, non-Gaussian elements, necessary resources for the hardness, are introduced through the input states and measurements.

Whenever the underlying circuit model is universal for quantum computation, in the sense that it can realize any polynomial-size adaptive circuit, the corresponding quantum mean-value problem is BQP-complete: additive-error estimation of suitable observables already captures the full power of a universal quantum computer. In our measurement-adaptive Gaussian architecture, this universality can be achieved by allowing the number $L$ of adaptive measurements to grow polynomially with the system size~\cite{knill2001scheme, baragiola2019all}, whereas our main results show that for any fixed constant $L$ the quantum mean-value problem for a broad class of observables remains efficiently classically simulable.

For the sampling problem, we restrict to Gaussian input states and Gaussian measurements; otherwise, the corresponding sampling problems are already known to be hard to classically simulate due to boson sampling's hardness results~\cite{aaronson2011computational, chakhmakhchyan2017boson, hamilton2017gaussian, deshpande2021quantum}.
Under this restriction, the first and third circuit families are already known to be efficiently simulable, since all elements are Gaussian~\cite{bartlett2002efficient}.
In contrast, the second family is still nontrivial because it necessarily involves non-Gaussian elements.
On the other hand, for the quantum mean-value problem, a general product input state and measurement, which may be non-Gaussian, is considered.

Finally, for simplicity, we will frequently assume that relevant quantum states (not for Gaussian states) have a finite maximum photon number:
\begin{align}\label{eq:psi}
    |\psi\rangle=
    \bigotimes_{i=1}^M\left(\sum_{m_i=0}^{n_\text{max}}a^{(i)}_{m_i}|m_i\rangle\right),
\end{align}
where $|m_i\rangle$ is the Fock state with photon number $m_i$ and $a^{(i)}_{m_i}$ is the coefficient of the $i$th mode's $m_i$ Fock state.
When a given state does not satisfy this, we can approximate it by truncating the state by choosing the maximum photon number $n_\text{max}$ appropriately.
Although it induces an approximation error, such an approximation error is negligible if we choose $n_\text{max}$ sufficiently large.
Since the high photon occupation probability decays quickly in typical physical states, $n_\text{max}$ does not have to be high in practice.
For example, the system with dual-rail qubit encodings, such as the KLM protocol, requires only $n_\text{max}=1$, and few-photon Fock states or weakly squeezed states may not require more than $n_\text{max}$ larger than 10 in general.
As an example, Ref.~\cite{oh2024classical}, simulating large-scale Gaussian boson sampling experiments, uses $n_\text{max}=7\sim10$ to suppress the truncation error.

\section{Classical algorithms}\label{sec:algo}
\subsection{Gaussian circuits without feedforward}\label{sec:gaussian}
In many optical quantum circuits, Gaussian circuits are relatively easy to implement, but not sufficient for universal quantum computation.
Hence, additional resources are introduced through non-Gaussian state preparation, non-Gaussian measurements, and measurement and feedforward~\cite{knill2001scheme, arrazola2021quantum, bourassa2021blueprint, larsen2025integrated, aghaee2025scaling, scheel2003measurement}.
In this section, however, we investigate Gaussian circuits without any measurement and feedforward operations, but allow arbitrary input states and measurements, and we consider the associated computational tasks.

\textit{Sampling.---}
For sampling, we consider the following configuration:
\begin{itemize}
    \item Input state: product state;
    \item Quantum circuit: Gaussian unitary circuit;
    \item Measurement: Gaussian measurement.
\end{itemize}
As emphasized in the previous section, since such a configuration can realize boson sampling circuits~\cite{aaronson2011computational, hamilton2017gaussian, deshpande2021quantum}, classical simulation of their sampling behavior is known to be hard.
On the other hand, if the input state is restricted to Gaussian states and hence the entire circuit is Gaussian, it is well established that such circuits are efficiently classically simulable~\cite{mari2012positive, weedbrook2012gaussian, rahimi2016sufficient}.

\textit{Quantum mean-value problem.---}
We now consider the quantum mean-value problem in this setup with the following configuration:
\begin{itemize}
    \item Input state: product state;
    \item Quantum circuit: Gaussian unitary circuit;
    \item Observable: product observable.
\end{itemize}
Thus, we consider approximation of $\langle \psi|\hat{G}^\dagger\hat{O}\hat{G}|\psi\rangle$, where $\hat{G}$ is a Gaussian unitary circuit, $|\psi\rangle$ is an arbitrary product state, and $\hat{O}=\hat{O}_A\otimes\hat{\mathbb{1}}_B$ is a product observable:
For linear-optical circuits, it was proven in Ref.~\cite{lim2025efficient} that estimating the expectation value within additive error $\epsilon$ with success probability $1-\delta$ requires $O(M^2\|\hat{O}_A\|_2^2 \Tr[\hat{\rho}_A^2] \log(1/\delta)/\epsilon^{2})$ computational cost, where $\hat{\rho}_A\equiv \Tr_B[\hat{G}|\psi\rangle\langle \psi|\hat{G}^\dagger]$ is the reduced density matrix on the system $A$.
Here we extend this result from linear-optical circuits to general Gaussian circuits:

\begin{theorem}[Quantum mean-value problem in Gaussian circuits without feedforward]\label{thm:gaussian}
    Consider an $M$-mode Gaussian circuit $\hat{G}$ and an arbitrary product input state $|\psi\rangle$, and an $M$-mode product operator $\hat{O}=\hat{O}_A\otimes \hat{\mathbb{1}}_B$.  
    The expectation value $\langle \psi|\hat{G}^\dagger \hat{O}\hat{G}|\psi\rangle$ can be approximated within additive error $\epsilon$ with probability $1-\delta$ in running time
    \begin{align}
        O\left(\frac{M^2\|\hat{O}_A\|_2^2\Tr[\hat{\rho}_A^2]\log(1/\delta)}{\epsilon^2}\right),
    \end{align}
    where $\hat{\rho}_A\equiv \Tr_B[\hat{G}|\psi\rangle\langle \psi|\hat{G}^\dagger]$ is the reduced density matrix on the system $A$.
\end{theorem}

\begin{figure*}[t]
\includegraphics[width=460px]{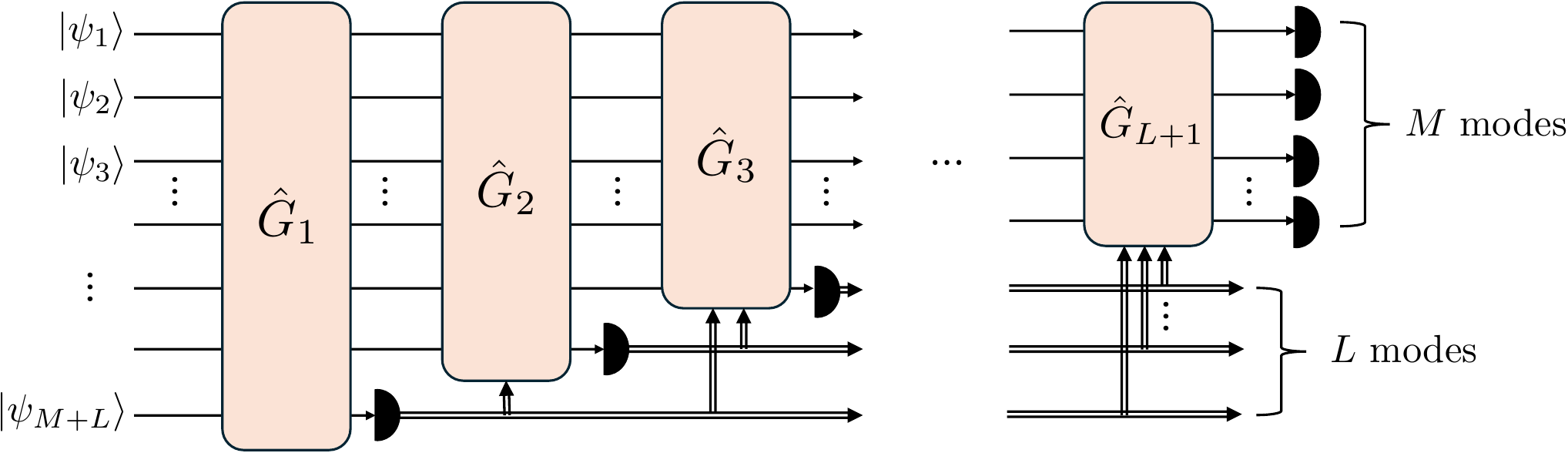} 
\caption{Schematics of quantum circuits of interest in this work. The input state is a product state, and the quantum circuits are composed of Gaussian circuits, measurements, and feedforward operations. More specifically, we consider photon number measurements (Sec.~\ref{sec:adaptive_photon}) and Gaussian measurements (Sec.~\ref{sec:adaptive_gaussian}) for feedforward operations. The final measurement on the computing register is always Gaussian for sampling and implicit for the quantum mean-value problem, depending on the observable. Note that the final measurement can also be adaptive through the final circuit layer $\hat{G}_{L+1}$.}
\label{fig:adaptive}
\end{figure*}

\begin{proof}[Proof Sketch]
Details are provided in Appendix~\ref{app:gaussian}.
Using Eq.~\eqref{eq:expansion}, the mean value can be rewritten as
\begin{align}
    \langle \psi|\hat{G}^\dagger (\hat{O}_A\otimes \hat{\mathbb{1}}_B)\hat{G}|\psi\rangle 
    =\int d^{2l}\bm{\alpha}_A\, p(\bm{\alpha}_A)X(\bm{\alpha}_A),
\end{align}
where $l$ is the number of modes in the subsystem $A$, and we defined a probability distribution and a random variable,
\begin{align}
    p(\bm{\alpha}_A)=\frac{\left|\chi_{\hat{O}_A}(\bm{\alpha}_A)\right|^2}{\pi^l\|\hat{O}_A\|_2^2},~~~
    X(\bm{\alpha}_A)=\frac{\chi^*_{\hat{\rho}_A}(\bm{\alpha}_A)\,\|\hat{O}_A\|_2^2}{\chi^*_{\hat{O}_A}(\bm{\alpha}_A)},
\end{align}
respectively.
By the twirling property, Eq.~\eqref{eq:twirl}, $p(\bm{\alpha}_A)$ is a normalized probability distribution.  
Thus, the expectation value equals the mean of the random variable $X(\bm{\alpha}_A)$, which allows us to employ the median-of-means estimator.

By direct calculation, we show that the variance of $X$ is bounded above by $\|\hat{O}_A\|_2^2\|\hat{\rho}_A\|_2^2$.  
Moreover, the product structure of $p(\bm{\alpha}_A)$ and Eq.~\eqref{eq:G_dis} guarantee that sampling $\bm{\alpha}_A\sim p(\bm{\alpha}_A)$ and evaluating $X(\bm{\alpha}_A)$ can be done efficiently, respectively.  
Explicitly, they can be written in a product structure:
\begin{align}\label{eq:prod}
    p(\bm{\alpha}_A)
    &=\prod_{i\in A} \frac{\left|\chi_{\hat{O}_i}(\alpha_i)\right|^2}{\pi\|\hat{O}_i\|_2^2},\\
    X(\bm{\alpha}_A)
    &=\prod_{i=1}^M \chi^*_{\psi_i}(\alpha'_i)\,
      \prod_{i\in A} \frac{\|\hat{O}_i\|_2^2}{\chi^*_{\hat{O}_i}(\alpha_i)},
\end{align}
where $\bm{\alpha}'$ is defined by
\begin{align}
    \hat{D}^\dagger(\bm{\alpha}')
    =\hat{G}^\dagger \bigl(\hat{D}^\dagger(\bm{\alpha}_A)\otimes \hat{D}^\dagger(\bm{0}_B)\bigr)\hat{G},
\end{align}
and computed in $O(lM)$.
Hence, applying the median-of-means estimator~\eqref{eq:mom} with the above variance bound then yields the claimed complexity.
\end{proof}

Hence, the algorithm runs as follows:
\begin{enumerate}
    \item Sample $\bm{\alpha}_A\in \mathbb{C}^l$ from $p(\bm{\alpha}_A)$ by independently sampling $\alpha_i$ from $|\chi_{\hat{O}_i}(\alpha_i)|^2/(\pi\|\hat{O}_i\|_2^2)$;
    \item Compute $X(\bm{\alpha}_A)$ for the sampled $\bm{\alpha}_A$ by using Eq.~\eqref{eq:gaussian_transform}.
    \item Apply the median-of-mean estimator to estimate the expectation value.
\end{enumerate}
Here, for sampling from $|\chi_{\hat{O}_i}(\alpha_i)|^2/(\pi\|\hat{O}_i\|_2^2)$, one may consider various well-known classical algorithms. For example, for general purposes rejection sampling and Markov Chain Monte Carlo algorithms are natural choices to perform the sampling~\cite{robert1999monte} since we can efficiently compute the probability density.

On the other hand, Theorem~\ref{thm:gaussian} shows that whenever $\|\hat{O}_A\|_2^2 = O(\poly(M))$, there exists a classical algorithm that efficiently solves the quantum mean-value problem for Gaussian circuits.
This condition is naturally met in many relevant scenarios, such as local observables with support on $l = O(1)$ modes and low-rank (possibly global) observables whose Hilbert-Schmidt norm remains polynomially bounded.
The former is particularly relevant for near-term applications that focus on estimating local quantities (e.g., VQAs~\cite{peruzzo2014variational,pappalardo2024photonic,hoch2024variational,facelli2024exact,baldazzi2025four,maring2024versatile}), as discussed in the Introduction.

We note that the dependence of the classical complexity on the number of modes $l$ in the subsystem $A$ enters mainly through the Hilbert-Schmidt norm $\|\hat{O}_A\|_2^2$ and the local purity $\Tr[\hat{\rho}_A^2]$.
When $l = \omega(1)$, the Hilbert-Schmidt norm can become extensive or even grow exponentially with $l$ (for instance, for highly non-local observables with large support), which in turn can make the runtime exponentially large.
In contrast, for genuinely local few-body observables or low-rank observables with bounded spectrum, $\|\hat{O}_A\|_2^2$ typically remains polynomial in $M$, so that this potential source of hardness is absent.
On the other hand, if the output state is highly entangled, the purity of the reduced state $\Tr[\hat{\rho}_A^2]$ can be significantly suppressed.
This suppression always reduces the complexity, but for typical highly entangled states, one still expects at most an inverse-polynomial suppression in $M$ since typical highly entangled states' purity is at least inverse-polynomial in the system size $M$~\cite{fukuda2019typical,iosue2023page}.

We emphasize that the input state may be arbitrary as long as it is a product state.
Hence, even if the input state has a large amount of non-Gaussian resources, such as a high stellar rank~\cite{chabaud2020stellar}, the corresponding quantum mean-value problem is still easy, which is in stark contrast to the sampling problem.

\subsection{Photon-number-measurement-adaptive Gaussian circuits}\label{sec:adaptive}\label{sec:adaptive_photon}
The Gaussian circuits considered in the preceding section are not believed to be as powerful as a universal quantum computer, even when supplemented with non-Gaussian inputs or measurements~\cite{aaronson2011computational}.
Consequently, in many photonic architectures, universal quantum computation is pursued through intermediate measurements and feedforward, as the direct implementation of non-Gaussian gates is often experimentally demanding.
For feasibility, we assume that the feedforward unitaries are Gaussian.  

In this section, we consider photon-number-resolving detections for feedforward operations.  
In other words, conditioned on the detected photon numbers, a measurement-outcome-dependent Gaussian operation is applied.  
We now analyze the computational complexity of simulating such circuits for both sampling and mean-value estimation, with particular emphasis on the dependence on the number of measurements and adaptive steps.

\textit{Measurement and feedforward.---}
We now formulate circuits augmented with measurement and feedforward (see Fig.~\ref{fig:adaptive}).
The circuit consists of a computing register and a measurement register: the modes in the computing register are retained until the final measurement, while the modes in the measurement register are measured in the middle, and the outcomes are used to determine feedforward operations on either the computing or measurement registers.

We begin with a product input state $|\psi\rangle$ and apply a Gaussian unitary $\hat{G}^{(1)}$,
\begin{align}
    |\psi^{(1)}\rangle = \hat{G}^{(1)}|\psi\rangle.
\end{align}
Next, one of the modes is measured using a rank-one POVM $\{|n_1\rangle\langle n_1|\}_{n_1}$, and an adaptive Gaussian unitary is applied to the remaining modes:
\begin{align}
    |\psi^{(1)}\rangle
    \to \hat{G}_{n_1}^{(2)}(\hat{\mathbb{1}}\otimes \langle n_1|)|\psi^{(1)}\rangle 
    \equiv |\tilde{\psi}_{n_1}^{(2)}\rangle,
\end{align}
where $|\tilde{\psi}_{n_1}^{(2)}\rangle$ is a unnormalized state.
More generally, at the $k$th step, we have
\begin{align}
    |\tilde{\psi}_{n_1,\dots,n_{k-1}}^{(k)}\rangle
    &\to \hat{G}_{n_1,\dots,n_k}^{(k+1)}(\hat{\mathbb{1}}\otimes \langle n_k|)|\tilde{\psi}^{(k)}_{n_1,\dots,n_{k-1}}\rangle \\
    &\equiv |\tilde{\psi}_{n_1,\dots,n_k}^{(k+1)}\rangle.
\end{align}
The intermediate (unnormalized) state can be rewritten as
\begin{align}
    |\tilde{\psi}_{n_1,\dots,n_k}^{(k+1)}\rangle
    &= \langle n_1,\dots,n_k|\hat{G}^{(k+1)}_{n_1,\dots,n_k}\cdots \hat{G}^{(2)}_{n_1}\hat{G}^{(1)}|\psi\rangle \\
    &\equiv \langle n_1,\dots,n_k|\hat{G}_{n_1,\dots,n_k}|\psi\rangle \\
    &\equiv \langle n_1,\dots,n_k|\psi_{n_1,\dots,n_k}\rangle,
\end{align}
and the probability of observing outcome $(n_1,\dots,n_k)$ is
\begin{align}
    p(n_1,\dots,n_k) \equiv 
    \langle \tilde{\psi}_{n_1,\dots,n_k}^{(k+1)}|
    \tilde{\psi}_{n_1,\dots,n_k}^{(k+1)}\rangle.
\end{align}

After $L$ steps of measurement and adaptive operations, the final state is written as
\begin{align}
    |\tilde{\psi}_{n_1,\dots,n_L}^{(L+1)}\rangle,
\end{align}
and the corresponding probability distribution factorizes as
\begin{align}
    p(n_1,\dots,n_L)
    = p(n_1)p(n_2|n_1)\cdots p(n_L|n_1,\dots,n_{L-1}).
\end{align}
By measuring the final quantum state $|\tilde{\psi}_{n_1,\dots,n_L}^{(L+1)}\rangle$, the computation is completed.
Now, let us consider the sampling problem and the quantum mean-value problem, respectively.

\textit{Sampling.---}
We now consider a Gaussian input state and a Gaussian measurement at the final stage.
Thus, the configuration is as follows:
\begin{itemize}
    \item Input state: Gaussian input state;
    \item Quantum circuit: Gaussian unitary circuit with $L$ adaptive photon-number measurements;
    \item Final measurement: Gaussian measurement.
\end{itemize}
This restriction is natural because if we allow a non-Gaussian measurement, then even without mid-circuit measurements, the classical simulation of sampling is already hard due to the hardness of boson sampling~\cite{aaronson2011computational, hamilton2017gaussian, deshpande2021quantum}.  
Similarly, if we allow a non-Gaussian input with Gaussian measurements, the resulting configuration has also been proven to be classically hard~\cite{chakhmakhchyan2017boson}. 
Therefore, throughout this subsection, we assume Gaussian inputs and Gaussian final measurements.  

Let $L$ denote the number of modes used for adaptive photon-number measurements.  
We prove that when $L=O(\log M)$, the sampling problem can be simulated efficiently.
More specifically, since the intermediate measurements that are not used for Gaussian adaptive feedforward can be deferred to the end of the circuit, we prove that the sampling problem can be efficiently simulated if the output measurement is composed of $M$ Gaussian measurements and $L$ photon-number resolving measurements, where the latter can be used for Gaussian feedforward.

\begin{theorem}[Sampling in Gaussian circuits with photon-number measurement and Gaussian feedforward]\label{thm:sampling_gaussian_photon}
    Consider an $(M+L)$-mode Gaussian input state, an $(M+L)$-mode Gaussian circuit, and a measurement scheme on the output modes that consists of
    (i) photon-number-resolving measurements on at most $L$ modes and
    (ii) Gaussian measurements on the remaining $M$ modes.
    The photon-number measurements may be performed adaptively, in the sense that arbitrary Gaussian operations may depend on the previous outcomes (Gaussian feedforward).
    If $L = O(\log M)$, then the sampling from the resulting output distribution can be efficiently simulated on a classical computer.
\end{theorem}

\begin{proof}[Proof]
The output probability distribution of the final Gaussian measurement $\{\hat{\Pi}(\bm{\alpha})\}$ under adaptive operations on $M$ modes can be written as
\begin{align}
    p(\bm{\alpha}) 
    &= \sum_{\bm{n}} p(\bm{n})\, p(\bm{\alpha}|\bm{n}) \\ 
    &= \sum_{\bm{n}} p(\bm{n})
    \frac{\langle \psi_{\bm{n}}|(\hat{\Pi}(\bm{\alpha})\otimes |\bm{n}\rangle\langle \bm{n}|)|\psi_{\bm{n}}\rangle}
         {\langle \psi_{\bm{n}}|(\hat{\mathbb{1}}_M\otimes |\bm{n}\rangle\langle \bm{n}|)|\psi_{\bm{n}}\rangle},
\end{align}
where $\bm{n}\equiv (n_1,\dots,n_L)$, $|\psi_{\bm{n}}\rangle \equiv \hat{G}_{\bm{n}}|\psi\rangle$, and
\begin{align}
    p(\bm{n})
    = \langle \psi_{\bm{n}}|(\hat{\mathbb{1}}_M\otimes |\bm{n}\rangle\langle \bm{n}|)|\psi_{\bm{n}}\rangle.
\end{align}
Here, $\hat{G}_{\bm{n}}$ denotes the feedforward Gaussian circuit conditioned on outcome $\bm{n}$, and $\hat{\Pi}(\bm{\alpha})$ is a Gaussian POVM.  
Hence, $p(\bm{n})$ is normalized to 1 as
\begin{align}
    1 = \sum_{\bm{n}} p(\bm{n})
    = \sum_{\bm{n}} \langle \psi_{\bm{n}}|(\hat{\mathbb{1}}_M\otimes |\bm{n}\rangle\langle \bm{n}|)|\psi_{\bm{n}}\rangle
    = \sum_{\bm{n}} \langle \tilde{\psi}_{\bm{n}}|\tilde{\psi}_{\bm{n}}\rangle,
\end{align}
where $|\tilde{\psi}_{\bm{n}}\rangle \equiv \langle \bm{n}|\psi_{\bm{n}}\rangle$ is the unnormalized post-measurement state obtained when the measurement for feedforward is $\bm{n}$.  

Our classical sampling algorithm proceeds as follows.  
First, sample $\bm{n}$ from $p(\bm{n})$ using a chain-rule-based algorithm:
\begin{align}
    p(\bm{n})
    = p(n_1)p(n_2|n_1)\cdots p(n_L|n_1,\dots,n_{L-1}),
\end{align}
where the computational cost mainly depends on the complexity of computing the marginal probabilities $p(n_1,\dots,n_k)$.  
Next, for the obtained $\bm{n}$, sample $\bm{\alpha}$ from
\begin{align}
    p(\bm{\alpha}|\bm{n})
    = \frac{\langle \psi_{\bm{n}}|(\hat{\Pi}(\bm{\alpha})\otimes |\bm{n}\rangle\langle \bm{n}|)|\psi_{\bm{n}}\rangle}
           {\langle \psi_{\bm{n}}|(\hat{\mathbb{1}}_M\otimes |\bm{n}\rangle\langle \bm{n}|)|\psi_{\bm{n}}\rangle},
\end{align}
which completes the procedure.  

The computational bottleneck is thus (i) computing the marginal probabilities of $p(\bm{n})$ and (ii) evaluating $\langle \psi_{\bm{n}}|(\hat{\Pi}(\bm{\alpha})\otimes |\bm{n}\rangle\langle \bm{n}|)|\psi_{\bm{n}}\rangle$.  
Since $|\psi_{\bm{n}}\rangle$ is Gaussian, these quantities reduce to evaluating loop hafnians~\cite{hamilton2017gaussian, quesada2022quadratic}, whose cost scales exponentially in the number of measured modes $L$~\cite{bulmer2022boundary}.
Because $\hat{\Pi}(\bm{\alpha})$ is Gaussian, the numerator also reduces to a loop hafnian of size $L$.  

Therefore, the overall simulation is efficient whenever $L=O(\log M)$, establishing the theorem.
\end{proof}

Hence, in summary, the algorithm runs as follows:
\begin{enumerate}
    \item Sample $\bm{n}$ from $p(\bm{n})$ by using the chain-rule-based algorithm, where $p(\bm{n})$'s marginal probabilities are given by loop hafnians~\cite{hamilton2017gaussian, quesada2022quadratic};
    \item Sample $\bm{\alpha}$ from $p(\bm{\alpha}|\bm{n})$ for a given $\bm{n}$ by using the chain-rule-based algorithm, which again reduces to computing loop hafnians.
\end{enumerate}

\textit{Quantum mean-value problem.---}
We now turn to the quantum mean-value problem in this circuit configuration.  
Thus, we now consider the quantum mean-value problem in this setup with the following configuration:
\begin{itemize}
    \item Input state: product state;
    \item Quantum circuit: Gaussian unitary circuit with $L$ adaptive photon-number measurements;
    \item Observable: product observable.
\end{itemize}
Here, we assume a general product state as the input, and the measurement is implicit since the goal is to estimate the expectation value of an observable~$\hat{O}$.  
We show that if the number of adaptive measurements $L$ is constant, i.e., $L=O(1)$, the observable can be efficiently estimated on a classical computer, provided the Hilbert-Schmidt norm of the operator $\hat{O}$ is at most polynomially large:

\begin{theorem}[Quantum mean-value problem in Gaussian circuits with photon-number measurement and Gaussian feedforward]\label{thm:gaussian_photon}
    Consider an arbitrary product input state and a Gaussian circuit, augmented by photon-number measurements and Gaussian feedforward, which outputs an $M$-mode quantum state, and assume that the input state has a fixed maximum photon number.
    If the number of photon-number measurements used for feedforward is $L=O(1)$, then the expectation value of an $M$-mode observable $\hat{O}=\hat{O}_A\otimes \hat{\mathbb{1}}_B$ can be approximated within additive error $\epsilon$ and success probability $1-\delta$ in running time
    \begin{align}
        O\left(\frac{\poly(M)\|\hat{O}_A\|_2^2 \mathbb{E}_{\bm{n}}[\Tr(\hat{\rho}_{\bm{n}}^2)] \log(1/\delta)}{\epsilon^2}\right),    
    \end{align}
    where $\hat{\rho}_{\bm{n}}$ is the reduced density matrix on the system $A$ when the photon-number measurement outcome is $\bm{n}$ and $\mathbb{E}_{\bm{n}}[\Tr(\hat{\rho}_{\bm{n}}^2)]$ is the average purity.
\end{theorem}

Again, we emphasize that because the input can be an arbitrary product state, it may have a large non-Gaussian resource. However, this resource does not change the complexity, unlike sampling; instead, the number of measurements for feedforward drastically changes the complexity.

Finally, we remark that, as in the sampling task, if the input state is Gaussian instead of an arbitrary state, then sampling from $p(\bm{n})$ and computing $X(\bm{n},\bm{\alpha}_A)$ that appear in the proof remain efficient up to $L=O(\log M)$ because computing $p(\bm{n})$ and $X(\bm{n},\bm{\alpha}_A)$ reduces to computing (loop-)hafnians associated to $L$-mode Gaussian states, which can be computed in $O(\exp(L))$~\cite{bulmer2022boundary}.

\begin{proof}[Proof sketch]
Details are provided in Appendix~\ref{app:gaussian_photon}.  
Consider a general product input $|\psi\rangle$.  
The probability of obtaining measurement outcome $\bm{n}$ is then written as
\begin{align}
    p(\bm{n})=\langle \psi_{\bm{n}}|(\hat{\mathbb{1}}\otimes|\bm{n}\rangle \langle \bm{n}|)|\psi_{\bm{n}}\rangle,
\end{align}
where $|\psi_{\bm{n}}\rangle\equiv \hat{G}_{\bm{n}}|\psi\rangle$ is the (unnormalized) post-measurement state for outcome $\bm{n}$.  

For a given measurement outcome $\bm{n}$, the conditional expectation value of $\hat{O}=\hat{O}_A\otimes \hat{\mathbb{1}}_B$ is given by
\begin{align}
    \frac{\langle \psi|\hat{G}_{\bm{n}}^\dagger(\hat{O}_A\otimes \hat{\mathbb{1}}_B\otimes|\bm{n}\rangle \langle \bm{n}|)\hat{G}_{\bm{n}}|\psi\rangle}{\langle \psi|\hat{G}_{\bm{n}}^\dagger(\hat{\mathbb{1}}_A\otimes \hat{\mathbb{1}}_B\otimes|\bm{n}\rangle \langle \bm{n}|)\hat{G}_{\bm{n}}|\psi\rangle},
\end{align}
and averaging over all outcomes gives
\begin{align}
    \langle\hat{O}\rangle=\sum_{\bm{n}}\langle \psi|\hat{G}_{\bm{n}}^\dagger(\hat{O}_A\otimes \hat{\mathbb{1}}_B\otimes|\bm{n}\rangle \langle \bm{n}|)\hat{G}_{\bm{n}}|\psi\rangle,
\end{align}
which is thus the expectation value we want to estimate.
We now rewrite the expectation value of $\hat{O}$ as
\begin{align}
    \langle\hat{O}\rangle 
    = \int d^{2l}\bm{\alpha}_A \sum_{\bm{n}} p(\bm{n})\, q(\bm{\alpha}_A)\, X(\bm{n},\bm{\alpha}_A),
\end{align}
where $l$ is the number of modes in the subsystem $A$, and we defined a probability distribution and a random variable
\begin{align}
    q(\bm{\alpha}_A) &\equiv \frac{\left|\chi_{\hat{O}_A}(\bm{\alpha}_A)\right|^2}{\pi^l\|\hat{O}_A\|_2^2},\\
    X(\bm{n},\bm{\alpha}_A) &\equiv \frac{\|\hat{O}_A\|_2^2\langle \psi_{\bm{n}}|(\hat{D}^\dagger(\bm{\alpha}_A)\otimes\hat{\mathbb{1}}_B\otimes |\bm{n}\rangle\langle \bm{n}|)|\psi_{\bm{n}}\rangle}{p(\bm{n})\chi^*_{\hat{O}_A}(\bm{\alpha}_A)},
\end{align}
respectively.
Here, the normalization of $q(\bm{\alpha}_A)$ follows from the twirling property of displacement operators, Eq.~\eqref{eq:twirl}.  
Due to the above expression of the expectation value, our algorithm estimates $\langle\hat{O}\rangle$ by sampling $\bm{n}\sim p(\bm{n})$ and $\bm{\alpha}_A\sim q(\bm{\alpha}_A)$, and then applying the median-of-means estimator to $X(\bm{n},\bm{\alpha}_A)$.  

By direct calculation, we upper-bound the variance of $X$ by $\|\hat{O}_A\|_2^2\,\mathbb{E}_{\bm{n}}[\Tr(\hat{\rho}_{\bm{n}}^2)]$, where $\hat{\rho}_{\bm{n}}$ is the reduced density matrix on subsystem $A$ for outcome $\bm{n}$.  

Finally, efficient implementation requires analyzing three subroutines:  
(i) sampling from $p(\bm{n})$,  
(ii) sampling from $q(\bm{\alpha}_A)$, and  
(iii) computing $X(\bm{n},\bm{\alpha}_A)$.  

Sampling from $q(\bm{\alpha}_A)$ is straightforward due to its product structure, cf. Eq.~\eqref{eq:prod}.  
Sampling from $p(\bm{n})$ amounts to computing marginal photon-number probabilities of a state obtained by applying a Gaussian circuit to a general product input. This is nontrivial in general. For Fock input states and linear-optical circuits, Gurvits introduced an algorithm (the so-called Gurvits' second algorithm) that is efficient whenever the number of marginal modes is constant~\cite{aaronson2011computational, ivanov2020complexity}. We extend this algorithm to arbitrary product input states and Gaussian circuits (see Sec.~\ref{sec:marginal} and Appendix~\ref{app:gurvits_general} for more details):
\begin{lemma}[Computing photon-number marginal probabilities in Gaussian circuits]\label{lemma:main}
    Consider an $M$-mode bosonic system, product states $|\psi\rangle$, $|\phi\rangle$, a Gaussian unitary $\hat{G}_0$, and an $L$-mode Fock state $|\bm{n}\rangle$.
    Here, we assume that the states $|\phi\rangle$ and $|\psi\rangle$ have a fixed maximum photon number.
    Then
    \begin{align}
        \langle \phi|\hat{G}_0^\dagger(\hat{\mathbb{1}}\otimes |\bm{n}\rangle\langle \bm{n}|)\hat{G}_0|\psi\rangle
    \end{align}
    can be exactly computed in polynomial time in $M$ if $L=O(1)$.
\end{lemma}
Note that $|\psi\rangle$ and $|\phi\rangle$ may differ, which is an additional generalization beyond Gurvits' original setting.  
Using this lemma, we can efficiently sample from $p(\bm{n})$ whenever $L=O(1)$.  

Finally, evaluating $X(\bm{n},\bm{\alpha}_A)$ requires computing a similar overlap term. Since the lemma also applies when $|\psi\rangle$ and $|\phi\rangle$ differ, we can absorb $\hat{D}^\dagger(\bm{\alpha}_A)$ into the input and compute the random variable $X(\bm{n},\bm{\alpha}_A)$ efficiently when $L=O(1)$.  

This completes the proof of the theorem. More detailed analysis can be found in Sec.~\ref{sec:marginal} and Appendix~\ref{app:gurvits_general}.
\end{proof}

Thus, in summary, the algorithm runs as follows:
\begin{enumerate}
    \item Sample $\bm{\alpha}_A\in \mathbb{C}^l$ from $|\chi_{\hat{O}_A}(\bm{\alpha}_A)|^2$ by sampling $\alpha_i$ independently from $|\chi_{\hat{O}_i}(\alpha_i)|^2$;
    \item Given $\bm{\alpha}_A$, we compute $X(\bm{n},\bm{\alpha}_A)$ by invoking Lemma~\ref{lemma:main};
    \item Apply the median-of-means estimator for the collected samples.
\end{enumerate}

As in Theorem~\ref{thm:gaussian}, the dependence on the system size $A$, $l$, is implicitly determined by the Hilbert-Schmidt norm of the operator $\hat{O}_A$ and the average purity.
For example, the $l=O(1)$ case can be efficiently solved using our algorithm if $L=O(1)$, whereas $l=\omega(1)$ can be hard for our algorithm.
We remark that this clearly suggests a transition between a setup with limited adaptivity and the universal quantum computer through increasing the adaptive steps because while the universal quantum computer, i.e., $L=\poly(M)$ with $l=O(1)$ is sufficient to solve BQP-complete problems, our theorem shows that the $l=O(1)$ case is easy for classical computer.

\subsection{Gaussian-measurement-adaptive Gaussian circuits}\label{sec:adaptive_gaussian}

In this section, instead of photon-number-measurement-based feedforward, we consider feedforward based on Gaussian measurements, which is more experimentally accessible.  
Although the operations remain Gaussian, it is known that if non-Gaussian input states are provided, Gaussian circuits with Gaussian-measurement-based feedforward and Gaussian final measurements are sufficient for universal quantum computation~\cite{baragiola2019all}.

\textit{Sampling.---}
Since this family includes, as a subset, the Gaussian circuits considered in the previous section, simulating the sampling problem with either general non-Gaussian input states or general non-Gaussian measurements is classically hard.  
On the other hand, if both the input states and final measurements are Gaussian, then the entire process remains Gaussian, and simulating such circuits is efficiently possible on a classical computer.

\textit{Quantum mean-value problem.---}
Let us consider the quantum mean-value problem in this circuit configuration with a general product input $|\psi\rangle$.
Thus, we now consider the quantum mean-value problem in this setup with the following configuration:
\begin{itemize}
    \item Input state: product state;
    \item Quantum circuit: Gaussian unitary circuit with $L$ adaptive Gaussian measurements;
    \item Observable: product observable.
\end{itemize}
We now show that the computational complexity for solving this problem is at most polynomial in the system size if the number of measurements for feedforward operations is $L=O(1)$:
\begin{theorem}[Quantum mean-value problem in Gaussian circuit augmented by Gaussian measurement and Gaussian feedforward]\label{thm:gaussian_gaussian}
    Consider an arbitrary product input and a Gaussian circuit, augmented by Gaussian measurements and Gaussian feedforward, which outputs an $M$-mode quantum state, and assume that the input state has a fixed maximum photon number.
    If the number of Gaussian measurements for feedforward is at most constant, i.e., $L=O(1)$, the expectation value of an $M$-mode product observable $\hat{O}\equiv \hat{O}_A\otimes \hat{\mathbb{1}}_B$ can be approximated within additive error $\epsilon$ with probability $1-\delta$ in running time,
    \begin{align}
        O\left(\frac{\poly(M)\|\hat{O}_A\|_2^2 \mathbb{E}_{\bm{\beta}}[\Tr(\hat{\rho}_{\bm{\beta}}^2)] \log(1/\delta)}{\epsilon^2}\right),    
    \end{align}
    where $\hat{\rho}_{\bm{\beta}}$ is the reduced density matrix on the system $A$ when the Gaussian measurement outcome is $\bm{\beta}$ and $\mathbb{E}_{\bm{\beta}}[\Tr(\hat{\rho}_{\bm{\beta}}^2)]$ is the average purity.
\end{theorem}

Interestingly, even if the measurements for feedforward are Gaussian, the complexity increases as in the photon-number measurement case.
While surprising, on a high level, this is consistent with a known result in Ref.~\cite{baragiola2019all}, which demonstrates that Gaussian operations, including measurements and feedforward, are sufficient when GKP states are provided as input.

\begin{proof}[Proof Sketch]
The details of the proof are provided in Appendix~\ref{app:gaussian_gaussian}.
First of all, the probability of obtaining $\bm{\beta}$ for feedforward is 
\begin{align}
    p(\bm{\beta})=\langle \psi_{\bm{\beta}}|(\hat{\mathbb{1}}_A\otimes \hat{\mathbb{1}}_B\otimes\hat{\Pi}(\bm{\beta}))|\psi_{\bm{\beta}}\rangle,
\end{align}
where $|\psi_{\bm{\beta}}\rangle\equiv \hat{G}_{\bm{\beta}}|\psi\rangle$ is the (unnormalized) post-measurement state upon measurement outcome $\bm{\beta}$ and $\{\hat{\Pi}_{\bm{\beta}}=\hat{D}(\bm{\beta})|\psi_G\rangle\langle\psi_G|\hat{D}^\dagger(\bm{\beta})/\pi^L\}$ is Gaussian POVMs.
Thus, the expectation value when the outcome is $\bm{\beta}$ is
\begin{align}
    \frac{\langle \psi|\hat{G}_{\bm{\beta}}^\dagger(\hat{O}_A\otimes\hat{\mathbb{1}}_B\otimes\hat{\Pi}(\bm{\beta}))\hat{G}_{\bm{\beta}}|\psi\rangle}{\langle \psi|\hat{G}_{\bm{\beta}}^\dagger(\hat{\mathbb{1}}_A\otimes \hat{\mathbb{1}}_B\otimes\hat{\Pi}(\bm{\beta}))\hat{G}_{\bm{\beta}}|\psi\rangle}.
\end{align}
If we take average over $\bm{\beta}$, we obtain the expectation value, which is written as
\begin{align}
    \langle\hat{O}\rangle 
    &=\int d^{2L}\bm{\beta}\langle \psi_{\bm{\beta}}|(\hat{O}_A\otimes\hat{\mathbb{1}}_B\otimes \hat{\Pi}(\bm{\beta}))|\psi_{\bm{\beta}}\rangle \\
    &=\int d^{2l}\bm{\alpha}_Ad^{2L}\bm{\beta}p(\bm{\beta})q(\bm{\alpha}_A)X(\bm{\beta},\bm{\alpha}_A),
\end{align}
where $l$ is the number of modes in the subsystem $A$, and we defined a probability distribution
\begin{align}
    q(\bm{\alpha}_A)\equiv \frac{\left|\chi_{\hat{O}_A}(\bm{\alpha}_A)\right|^2}{\pi^l\|\hat{O}_A\|_2^2},
\end{align}
and the random variable
\begin{align}
    X(\bm{\beta},\bm{\alpha}_A)
    =\frac{\|\hat{O}_A\|_2^2\langle \psi_{\bm{\beta}}|(\hat{D}^\dagger(\bm{\alpha}_A)\otimes \hat{\mathbb{1}}_B\otimes \hat{\Pi}
    (\bm{\beta}))|\psi_{\bm{\beta}}\rangle}{p(\bm{\beta})\chi_A^*(\bm{\alpha}_A)}.
\end{align}

By direct calculation, we prove that the variance of the random variable is upper-bounded by $O(\|\hat{O}_A\|_2^2\mathbb{E}_{\bm{\beta}}[\text{Tr}(\hat{\rho}^2_{\bm{\beta}})])$, where $\hat{\rho}_{\bm{\beta}}$ is the reduced density matrix on the Hilbert space $A$ when the measurement outcomes are $\bm{\beta}$ and the average is taken over the probability distribution $p(\bm{\beta})$.

For the algorithm, we need to be able to (i) sample from $q(\bm{\alpha}_A)$, (ii) sample from $p(\bm{\beta})$, and (iii) compute $X(\bm{\beta},\bm{\alpha}_A)$.
As before, (i) is easy because of the product structure of $\hat{O}_A$.
Computing the marginal probability of $p(\bm{\beta})$ can be done efficiently when the number of modes involved in $\hat{\Pi}(\bm{\beta})$ is at most constant, i.e., $L=O(1)$ due to Lemma~\ref{lemma:main}.
Similarly, using Lemma~\ref{lemma:main}, $X(\bm{\beta},\bm{\alpha}_A)$ can be efficiently computed when $L=O(1)$.
Hence, the theorem.

\end{proof}

Thus, in summary, the algorithm runs as follows:
\begin{enumerate}
    \item Sample $\bm{\alpha}_A\in \mathbb{C}^l$ from $|\chi_{\hat{O}_A}(\bm{\alpha}_A)|^2$ by sampling $\alpha_i$ independently from $|\chi_{\hat{O}_i}(\alpha_i)|^2$;
    \item Given $\bm{\alpha}_A$, we compute $X(\bm{\beta},\bm{\alpha}_A)$ by invoking Lemma~\ref{lemma:main};
    \item Apply the median-of-means estimator for the collected samples.
\end{enumerate}

As we have seen from the two previous theorems, Lemma~\ref{lemma:main} plays a crucial role in our classical algorithms.
For the rest of this paper, we introduce the classical algorithm for the lemma.

\section{Classical algorithms using low-rank structures}\label{sec:marginal}
In this section, we present the details of Lemma~\ref{lemma:main}, which plays a crucial role in our classical algorithms.
First, we reduce the task of computing marginal probabilities in the photon-number basis to evaluating an associated generating function (see below for the precise definition).
We then compute this generating function.
As a first step, we show that the amplitudes of few-mode quantum states can be computed efficiently.
Building on this, we demonstrate that the generating function corresponding to photon-number marginals over a small number of modes can likewise be computed efficiently on a classical computer.
The latter can be viewed as an extension of Gurvits' second algorithm~\cite{aaronson2011computational, ivanov2020complexity}.
Conceptually, these results leverage the correspondence between few-mode structure and low-rank matrices, making the relevant quantities amenable to efficient classical computation or simulation.
We expect that the results---and their proofs---are of independent interest in the study of computational complexity for quantum optical systems.

\subsection{Computing the marginal probabilities of Gaussian circuits on photon-number basis}
In this section, we show a way to compute the marginal probability in the photon-number basis for the output state generated by a Gaussian circuit $\hat{G}_0$ applied to a product state $|\psi\rangle$.
Let us consider a general expression:
\begin{align}
    q(\bm{n})
    =\langle \phi|\hat{G}_0^\dagger |\bm{n}\rangle\langle \bm{n}|\hat{G}_0|\psi\rangle.
\end{align}
and its marginal
\begin{align}
    &q(n_1,\dots,n_L) 
    =\sum_{n_{L+1},\dots,n_M=0}^\infty q(\bm{n}) \\
    &=\langle \phi|\hat{G}_0^\dagger (|n_1,\dots,n_L\rangle\langle n_1,\dots,n_L|\otimes \hat{\mathbb{1}})\hat{G}_0|\psi\rangle.
\end{align}
One may see that when $|\phi\rangle=|\psi\rangle$, $q(\bm{n})$ represents the (marginal) probability in the photon number basis.
Thus, we call this function a marginal probability in a broader sense even when $|\psi\rangle\neq |\phi\rangle$.
We have seen that this function plays a central role in the classical algorithms introduced in Sec.~\ref{sec:algo}.
We present a general recipe to compute this function below.
Note that we now assume that $|\psi\rangle$ and $|\phi\rangle$ have a maximum photon number cutoff $n_\text{max}$ for simplicity, as discussed in Sec.~\ref{sec:setup}.

Before considering general cases, we first note that in the case of a Gaussian state $|\psi\rangle$, this is expressed simply as a loop hafnian with matrix size given by~$L$~\cite{bulmer2022boundary, quesada2019franck, bjorklund2019faster} and thus the complexity is exponential in the number of modes.
Thus, when $L$ is at most logarithmic, we can compute this quantity in polynomial time.
However, in more general cases, i.e., for non-Gaussian states, it may not be written as a simple loop hafnian.
Remarkably, for Fock-state input $|\psi\rangle=|\bm{m}\rangle$ and linear-optical circuit, there is a well-known algorithm, so-called Gurvits' second algorithm, which enables us to compute $q(\bm{n})$ with complexity $O(N^{2L+1})$, where $N$ is the total photon number of $|\bm{m}\rangle$~(see Sec.~\ref{sec:poly}) ~\cite{aaronson2011computational, ivanov2020complexity}.
Hence, when $L=O(1)$, i.e., the number of marginal modes is at most constant, the corresponding quantity can be efficiently computed using classical computers.
Our main contribution is to generalize this result for Fock states and linear-optical circuits to general product input states $|\phi\rangle$ and $|\psi\rangle$ and Gaussian circuits.

The main idea is to exploit the following Fourier relation~\cite{aaronson2011computational, oh2024quantum} (see Appendix~\ref{app:MVS} for the details):
\begin{align}\label{eq:g_k}
    \tilde{G}(k)
    &\equiv \langle \phi|\hat{G}_0^\dagger e^{-ik\theta \hat{\bm{n}}\cdot \bm{\omega}}\hat{G}_0|\psi\rangle, \\ 
    G(\Omega)
    &\equiv \sum_{\bm{n}=\bm{0}}^\infty q(\bm{n})\delta(\Omega-\bm{\omega}\cdot \bm{n})
    =\frac{1}{\Omega_\text{max}+1}\sum_{k=0}^{\Omega_\text{max}}\tilde{G}(k)e^{ik\Omega},\label{eq:g_omega}
\end{align}
where $\hat{\bm{n}}\equiv (\hat{n}_1,\dots,\hat{n}_M)$ is the photon number operator vector for $M$ modes, where $\hat{n}_i\equiv\hat{a}_i^\dagger\hat{a}_i$, $k,\Omega\in \{0,\dots,\Omega_\text{max}\}$, $\theta\equiv 2\pi/(\Omega_\text{max}+1)$, and $\bm{\omega}\in \mathbb{R}^M$.
Thus, if we want to compute the marginal probability of $q(\bm{n})$ on the Fock basis up to the $L$th mode, we set 
\begin{align}
    \bm{\omega}=(\omega_1,\dots,\omega_L,0,\dots,0),
\end{align}
where $\omega_i=(n_\text{max}+1)^{i-1}$.
Then, each $\Omega \in \{0,\dots,(n_\text{max}+1)^L-1\}$ corresponds to distinct marginal probabilities of $q(\bm{n})$ where $n_i\leq n_\text{max}$.
More explicitly, if we want to compute $q(n_1,\dots,n_L)$, we set $\Omega=\sum_{i=1}^L \omega_i n_i$ with $\bm{\omega}$ defined above and compute the corresponding $G(\Omega)$, which is equivalent to $q(n_1,\dots,n_L)$.
Therefore, if we can compute $\tilde{G}(k)$, in principle, we can compute $G(\Omega)$ and thus its marginal probabilities.
Here, because we need to sum $k$ from $0$ to $\Omega_\text{max}$, the complexity increases as $O((n_\text{max}+1)^L)$.
With this in mind, computing $q(\bm{n})$ or its marginal probabilities boils down to exactly computing
\begin{align}
    \langle \phi|\hat{G}_0^\dagger \hat{P}(\bm{\varphi})\hat{G}_0|\psi\rangle
    \equiv \langle \phi|\hat{G}|\psi\rangle,
\end{align}
where $\hat{P}(\bm{\varphi})\equiv  e^{i\bm{\varphi}\cdot\hat{\bm{n}}}$ is a phase shifter with $(\varphi_1,\dots,\varphi_L,0,\dots,0)$, which has exactly the same form of $\tilde{G}(k)$.
We refer to this function as a generating function.
Hence, the main complexity of this algorithm is attributed to computing $\tilde{G}(k)$.

\subsection{Generalized Gurvits' second algorithm}
In this section, we prove that 
\begin{align}
    \tilde{G}(k)=\langle \phi|\hat{G}_0^\dagger \hat{P}(\bm{\varphi})\hat{G}_0|\psi\rangle
\end{align}
is easy to compute when $|\psi\rangle$ and $|\phi\rangle$ are product states and $\bm{\varphi}$ has at most constant non-zero phases.
Without loss of generality, let $\bm{\varphi}=(\varphi_1,\dots,\varphi_L,0,\dots,0)$ throughout this section.

As mentioned above, it has been known for Fock states and linear-optical circuits, $|\psi\rangle=|\phi\rangle=|\bm{n}\rangle$ and $\hat{G}_0=\hat{U}$ is a linear-optical circuit, that $\tilde{G}(k)$ can be computed in a polynomial time when $L$ is at most a constant:
\begin{lemma}
    [Gurvits' second algorithm~\cite{aaronson2011computational, ivanov2020complexity}]
    Consider an $M$-mode linear-optical circuit $\hat{U}$, a phase shifter operator $\hat{P}(\bm{\varphi})$ with $\bm{\varphi}\equiv (\varphi_1,\dots,\varphi_L,0,\dots,0)$, and a Fock state $|\bm{n}\rangle$.
    If $L=O(1)$,
    \begin{align}
        \langle \bm{n}|\hat{U}^\dagger \hat{P}(\bm{\varphi})\hat{U}|\bm{n}\rangle
    \end{align}
    can be exactly computed in polynomial time in the total photon number in $|\bm{n}\rangle$ and the number of modes $M$.
\end{lemma}

We generalize this to arbitrary product states $|\psi\rangle$ and $|\phi\rangle$, which may be different, and an arbitrary Gaussian circuit $\hat{G}_0$:
\begin{lemma}[Generating function]\label{lemma:MVS}
    Consider an $M$-mode bosonic system. Let $\hat{G}_0$ be a Gaussian unitary circuit and $\hat{P}(\bm{\varphi})$ be the phase shifter with phase $\bm{\varphi}$ and $|\psi\rangle$ and $|\phi\rangle$ be product states.
    If the number of non-zero elements in $\bm{\varphi}\equiv (\varphi_1,\dots,\varphi_L,0,\dots,0)$ is $L=O(1)$,
    \begin{align}
        \langle \phi |\hat{G}_0^\dagger\hat{P}(\bm{\varphi})\hat{G}_0|\psi\rangle
    \end{align}
    can be exactly computed in polynomial time in $M$ and $n_\text{max}$.
\end{lemma}

Using the above lemma, we can immediately see that Lemma~\ref{lemma:main} is true by using the relationship between $G(\Omega)$ and $\tilde{G}(k)$, which is restated here:
\begin{lemma}[Lemma~\ref{lemma:main} restated]\label{lemma:main_restated}
    Consider an $M$-mode bosonic system, product states $|\psi\rangle$, $|\phi\rangle$, a Gaussian unitary $\hat{G}_0$, and an $L$-mode Fock state $|\bm{n}\rangle$.  
    Then
    \begin{align}
        \langle \phi|\hat{G}_0^\dagger(|\bm{n}\rangle\langle \bm{n}|\otimes \hat{\mathbb{1}})\hat{G}_0|\psi\rangle
    \end{align}
    can be exactly computed in polynomial time in $M$ and $n_\text{max}$ if $L=O(1)$.
\end{lemma}
Hence, from now on, we will focus on proving Lemma~\ref{lemma:MVS}.

Before we do that, we emphasize that the generating function is reduced to the permanent of a matrix with rank depending on $L$~\cite{aaronson2011computational, ivanov2020complexity} when $|\psi\rangle=|\phi\rangle$ are Fock states and the circuit is a linear-optical circuit. We will demonstrate how to compute this quantity as a warm-up example to generalize the approach further below.
On the other hand, for Fock states and general Gaussian circuits, the generating function reduces to a loop hafnian~\cite{oh2024quantum, quesada2019franck}, which also allows a classical algorithm to exploit a low-rank structure to reduce the computational complexity~\cite{oh2024quantum}.
Notably, such a structure still allows us to compute the generating function even when $|\psi\rangle$ and $|\phi\rangle$ are general product states by using the classical algorithm from Ref.~\cite{oh2024quantum}, which will be shown in Appendix~\ref{app:gurvits_general}.

\subsubsection{Warm-up example: Permanent}\label{sec:poly}
To provide the main technique, we first consider the simplest example: permanent.
The permanent of an $N\times N$ matrix $A$ is defined as
\begin{align}
    \per(A)\equiv \sum_{\sigma\in \mathcal{S}_N}\prod_{i=1}^N A_{i,\sigma(i)}.
\end{align}
There are many ways to compute the permanent of a matrix, such as Ryser's formula~\cite{ryser1963combinatorial}.
One of the methods that is relevant to our work is the polynomial method~\cite{ryser1963combinatorial}.
To be more specific, we define a polynomial of indeterminates $x_1,\dots,x_N$ as
\begin{align}
    F(x_1,\dots,x_N)=\prod_{i=1}^N\left(\sum_{j=1}^N A_{ij}x_j\right).
\end{align}
One may then notice that the coefficient of $x_1\cdots x_N$ corresponds to the permanent of the matrix $A$.

To illustrate how we leverage the low-rank structure of a relevant matrix, let us review Gurvits' second algorithm~\cite{ivanov2020complexity, aaronson2011computational}.
Let $A=W+I$, where $W$ is a rank-$L$ matrix, i.e., $W$'s matrix elements can be written as
\begin{align}
    W_{ik}
    \equiv \sum_{s=1}^L u_i^{(s)}v_{k}^{(s)}.
\end{align}
Then, Ref.~\cite{ivanov2020complexity} shows that if we define a polynomial of indeterminates $\{a_u^{(s)}\}_{s=1}^L,\{a_v^{(s)}\}_{s=1}^L$,
\begin{align}
    F(\{a_u^{(s)}\}_{s=1}^L,\{a_v^{(s)}\}_{s=1}^L)
    =\prod_{i=1}^N\left[1+\sum_{s,s'=1}^L a_u^{(s)}a_v^{(s')}u_i^{(s)}v_i^{(s')}\right]
\end{align}
and sum the coefficients of monomials whose degrees of $\{a_u^{(s)}\}$ and $\{a_v^{(s)}\}$ match each other with multiplying the factorial of the degrees of $\{a_u^{(s)}\}$ (or equivalently $\{a_v^{(s)}\}$), then we obtain the permanent.
Note that the number of monomials that appear in the computation is $O(N^{2L})$.
The calculation of the coefficients of the monomials involves multiplying out $N$ terms, where at each multiplication the $O(N^{2L})$ coefficients have to be updated.
Therefore, the total cost is $O(N^{2L+1})$.
Thus, when the rank of the matrix $W$ is at most a constant, the permanent of the matrix $W+I$ can be efficiently computed, unlike that of general matrices.

\subsubsection{Amplitude of low-mode states}\label{sec:low-mode}
In the section, for the subroutine of our algorithm for computing the generating function, we consider the amplitude in the form of $\langle \phi|\hat{U}|\psi\rangle$, where $\hat{U}$ is a linear-optical circuit and we assume that $|\phi\rangle$ is an $L$-mode state, which is defined as
\begin{align}
    |\phi\rangle=|\phi'\rangle\otimes |\bm{0}_{M-L}\rangle,
\end{align}
where $|\phi'\rangle$ is the nontrivial part and an arbitrary product state on $L$ modes and $|\bm{0}_{M-L}\rangle$ is the vacuum state on the remaining $M-L$ modes.
We show that the amplitude is easy to compute when $L$ is small:
\begin{lemma}\label{lemma:low-mode}
    Let $\hat{U}$ be an $M$-mode linear-optical circuit, $|\psi\rangle$ be an $M$-mode general product state, and $|\phi\rangle$ be an arbitrary $L$-mode state. More precisely, all but $L$ modes are vacuum.
    Then,
    \begin{align}
        \langle \phi|\hat{U}|\psi\rangle
    \end{align}
    can be exactly computed in $O(M(Mn_\text{max}+1)^L)$, where $n_\text{max}$ is the maximum photon number for each mode.
    Hence, it can be efficiently computed when $L=O(1)$.
\end{lemma}
\begin{proof}
Let us first examine Fock states.
Consider a few-mode Fock state $|\bm{n}\rangle=|n_1,\dots,n_L,0,\dots,0\rangle$ and arbitrary Fock state $|\bm{m}\rangle$.
Then, the amplitude can be simplified as
\begin{align}
    \langle \bm{n}|\hat{U}|\bm{m}\rangle
    &=\langle \bm{n}| \prod_{i=1}^M\frac{1}{\sqrt{m_i!}}\left(\sum_{j=1}^M U_{ij}\hat{a}_j^\dagger\right)^{m_i}|0\rangle \\
    &=\langle \bm{n}| \prod_{i=1}^M \frac{1}{\sqrt{m_i!}}\left(\sum_{j=1}^L U_{ij}\hat{a}_j^\dagger\right)^{m_i}|0\rangle.
\end{align}
Now, note that the creation operators in the form of $(\hat{a}_1^\dagger)^{n_1}\cdots (\hat{a}_L^\dagger)^{n_L}$ survive and others are zero when contracted with $\langle \bm{n}|$.
Also, such a term renders $\prod_{i=1}^M\sqrt{n_i}$ after contracting with $\langle \bm{n}|$.
This can also be interpreted using the polynomial method by replacing the creation operators by indeterminates $x_1,\dots,x_L$ and introducing the prefactor $\prod_{i=1}^M\sqrt{n_i}$ as
\begin{align}
    \prod_{i=1}^M \frac{\sqrt{n_i!}}{\sqrt{m_i!}}\left(\sum_{j=1}^L U_{ij}x_j\right)^{m_i}.
\end{align}
We then just need to pick up the monomials of $x_j$ after expansion such that their degree matches with $\bm{n}$, i.e., $x_1^{n_1}\cdots x_L^{n_L}$ to compute $\langle \bm{n}|\hat{U}|\bm{m}\rangle$.
In this procedure, the number of monomials we need to take care of is at most $O((MN+1)^L)$, where $N$ is the total photon number.
Also, because we need to keep track of the coefficient for $M$ multiplications, the total cost is $O(M(MN+1)^L)$.

\begin{figure*}[t]
\includegraphics[width=380px]{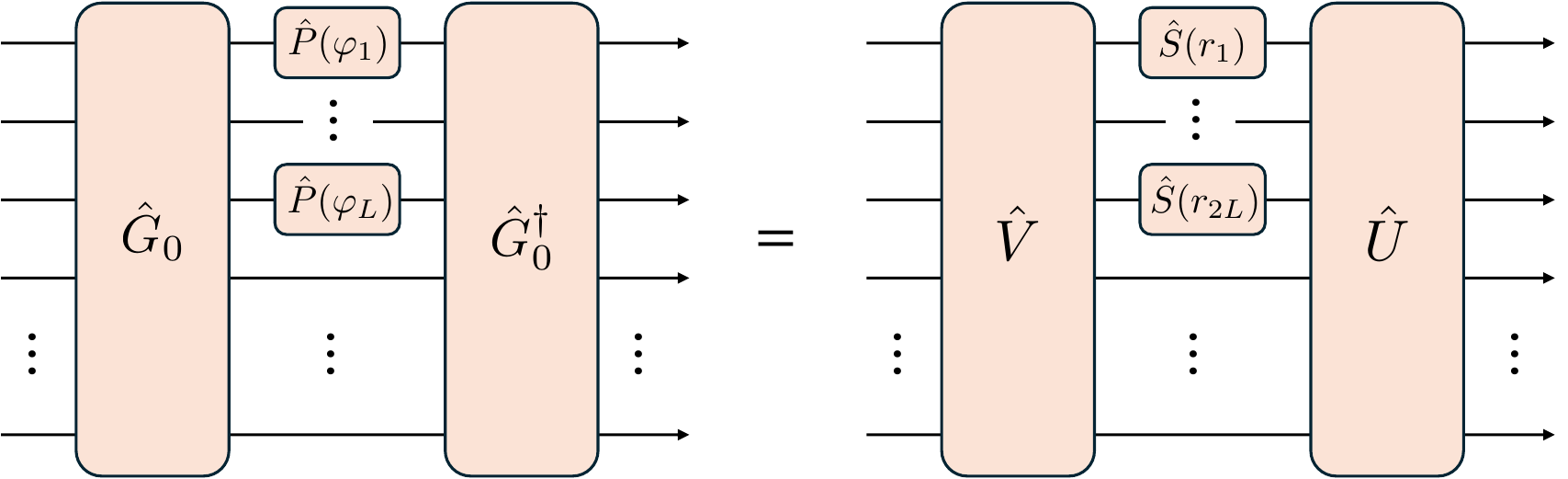} 
\caption{Schematics of Lemma~\ref{lemma:low-rank}. The lemma shows that for a Gaussian unitary circuit $\hat{G}_0$ and a phase shifter $\hat{P}(\bm{\varphi})$ where $\bm{\varphi}\equiv (\varphi_1,\dots,\varphi_L,0,\dots,0)$, $\hat{G}_0^\dagger \hat{P}(\bm{\varphi})\hat{G}_0$ is equivalent to $\hat{U}\hat{S}(\bm{r})\hat{V}$, where $\bm{r}=(r_1,\dots,r_{2L},0,\dots,0)$.}
\label{fig:MVS}
\end{figure*}

Now, let us consider general product states $|\psi\rangle$ and $|\phi\rangle$, which are explicitly written as
\begin{align}\label{eq:psi}
    |\psi\rangle=
    \bigotimes_{i=1}^M\left(\sum_{m_i=0}^{n_\text{max}}a^{(i)}_{m_i}|m_i\rangle\right),~
    |\phi\rangle=
    \bigotimes_{i=1}^M\left(\sum_{n_i=0}^{n_\text{max}}b^{(i)}_{n_i}|n_i\rangle\right).
\end{align}
For a general case, as in the lemma, we can simply use the same method
\begin{align}
    \langle \phi|\hat{U}|\psi\rangle 
    &=\sum_{\bm{n}}b^*_{\bm{n}}\langle \bm{n}|\prod_{i=1}^M \sum_{m_i=0}^{n_\text{max}} \frac{a_{m_i}^{(i)}}{\sqrt{m_i!}}\left(\sum_{j=1}^M U_{ij}\hat{a}^\dagger_j\right)^{m_i}|0\rangle \\
    &=\sum_{\bm{n}}b^*_{\bm{n}}\langle \bm{n}|\prod_{i=1}^M \sum_{m_i=0}^{n_\text{max}} \frac{a_{m_i}^{(i)}}{\sqrt{m_i!}}\left(\sum_{j=1}^L U_{ij}\hat{a}^\dagger_j\right)^{m_i}|0\rangle \\
    &\to \prod_{i=1}^M \sum_{n_i,m_i=0}^{n_\text{max}} \frac{a_{m_i}^{(i)} b_{n_i}^{(i)*}\sqrt{n_i!}}{\sqrt{m_i!}}\left(\sum_{j=1}^L U_{ij}x_j\right)^{m_i},
\end{align}
where $a_{\bm{m}}\equiv \prod_{i=1}^M a^{(i)}_{m_i}$.
For the last arrow, we apply the same method for the Fock state above.
Thus, again, after expanding the polynomial of indeterminates $x_1,\dots,x_L$ and taking the summation over the coefficients of all the monomials, we obtain the amplitude $\langle \phi|\hat{U}|\psi\rangle$.
Since the number of monomials we need to consider is at most $(Mn_\text{max}+1)^L$, the complexity is bounded by $(Mn_\text{max}+1)^L$, where $n_\text{max}$ is the maximum photon number for each mode.
Again, for the update for each multiplication, the total cost becomes $O(M(Mn_\text{max}+1)^L)$.
\end{proof}

\subsubsection{Low-rank structure}
As shown in the permanent case~\cite{ivanov2020complexity}, a low-rank structure of a relevant matrix is crucial to reducing the complexity.
Now, we reveal the low-rank structure of the Gaussian circuit that is obtained by $\hat{G}=\hat{G}_0^\dagger \hat{P}(\bm{\varphi})\hat{G}_0$ when $\bm{\varphi}$ has a few non-zero phases, which is depicted in Fig.~\ref{fig:MVS}.
To be more precise, 
\begin{lemma}[Bloch-Messiah decomposition and low-rank structure]\label{lemma:low-rank}
    Consider a general Gaussian unitary circuit $\hat{G}_0$ and phase-shifter operator $\hat{P}(\bm{\varphi})$ with $\bm{\varphi}=(\varphi_1,\dots,\varphi_L,0,\dots,0)$. Then, $\hat{G}_0^\dagger\hat{P}(\bm{\varphi})\hat{G}_0$ can be decomposed into a linear-optical circuit, at most $2L$ non-zero squeezing operators, and another linear-optical circuit.
\end{lemma}

\begin{proof}[Proof Sketch]
See Appendix~\ref{app:rank} for the details.
We prove this lemma by direct calculation of the transformation of bosonic annihilation operators through the Gaussian unitary circuit $\hat{G}=\hat{G}_0^\dagger \hat{P}(\bm{\varphi})\hat{G}_0$, where $\hat{G}_0$ is an $M$-mode general Gaussian unitary operation, so that using the Bloch-Messiah decomposition it can be written as $\hat{G}_0=\hat{U}\hat{S}(\bm{s})\hat{V}$, where $\hat{U}$ and $\hat{V}$ are linear-optical circuits and $\hat{S}(\bm{s})$ is squeezing operator with squeezing parameter $\bm{s}\in \mathbb{R}^M$.
For this circuit, we can show that 
\begin{align}
    \hat{G}_0\hat{a}_i\hat{G}_0^\dagger
    =\sum_{k=1}^M \left(A_{ik}\hat{a}_k+B_{ik}\hat{a}_k^\dagger\right),
\end{align}
where we defined $M\times M$ matrices
\begin{align}
    A_{ik}\equiv \sum_{j=1}^M V_{ji}^*\cosh s_j U_{kj}^*,~~
    B_{ik}\equiv -\sum_{j=1}^M V_{ji}^*\sinh s_j U_{kj}.
\end{align}
Here, one can see that the rank of the matrix $B$ is equal to the number of non-zero squeezing parameters $s_j$'s, which is crucial for later use.

Using the above transformation, by direct calculation, one may easily show that after the entire process $\hat{G}$, we have
\begin{align}\label{eq:G_decompose}
    \hat{G}^\dagger\hat{a}_i\hat{G}
    &=\hat{a}_i+\sum_{k=1}^M \left(W_{ik}\hat{a}_k+Z_{ik}\hat{a}_k^\dagger\right),
\end{align}
where, if assume that $\varphi_j=0$ for $j>L$, $W$ and $Z$'s matrix elements are written as
\begin{align}
    W_{ik}
    &=\sum_{j=1}^L \left[(e^{i\varphi_j}-1)A_{ij}A_{kj}^*-(e^{-i\varphi_j}-1)B_{ij}B_{kj}^*\right], \\
    Z_{ik}
    &=\sum_{j=1}^L \left[-(e^{i\varphi_j}-1)A_{ij}B_{kj}+(e^{-i\varphi_j}-1)B_{ij}A_{kj}\right].
\end{align}
Hence, the ranks of $W$ and $Z$ are at most $2L$.
Thus, we can write the matrices as
\begin{align}
    W=\sum_{j=1}^{2L}u^{(j)}v^{(j)\T},~~~
    Z=\sum_{j=1}^{2L}w^{(j)}z^{(j)\T}.
\end{align}
In addition, the low-rankness of the matrix $Z$ implies that the Gaussian transform by $\hat{G}=\hat{G}_0^\dagger \hat{P}(\bm{\varphi})\hat{G}_0$ has a structure in which its squeezing parameters are nonzero only for $2L$ modes, which completes the proof.
\end{proof}

\subsubsection{Proof sketch of Lemma~\ref{lemma:MVS}}
In this section, we finally present the sketch of the classical algorithm that computes the generating function for Lemma~\ref{lemma:MVS}.
More specifically, the circuit is a general Gaussian unitary $\hat{G}_0$, and thus the relevant quantity to be computed is written as 
\begin{align}\label{eq:G}
    \langle \phi|\hat{G}_0^\dagger \hat{P}(\bm{\varphi})\hat{G}_0|\psi\rangle
    \equiv \langle \phi|\hat{G}|\psi\rangle
    =\sum_{\bm{n},\bm{m}}a_{\bm{m}} b_{\bm{n}}^*\langle\bm{n}|\hat{G}|\bm{m}\rangle,
\end{align}
where $|\psi\rangle$ and $|\phi\rangle$ are arbitrary product states, which are written as Eq.~\eqref{eq:psi}.
Also, the phase vector is assumed to be written as $\bm{\varphi}=(\varphi_1,\dots,\varphi_L,0,\dots,0)$, without loss of generality.
By using Lemma~\ref{lemma:low-rank}, the Bloch-Messiah decomposition of the Gaussian unitary $\hat{G}$ is written as $\hat{U}\hat{S}(\bm{r})\hat{V}$, where $\hat{U}$ and $\hat{V}$ are linear-optical unitary circuits, and $\bm{r}$ is the squeezing parameter vector with at most $2L$ nonzero elements.

Let us focus on a simpler case: $\langle\bm{n}|\hat{G}|\bm{m}\rangle$.
Using the low-rank structure, it can be rewritten as
\begin{align}\label{eq:G2}
    \langle\bm{n}|\hat{G}|\bm{m}\rangle &=\frac{1}{\sqrt{\bm{n}!}}\langle\bm{r}|\hat{V} \left[\prod_{i\in [M]}\left(\hat{A}_i+\hat{B}_i^\dagger\right)^{n_i}\right]|\bm{m}\rangle,
\end{align}
where we defined
\begin{align}
    &\hat{A}_i=\hat{a}_i+\sum_{k=1}^MW_{ik}\hat{a}_k,~~~\hat{B}_i^\dagger=\sum_{k=1}^M Z_{ik}\hat{a}_k^\dagger, \\
    &C_{ij}=[\hat{A}_i,\hat{B}_j^\dagger]=Z_{ji}+\sum_{k=1}^M W_{ik}Z_{jk}.
\end{align}
Note that by Lemma~\ref{lemma:low-rank}, the rank of the matrix $C$ is at most $2L$ since $W$ and $Z$ have at most $2L$ ranks and the number of nonzero squeezing parameters is limited to at most $2L$ as shown in Lemma~\ref{lemma:low-rank}.

By using the antinormal ordering~\cite{varvak2005rook} (See Appendix~\ref{app:anti}), we expand Eq.~\eqref{eq:G2} and move all the creation operators to the right, 
\begin{align}\label{eq:G3}
    &\frac{1}{\sqrt{\bm{n}!\bm{m}!}}\sum_{Y\subset X_{\bm{n}}}\sum_{\mathcal{M}\in \text{Match}(F_{X_{\bm{n}}};Y\to Y^c)}(-1)^{|\mathcal{M}|}\left(\prod_{(i,j)\in \mathcal{M}}C_{ij}\right) \nonumber \\ 
    &\times \langle \bm{r}|\hat{V}\left(\prod_{i\in Y\setminus A(\mathcal{M})}\hat{A}_i\right)\left(\prod_{i\in Y^c\setminus B(\mathcal{M})}\hat{B}^\dagger_i\right)\prod_{i\in X_{\bm{m}}}\hat{a}_i^\dagger|0\rangle,
\end{align}
where $X_{\bm{n}}$ is the multiset composed of the elements of~$\bm{n}$.
For example, if $\bm{n}=(0,2,1)$, $X_{\bm{n}}=\{2,2,3\}$.
$\text{Match}(F_{X_{\bm{n}}};Y\to Y^c)$ is the match of the Ferrers graph~\cite{stanley2011enumerative} with vertices $Y$ and $Y^c$.
More explicitly, the Ferrers graph's edges are defined as $E=\{(i,j)|i,j\in X_{\bm{n}},i<j\}$.
$A(\mathcal{M})$ and $B(\mathcal{M})$ represent the elements in $Y$ and $Y^c$ that are involved in the matching $\mathcal{M}$, respectively.

We then show that the above expression can be computed by using the polynomial method introduced in Sec.~\ref{sec:poly}.
More precisely, we show that Eq.~\eqref{eq:G3} can be written as a loop hafnian of a matrix $K$ which has a low-rank structure.
After that, since the matrix $C$ has a low rank, which turns out to be relevant to the low-rank structure of the matrix $K$, we invoke the following lemma, which states that a loop hafnian of a matrix with a low-rank structure is easy to compute~\cite{oh2024quantum}:
\begin{lemma}
    [Loop hafnian of a matrix with low-rank structure~\cite{bjorklund2019faster, oh2024quantum}] 
    Let $\Sigma$ be an $n\times n$ symmetric matrix with rank $r$ and $A$ be a matrix constructed by filling $\Sigma$'s diagonal matrix with arbitrary elements.
    Then the loop hafnian of the matrix $A$ can be computed in $O(n\binom{2n+r-1}{r-1})$. Thus, when $r$ is fixed, the cost is polynomial in $n$ and when $r=O(\log n)$, the cost is quasipolynomial in $n$.
\end{lemma}
Another observation to make this calculation easier is that the number of non-zero squeezing parameters is at most $2L$, so that the term $\langle \bm{r}|\cdots |0\rangle$ in Eq.~\eqref{eq:G3} can be efficiently computed using Lemma~\ref{lemma:low-mode}.
In short, we write Eq.~\eqref{eq:G3} as a polynomial of indeterminates, the number of which is at most $(Mn_\text{max}+1)^{16L}$, which makes it easy to compute when $L=O(1)$.
We note that while the degree of the polynomial is quite large at this moment, we expect that it may be reduced, which we leave open.
Also, if we restrict our attention to linear-optical circuits, the degree of the polynomial reduces to $(Mn_\text{max}+1)^{2L}$ (see Appendix~\ref{app:gurvits_linear}).

\section{Conclusion}\label{sec:conclusion}
In this work, we presented classical algorithms for solving the quantum mean-value problem in Gaussian circuits, as well as in Gaussian circuits augmented by (i) photon-number-resolving measurements with Gaussian feedforward and (ii) Gaussian measurements with Gaussian feedforward.
To support these results, we also developed algorithms for computing amplitudes of low-mode states and marginal probabilities of number-basis measurements, which may be of independent interest and connect to Gurvits' second algorithm.

Our results reveal several conceptual insights.
First, low-mode and low-rank structures play a central role in making certain bosonic computations classically tractable.
Second, we highlight an important contrast between sampling problems and the quantum mean-value problems: while non-Gaussian resources typically induce classical hardness in sampling, mean-value estimation can remain efficiently simulatable even in the presence of highly non-Gaussian inputs, provided the number of adaptive measurements is bounded.
This distinction illustrates a richer landscape of classical–quantum complexity for bosonic systems.

Taken together, these findings not only provide concrete classical algorithms but also clarify where the boundaries of quantum advantage may lie in Gaussian and non-Gaussian regimes. We anticipate that these tools and perspectives will aid future studies of bosonic circuits, both as a benchmark for quantum experiments and as a stepping stone toward understanding the computational power of non-Gaussian resources.

On the other hand, since our results emphasize the role of adaptive measurements in the quantum mean-value problem, it is natural to ask what role adaptive measurements play in the corresponding sampling tasks. In our setting, increasing the amount of adaptivity leads to an easy-to-hard transition in the complexity of the mean-value problem, whereas in the non-adaptive regime sampling is already believed to be hard due to boson sampling. It is therefore unclear whether, and in what sense, additional rounds of adaptive (or non-Gaussian) operations change the complexity of the sampling problem. There are related studies that add non-linear operations to boson sampling~\cite{spagnolo2023non}, but neither our results nor these works currently provide a satisfactory complexity-theoretic understanding of how sampling interpolates between non-universal boson sampling and universal photonic quantum computation.

While we generalized the classical algorithm for quantum mean-value problem in Ref.~\cite{lim2025efficient}, the latter provides another classical algorithm that computes the transition amplitude of a quantum state under linear-optical circuits.
We believe that it is an interesting open question to extend this algorithm to the adaptive setting due to its close relation to various applications of bosonic circuits (e.g., quantum kernel methods~\cite{chabaud2021quantum, hoch2025quantum}).

Finally, we note that our generalized Gurvits' second algorithm takes polynomial time in the number of modes $M$ when the number of measurements for feedforward is at most constant, i.e., $L=O(1)$, and the degree of the polynomial is not high for linear-optical circuits, as discussed above.
Hence, our algorithm is expected to be practically efficient to simulate recent experiments, e.g., the one in Ref.~\cite{hoch2025quantum}.
On the other hand, while our generalized Gurvits' second algorithm for general Gaussian circuits also takes polynomial time when $L=O(1)$, the degree of the polynomial is quite high, which may make it less practical.
Hence, improving the practical running time is an important open question.
More importantly, whether our algorithm can be improved further to be efficient even when $L$ is logarithmic in the system size is another open question.

\begin{acknowledgements}
C.O. and Y.L. were supported by the National
Research Foundation of Korea Grants (No. RS-2023-NR119931, No. RS-2024-00431768 and No. RS-2025-00515456) funded by the Korean government (Ministry of Science and ICT~(MSIT)) and the Institute of Information \& Communications Technology Planning \& Evaluation (IITP) Grants funded by the Korea government (MSIT) (No. RS-2024-00437284, No. IITP-2025-RS-2025-02283189 and No. IITP-2025-RS-2025-02263264).
\end{acknowledgements}

\bibliography{reference.bib}

\onecolumngrid

\appendix

\section{Quantum mean-value problem proofs}\label{app:mean-value}
In this Appendix, we provide detailed proofs for classical algorithms for quantum mean-value problems.

\subsection{Quantum mean-value problem in Gaussian circuits without feedforward}\label{app:gaussian}
We provide the full proof of Theorem~\ref{thm:gaussian}, which is restated here:
\begin{theorem}[Theorem~\ref{thm:gaussian} restated]
    Consider an $M$-mode Gaussian circuit $\hat{G}$, an arbitrary product input state $|\psi\rangle$, and an $M$-mode product operator $\hat{O}=\hat{O}_A\otimes \hat{\mathbb{1}}_B$.  
    The expectation value $\langle \psi|\hat{G}^\dagger \hat{O}\hat{G}|\psi\rangle$ can be approximated within additive error $\epsilon$ with probability $1-\delta$ in running time
    \begin{align}
        O\left(\frac{M^2\|\hat{O}_A\|_2^2\Tr[\hat{\rho}_A^2]\log(1/\delta)}{\epsilon^2}\right),
    \end{align}
    where $\hat{\rho}_A\equiv \Tr_B[\hat{G}|\psi\rangle\langle \psi|\hat{G}^\dagger]$ is the reduced density matrix on the system $A$.
\end{theorem}
\begin{proof}
Let us consider a general Gaussian unitary circuit $\hat{G}$ and a general product input state $|\psi\rangle$.
Then, the expectation value of a product observable $\hat{O}=\hat{O}_A\otimes\hat{\mathbb{1}}_B$ can be rewritten as
\begin{align}
    \langle \psi|\hat{G}^\dagger (\hat{O}_A\otimes \hat{\mathbb{1}}_B) \hat{G}|\psi\rangle 
    &=\frac{1}{\pi^l}\int d^{2l}\bm{\alpha}_A\Tr[\hat{O}_A \hat{D}(\bm{\alpha}_A)]\Tr[(\hat{D}^\dagger(\bm{\alpha}_A)\otimes \hat{\mathbb{1}}_B)\hat{G}|\psi\rangle\langle \psi|\hat{G}^\dagger] \\ 
    &=\frac{1}{\pi^l}\int d^{2l}\bm{\alpha}_A\Tr[\hat{O}_A \hat{D}(\bm{\alpha}_A)]\Tr[\hat{D}^\dagger(\bm{\alpha}_A)\hat{\rho}_A] \\ 
    &=\frac{1}{\pi^l}\int d^{2l}\bm{\alpha}_A\chi_{\hat{O}_A}(\bm{\alpha}_A)\chi_{\hat{\rho}_A}^*(\bm{\alpha}_A) \\ 
    &=\int d^{2l}\bm{\alpha}_A \frac{|\chi_{\hat{O}_A}(\bm{\alpha}_A)|^2}{\pi^l\|\hat{O}_A\|_2^2}\frac{\chi_{\hat{\rho}_A}^*(\bm{\alpha}_A)\|\hat{O}_A\|_2^2}{\chi_{\hat{O}_A}^*(\bm{\alpha})_A} \\ 
    &\equiv \int d^{2l}\bm{\alpha}_A p(\bm{\alpha}_A)X(\bm{\alpha}_A),
\end{align}
where we defined the reduced density matrix, a probability distribution, and a random variable
\begin{align}
    \hat{\rho}_A\equiv \Tr_B[\hat{G}|\psi\rangle\langle \psi|\hat{G}^\dagger],~~~
    p(\bm{\alpha}_A)=\frac{\left|\chi_{\hat{O}_A}(\bm{\alpha}_A)\right|^2}{\pi^l\|\hat{O}_A\|_2^2},~~~
    X(\bm{\alpha}_A)=\frac{\chi^*_{\hat{\rho}_A}(\bm{\alpha}_A) \|\hat{O}_A\|_2^2}{\chi^*_{\hat{O}_A}(\bm{\alpha}_A)}.
\end{align}

Here, using the property of the displacement operator, we have
\begin{align}
    \frac{1}{\pi^{l}}\int d^{2l}\bm{\alpha}_A|\chi_{\hat{O}_A}(\bm{\alpha}_A)|^2
    =\Tr[\hat{O}_A^\dagger\hat{O}_A]
    =\|\hat{O}_A\|_2^2,
\end{align}
which shows that $p(\bm{\alpha}_A)$ is a proper probability distribution.
Thus, the above expression allows us to use the median-of-means estimator for estimating the mean value.

We now analyze the complexity.
First of all, sampling $\bm{\alpha}_A$ from $p(\bm{\alpha}_A)$ is easy due to the product structure of $\hat{O}_A$:
\begin{align}\label{eq:chi_O}
    p(\bm{\alpha}_A)
    =\frac{\left|\chi_{\hat{O}_A}(\bm{\alpha}_A)\right|^2}{\pi^l\|\hat{O}_A\|_2^2}
    =\prod_{i\in A} \frac{\left|\chi_{\hat{O}_i}(\alpha_i)\right|^2}{\pi\|\hat{O}_i\|_2^2},
\end{align}
which takes $O(l)$ time to compute.
Also, the random variable $X(\bm{\alpha}_A)$ is easy to compute because $\|\hat{O}_A\|_2^2$ is easy to compute due to the product structure and $\chi_{\hat{O}_A}^*(\bm{\alpha}_A)$ is easy to compute as shown above and 
\begin{align}
    \chi^*_{\hat{\rho}_A}(\bm{\alpha}_A)
    =\Tr[\hat{D}^\dagger(\bm{\alpha}_A)\hat{\rho}_A]
    =\Tr[(\hat{D}^\dagger(\bm{\alpha}_A)\otimes \hat{\mathbb{1}}_B)\hat{G}|\psi\rangle\langle \psi|\hat{G}^\dagger]
    =\langle \psi|\hat{D}^\dagger(\bm{\alpha}')|\psi\rangle
    =\prod_{i=1}^M \langle \psi_i|\hat{D}^\dagger(\alpha'_i)|\psi_i\rangle,
\end{align}
where we used the property of displacement operators and Gaussian unitary circuits,
\begin{align}
    \hat{D}^\dagger(\bm{\alpha}')
    =\hat{G}^\dagger (\hat{D}^\dagger(\bm{\alpha}_A)\otimes \hat{\mathbb{1}}_B)\hat{G}.
\end{align}
Here, $\bm{\alpha}'$ has a linear relation with $\bm{\alpha}_A$ so that computing $\bm{\alpha}'$ takes at most $O(lM)$ time.
Therefore, the random variable can be rewritten as
\begin{align}
    X(\bm{\alpha_A})
    &=\frac{\chi^*_{\hat{\rho}_A}(\bm{\alpha}_A) \|\hat{O}_A\|_2^2}{\chi^*_{\hat{O}_A}(\bm{\alpha}_A)}
    =\chi^*_{\hat{\rho}_A}(\bm{\alpha}_A)\prod_{i\in A} \frac{\|\hat{O}_i\|_2^2}{\chi^*_{\hat{O}_i}(\alpha_i)}  
    =\Tr[\hat{G}|\psi\rangle\langle\psi|\hat{G}^\dagger \hat{D}(\bm{\alpha}_A)]^*\prod_{i\in A} \frac{\|\hat{O}_i\|_2^2}{\chi^*_{\hat{O}_i}(\alpha_i)}  \\
    &=\Tr[|\psi\rangle\langle\psi|\hat{D}(\bm{\alpha}_A')]^*\prod_{i\in A} \frac{\|\hat{O}_i\|_2^2}{\chi^*_{\hat{O}_i}(\alpha_i)} 
    =\prod_{i=1}^M \chi^*_{\psi_i}(\alpha'_i)\prod_{i\in A} \frac{\|\hat{O}_i\|_2^2}{\chi^*_{\hat{O}_i}(\alpha_i)},
\end{align}
and computing this takes at most $O(l+M)$ cost.

Finally, we find the upper bound on the variance of the random variable:
\begin{align}
    \text{Var}(X)&\leq \int d^{2l}\bm{\alpha}_A p(\bm{\alpha}_A)|X(\bm{\alpha}_A)|^2
    =\int d^{2l}\bm{\alpha}_A \frac{|\chi_{\hat{O}_A}(\bm{\alpha}_A)|^2}{\pi^l\|\hat{O}_A\|_2^2}\frac{|\chi_{\psi}(\bm{\alpha}')|^2\|\hat{O}_A\|_2^4}{|\chi_{\hat{O}_A}(\bm{\alpha}_A)|^2}
    =\|\hat{O}_A\|_2^2\int d^{2l}\bm{\alpha}_A\frac{|\chi_{\psi}(\bm{\alpha}')|^2}{\pi^l} \\
    &=\|\hat{O}_A\|_2^2\|\hat{\rho}_A\|_2^2,
\end{align}
where, for the last equality, we used
\begin{align}
    \int d^{2l}\bm{\alpha}_A\frac{|\chi_{\psi}(\bm{\alpha}')|^2}{\pi^l}
    =\int d^{2l}\bm{\alpha}_A\frac{|\langle \psi|\hat{G}^\dagger(\hat{D}^\dagger(\bm{\alpha}_A)\otimes \hat{\mathbb{1}}_B) \hat{G}|\psi\rangle|^2}{\pi^l}
    =\int d^{2l}\bm{\alpha}_A\frac{|\chi_{\hat{\rho}_A}(\bm{\alpha}_A)|^2}{\pi^l}
    =\|\hat{\rho}_A\|_2^2.
\end{align}
Hence, the median-of-means estimator can estimate the mean value with the sampling complexity $O(\|\hat{O}_A\|_2^2\|\hat{\rho}_A\|_2^2\log(1/\delta)/\epsilon^2)$.
Thus, combining all the complexities, the total running time is given as $O(Ml\|\hat{O}_A\|_2^2\|\hat{\rho}_A\|_2^2\log(1/\delta)/\epsilon^2)$.
\end{proof}

\subsection{Quantum mean-value problem in Gaussian circuits with photon-number measurements and Gaussian feedforward}\label{app:gaussian_photon}
We provide the full proof of Theorem~\ref{thm:gaussian_photon}, which is restated here:
\begin{theorem}[Theorem~\ref{thm:gaussian_photon} restated]
    Consider an arbitrary product input state and a Gaussian circuit, augmented by photon-number measurements and Gaussian feedforward, which outputs an $M$-mode quantum state.
    If the number of photon-number measurements used for feedforward is $L=O(1)$, then the expectation value of an $M$-mode observable $\hat{O}=\hat{O}_A\otimes \hat{\mathbb{1}}_B$ can be approximated within additive error $\epsilon$ and success probability $1-\delta$ in running time
    \begin{align}
        O\left(\frac{\poly(M)\|\hat{O}_A\|_2^2 \mathbb{E}_{\bm{n}}[\Tr(\hat{\rho}_{\bm{n}}^2)] \log(1/\delta)}{\epsilon^2}\right),    
    \end{align}
    where $\hat{\rho}_{\bm{n}}$ is the reduced density matrix on the system $A$ when the photon-number measurement outcome is $\bm{n}$ and $\mathbb{E}_{\bm{n}}[\Tr(\hat{\rho}_{\bm{n}}^2)]$ is the average purity.
\end{theorem}
\begin{proof}
Consider a general product input $|\psi\rangle$.
First of all, when the measurement outcomes are $\bm{n}$, the corresponding post-measurement output state is (up to normalization) (see Sec.~\ref{sec:adaptive})
\begin{align}
    \langle \bm{n}|\hat{G}_{\bm{n}}|\psi\rangle,
\end{align}
with probability
\begin{align}
    p(\bm{n})
    =\langle \psi|\hat{G}_{\bm{n}}^\dagger(\hat{\mathbb{1}}_M\otimes|\bm{n}\rangle \langle \bm{n}|)\hat{G}_{\bm{n}}|\psi\rangle
    =\langle \psi_{\bm{n}}|(\hat{\mathbb{1}}_M\otimes|\bm{n}\rangle \langle \bm{n}|)|\psi_{\bm{n}}\rangle.
\end{align}
Thus, the expectation value when the outcome is $\bm{n}$ is
\begin{align}
    \frac{\langle \psi_{\bm{n}}|(\hat{O}_A\otimes \hat{\mathbb{1}}_B\otimes|\bm{n}\rangle \langle \bm{n}|)|\psi_{\bm{n}}\rangle}{\langle \psi_{\bm{n}}|(\hat{\mathbb{1}}_{A}\otimes\hat{\mathbb{1}}_{B}\otimes|\bm{n}\rangle \langle \bm{n}|)|\psi_{\bm{n}}\rangle}.
\end{align}

Since the measurement $\bm{n}$ is obtained with probability $p(\bm{n})$, we take the average over all possible outcomes in the middle:
\begin{align}
    \langle\hat{O}\rangle 
    &=\sum_{\bm{n}}\langle \psi_{\bm{n}}|(\hat{O}_A\otimes\hat{\mathbb{1}}_B\otimes |\bm{n}\rangle\langle \bm{n}|)|\psi_{\bm{n}}\rangle \\
    &=\frac{1}{\pi^l}\int d^{2l}\bm{\alpha}\sum_{\bm{n}}\langle \psi_{\bm{n}}|(\hat{D}^\dagger(\bm{\alpha}_A)\otimes\hat{\mathbb{1}}_B\otimes |\bm{n}\rangle\langle \bm{n}|)|\psi_{\bm{n}}\rangle \chi_{\hat{O}_A}(\bm{\alpha}_A) \\ 
    &=\int d^{2l}\bm{\alpha}_A\sum_{\bm{n}}p(\bm{n})q(\bm{\alpha}_A)\frac{\|\hat{O}_A\|_2^2\langle \psi_{\bm{n}}|(\hat{D}^\dagger(\bm{\alpha}_A)\otimes\hat{\mathbb{1}}_B\otimes |\bm{n}\rangle\langle \bm{n}|)|\psi_{\bm{n}}\rangle}{p(\bm{n})\chi^*_{\hat{O}_A}(\bm{\alpha}_A)} \\ 
    &=\int d^{2l}\bm{\alpha}_A\sum_{\bm{n}}p(\bm{n})q(\bm{\alpha}_A)X(\bm{n},\bm{\alpha}_A),
\end{align}
where we defined the probability distribution and the random variable
\begin{align}
    q(\bm{\alpha}_A)\equiv \frac{\left|\chi_{\hat{O}_A}(\bm{\alpha}_A)\right|^2}{\pi^l\|\hat{O}_A\|_2^2},\quad
    X(\bm{n},\bm{\alpha}_A)\equiv \frac{\|\hat{O}_A\|_2^2\langle \psi_{\bm{n}}|(\hat{D}^\dagger(\bm{\alpha}_A)\otimes\hat{\mathbb{1}}_B\otimes |\bm{n}\rangle\langle \bm{n}|)|\psi_{\bm{n}}\rangle}{p(\bm{n})\chi^*_{\hat{O}_A}(\bm{\alpha}_A)}.
\end{align}

Now, we analyze the complexity.
First, to upper-bound the error of the estimator, we compute the second moment:
\begin{align}
    &\int d^{2l}\bm{\alpha}_A\sum_{\bm{n}}p(\bm{n})q(\bm{\alpha}_A)\frac{\|\hat{O}_A\|_2^4|\langle \psi_{\bm{n}}|(\hat{D}^\dagger(\bm{\alpha}_A)\otimes \hat{\mathbb{1}}_B\otimes |\bm{n}\rangle\langle \bm{n}|)|\psi_{\bm{n}}\rangle|^2}{p(\bm{n})^2\left|\chi_{\hat{O}_A}(\bm{\alpha}_A)\right|^2} \\ 
    &=\frac{\|\hat{O}_A\|_2^2}{\pi^l}\int d^{2l}\bm{\alpha}_A\sum_{\bm{n}}\frac{|\langle \psi_{\bm{n}}|(\hat{D}^\dagger(\bm{\alpha}_A)\otimes\hat{\mathbb{1}}_B\otimes |\bm{n}\rangle\langle \bm{n}|)|\psi_{\bm{n}}\rangle|^2}{p(\bm{n})} \\ 
    &=\frac{\|\hat{O}_A\|_2^2}{\pi^l}\int d^{2l}\bm{\alpha}_A\sum_{\bm{n}}p(\bm{n})|\chi_{\hat{\rho}_{\bm{n}}}(\bm{\alpha}_A)|^2 \\ 
    &=\|\hat{O}_A\|_2^2\sum_{\bm{n}}p(\bm{n}) \Tr[\hat{\rho}_{\bm{n}}^2],
\end{align}
where $\hat{\rho}_{\bm{n}}$ is the reduced density matrix on the system $A$ when the measurement outcome is $\bm{n}$.
Thus, the variance is upper-bounded by the Hilbert-Schmidt norm of the operator $\hat{O}_A$ multiplied by the average purity of the reduced state on the observable's Hilbert space.
Hence, the sample complexity for the median-of-means estimator is given as $O(\|\hat{O}_A\|_2^2\mathbb{E}_{\bm{n}}[\text{Tr}[\hat{\rho}_{\bm{n}}^2]]\log(1/\delta)/\epsilon^2)$.

Meanwhile, we need to be able to sample from $p(\bm{\beta})$ and $q(\bm{\alpha}_A)$.
Sampling from $q(\bm{\alpha}_A)$ can be efficiently done because of the product structure:
\begin{align}
    q(\bm{\alpha}_A)
    =\frac{\left|\chi_{\hat{O}_A}(\bm{\alpha}_A)\right|^2}{\pi^l\|\hat{O}_A\|_2^2}
    =\prod_{i\in A}\frac{|\chi_{\hat{O}_i}(\alpha_i)|^2}{\pi \|\hat{O}_i\|_2^2},
\end{align}
which takes at most $O(l)$.
Also, $p(\bm{n})$ can be easily sampled if the number of modes in $|\bm{n}\rangle$ is not too large by Lemma~\ref{lemma:main}.
More specifically, when the number of modes in $|\bm{n}\rangle$ is $L=O(1)$, sampling from $p(\bm{n})$ takes at most polynomial in $M$.

Also, in computing the random variable $X(\bm{n},\bm{\alpha}_A)$, $\|\hat{O}_A\|_2^2=\prod_{i\in A}\|\hat{O}_i\|_2^2$ is easy to compute, and $p(\bm{n})$ is the same as the above.
Also, $\chi_{\hat{O}_A}(\bm{\alpha}_A)$ is easy to compute due to the product structure of $\hat{O}_A$~(see Eq.~\eqref{eq:chi_O}).
Finally,
\begin{align}
    \langle \psi_{\bm{n}}|(\hat{D}^\dagger(\bm{\alpha}_A)\otimes \hat{\mathbb{1}}_B\otimes |\bm{n}\rangle\langle \bm{n}|)|\psi_{\bm{n}}\rangle
    &=\langle \psi|\hat{G}^\dagger_{\bm{n}}(\hat{D}^\dagger(\bm{\alpha}_A)\otimes \hat{\mathbb{1}}_B\otimes|\bm{n}\rangle\langle \bm{n}|)\hat{G}_{\bm{n}}|\psi\rangle
    =\langle \psi|\hat{D}^\dagger(\bm{\alpha}')\hat{G}^\dagger_{\bm{n}}(\hat{\mathbb{1}}_A\otimes \hat{\mathbb{1}}_B\otimes|\bm{n}\rangle\langle \bm{n}|)\hat{G}_{\bm{n}}|\psi\rangle \\
    &=\langle \phi|\hat{G}^\dagger_{\bm{n}}(\hat{\mathbb{1}}_A\otimes \hat{\mathbb{1}}_B\otimes|\bm{n}\rangle\langle \bm{n}|)\hat{G}_{\bm{n}}|\psi\rangle
\end{align}
where we used
\begin{align}
    \hat{G}^\dagger_{\bm{n}}\hat{D}^\dagger(\bm{\alpha}_A)\hat{G}_{\bm{n}}=\hat{D}^\dagger(\bm{\alpha}').
\end{align}
Here, computing $\bm{\alpha}'$ takes at most $O(Ml)$ and again by using Lemma~\ref{lemma:main}, we can compute this quantity in polynomial time when $L=O(1)$.
Therefore, the total running time is at most
\begin{align}
    O\left(\frac{\poly(M)\|\hat{O}_A\|_2^2 \mathbb{E}_{\bm{n}}[\Tr(\hat{\rho}_{\bm{n}}^2)] \log(1/\delta)}{\epsilon^2}\right),
\end{align}
which completes the proof.

\end{proof}

\subsection{Quantum mean-value problem in Gaussian circuit augmented by Gaussian measurement and Gaussian feedforward}\label{app:gaussian_gaussian}
We provide the full proof of Theorem~\ref{thm:gaussian_gaussian}, which is restated here:
\begin{theorem}[Theorem~\ref{thm:gaussian_gaussian} restated]
    Consider an arbitrary product input and a Gaussian circuit, augmented by Gaussian measurements and Gaussian feedforward, which outputs an $M$-mode quantum state.
    If the number of Gaussian measurements for feedforward is at most constant, i.e., $L=O(1)$, the expectation value of an $M$-mode product observable $\hat{O}\equiv \hat{O}_A\otimes \hat{\mathbb{1}}_B$ can be approximated within additive error $\epsilon$ with probability $1-\delta$ in running time,
    \begin{align}
        O\left(\frac{\poly(M)\|\hat{O}_A\|_2^2 \mathbb{E}_{\bm{\beta}}[\Tr(\hat{\rho}_{\bm{\beta}}^2)] \log(1/\delta)}{\epsilon^2}\right),    
    \end{align}
    where $\hat{\rho}_{\bm{\beta}}$ is the reduced density matrix on the system $A$ when the Gaussian measurement outcome is $\bm{\beta}$ and $\mathbb{E}_{\bm{\beta}}[\Tr(\hat{\rho}_{\bm{\beta}}^2)]$ is the average purity.
\end{theorem}

\begin{proof}
Consider a general product input $|\psi\rangle$ and Gaussian measurement POVMs:
\begin{align}
    \frac{1}{\pi^L}\hat{D}(\bm{\beta})|\psi_G\rangle\langle \psi_G|\hat{D}^\dagger(\bm{\beta}).
\end{align}
First of all, when the measurement outcomes are $\bm{\beta}$, the probability of obtaining $\bm{\beta}$ is 
\begin{align}
    p(\bm{\beta})
    =\frac{1}{\pi^L}\langle \psi|\hat{G}_{\bm{\beta}}^\dagger(\hat{\mathbb{1}}_M\otimes\hat{D}(\bm{\beta})|\psi_G\rangle\langle \psi_G|\hat{D}^\dagger(\bm{\beta}))\hat{G}_{\bm{\beta}}|\psi\rangle
    =\frac{1}{\pi^L}\langle \psi_{\bm{\beta}}|(\hat{\mathbb{1}}_M\otimes\hat{D}(\bm{\beta})|\psi_G\rangle\langle \psi_G|\hat{D}^\dagger(\bm{\beta}))|\psi_{\bm{\beta}}\rangle,
\end{align}
and the (unnormalized) post-measurement state is
\begin{align}
    \langle \psi_G|\hat{D}^\dagger(\bm{\beta})|\psi_{\bm{\beta}}\rangle.
\end{align}
Thus, the expectation value when the outcome is $\bm{\beta}$ is
\begin{align}
    \frac{\langle \psi_{\bm{\beta}}|(\hat{O}_A\otimes\hat{\mathbb{1}}_B\otimes\hat{\Pi}(\bm{\beta}))|\psi_{\bm{\beta}}\rangle}{\langle \psi_{\bm{\beta}}|(\hat{\mathbb{1}}_A\otimes \hat{\mathbb{1}}_B \otimes\hat{\Pi}(\bm{\beta}))|\psi_{\bm{\beta}}\rangle}.
\end{align}
Taking the average over all outcomes $\bm{\beta}$ and defining $|\psi_{\bm{\beta}}\rangle=\hat{G}_{\bm{\beta}}|\psi\rangle$, the expectation value is written as
\begin{align}
    \langle\hat{O}\rangle 
    &=\frac{1}{\pi^L}\int d^{2L}\bm{\beta}\langle \psi_{\bm{\beta}}|(\hat{O}_A\otimes\hat{\mathbb{1}}_B\otimes \hat{D}(\bm{\beta})|\psi_G\rangle\langle \psi_G|\hat{D}^\dagger(\bm{\beta}))|\psi_{\bm{\beta}}\rangle \\
    &=\frac{1}{\pi^L}\int \frac{d^{2l}\bm{\alpha}_A}{\pi^l}d^{2L}\bm{\beta}\langle \psi_{\bm{\beta}}|(\hat{D}^\dagger(\bm{\alpha}_A)\otimes\hat{\mathbb{1}}_B\otimes \hat{D}(\bm{\beta})|\psi_G\rangle\langle \psi_G|\hat{D}^\dagger(\bm{\beta}))|\psi_{\bm{\beta}}\rangle \chi_{\hat{O}_A}(\bm{\alpha}_A) \\ 
    &=\frac{1}{\pi^L}\int d^{2l}\bm{\alpha}_Ad^{2L}\bm{\beta}\frac{\|\hat{O}_A\|_2^2\langle \psi_{\bm{\beta}}|(\hat{D}^\dagger(\bm{\alpha}_A)\otimes\hat{\mathbb{1}}_B\otimes \hat{D}(\bm{\beta})|\psi_G\rangle\langle \psi_G|\hat{D}^\dagger(\bm{\beta}))|\psi_{\bm{\beta}}\rangle}{\chi^*_{\hat{O}_A}(\bm{\alpha}_A)} \frac{\left|\chi_{\hat{O}_A}(\bm{\alpha}_A)\right|^2}{\pi^l\|\hat{O}_A\|_2^2} \\ 
    &=\int d^{2l}\bm{\alpha}_Ad^{2L}\bm{\beta}p(\bm{\beta})\frac{\|\hat{O}\|_2^2\langle \psi_{\bm{\beta}}|(\hat{D}^\dagger(\bm{\alpha}_A)\otimes\hat{\mathbb{1}}_B\otimes \hat{D}(\bm{\beta})|\psi_G\rangle\langle \psi_G|\hat{D}^\dagger(\bm{\beta}))|\psi_{\bm{\beta}}\rangle}{\pi^L p(\bm{\beta})\chi^*_{\hat{O}_A}(\bm{\alpha}_A)} q(\bm{\alpha}_A) \\ 
    &=\int d^{2l}\bm{\alpha}_Ad^{2L}\bm{\beta}p(\bm{\beta})q(\bm{\alpha}_A)X(\bm{\beta},\bm{\alpha}_A),
\end{align}
where we defined a probability distribution and a random variable
\begin{align}
    q(\bm{\alpha}_A)\equiv \frac{\left|\chi_{\hat{O}_A}(\bm{\alpha}_A)\right|^2}{\pi^l\|\hat{O}_A\|_2^2},~~~
    X(\bm{\beta},\bm{\alpha})
    =\frac{\|\hat{O}_A\|_2^2\langle \psi_{\bm{\beta}}|(\hat{D}^\dagger(\bm{\alpha}_A)\otimes \hat{\mathbb{1}}_B\otimes \hat{D}(\bm{\beta})|\psi_G\rangle\langle \psi_G|\hat{D}^\dagger(\bm{\beta}))|\psi_{\bm{\beta}}\rangle}{\pi^L p(\bm{\beta})\chi_{\hat{O}_A}^*(\bm{\alpha}_A)}.
\end{align}
Hence, the quantum mean value can be written as
\begin{align}
    \langle \hat{O}\rangle
    =\int d^{2l}\bm{\alpha}_Ad^{2L}\bm{\beta}p(\bm{\beta})q(\bm{\alpha}_A)X(\bm{\beta},\bm{\alpha}_A),
\end{align}
which implies that the average of the random variable $X(\bm{\beta},\bm{\alpha}_A)$ over the probability distribution $p(\bm{\beta})q(\bm{\alpha}_A)$ is equal to the mean value.
If we use the median-of-means estimator, we can upper-bound the error of the estimator by finding the upper bound on the second moment:
\begin{align}
    &\int d^{2l}\bm{\alpha_A}d^{2L}\bm{\beta}p(\bm{\beta})q(\bm{\alpha}_A)\frac{\|\hat{O}_A\|_2^4|\langle \psi_{\bm{\beta}}|(\hat{D}^\dagger(\bm{\alpha}_A)\otimes\hat{\mathbb{1}}_B\otimes \hat{D}(\bm{\beta})|\psi_G\rangle\langle \psi_G|\hat{D}^\dagger(\bm{\beta}))|\psi_{\bm{\beta}}\rangle|^2}{\pi^{2L}p(\bm{\beta})^2\left|\chi_{\hat{O}_A}(\bm{\alpha}_A)\right|^2} \\ 
    &=\frac{\|\hat{O}_A\|_2^2}{\pi^{2L}\pi^{l}}\int d^{2l}\bm{\alpha}_Ad^{2L}\bm{\beta}\frac{|\langle \psi_{\bm{\beta}}|(\hat{D}^\dagger(\bm{\alpha}_A)\otimes\hat{\mathbb{1}}_B\otimes \hat{D}(\bm{\beta})|\psi_G\rangle\langle \psi_G|\hat{D}^\dagger(\bm{\beta}))|\psi_{\bm{\beta}}\rangle|^2}{p(\bm{\beta})} \\ 
    &=\frac{\|\hat{O}_A\|_2^2}{\pi^{l}}\int d^{2l}\bm{\alpha}_Ad^{2L}\bm{\beta}\frac{|\langle \psi_{\bm{\beta}}|(\hat{D}^\dagger(\bm{\alpha}_A)\otimes\hat{\mathbb{1}}_B\otimes \hat{D}(\bm{\beta})|\psi_G\rangle\langle \psi_G|\hat{D}^\dagger(\bm{\beta}))|\psi_{\bm{\beta}}\rangle|^2}{\langle \psi_{\bm{\beta}}|(\hat{\mathbb{1}}_A\otimes\hat{D}(\bm{\beta})|\psi_G\rangle\langle \psi_G|\hat{D}^\dagger(\bm{\beta}))|\psi_{\bm{\beta}}\rangle} \\ 
    &=\frac{\|\hat{O}_A\|_2^2}{\pi^{l}}\int d^{2l}\bm{\alpha}_Ad^{2L}\bm{\beta}p(\bm{\beta})|\chi_{\hat{\rho}_{\bm{\beta}}}(\bm{\alpha}_A)|^2 \\ 
    &=\|\hat{O}_A\|_2^2\int d^{2L}\bm{\beta}p(\bm{\beta})\Tr[\hat{\rho}_{\bm{\beta}}^2].
\end{align}
Thus, the variance is upper-bounded by the Hilbert-Schmidt norm of the operator $\hat{O}_A$ multiplied by the average purity of the reduced state when tracing out the observable's Hilbert space.
Thus, the sample complexity for the estimator is given as $O(\|\hat{O}_A\|^2_2 \mathbb{E}_{\bm{\beta}}[\text{Tr}[\hat{\rho}_{\bm{\beta}}^2]]\log(1/\delta)/\epsilon^2)$.

Meanwhile, we need to be able to sample from $p(\bm{\beta})$ and $q(\bm{\alpha}_A)$.
Sampling from $q(\bm{\alpha}_A)$ can be efficiently done because of the product structure:
\begin{align}
    q(\bm{\alpha}_A)
    =\frac{\left|\chi_{\hat{O}_A}(\bm{\alpha}_A)\right|^2}{\pi^l\|\hat{O}_A\|_2^2}
    =\prod_{i\in A}\frac{|\chi_{\hat{O}_i}(\alpha_i)|^2}{\pi \|\hat{O}_i\|_2^2},
\end{align}
which takes at most $O(M)$.
Also, $p(\bm{\beta})$ can also be easily sampled if the number of modes $L$ in $|\bm{\beta}\rangle$ is not too large:
\begin{align}
    p(\bm{\beta})
    &=\frac{1}{\pi^L}\langle \psi|\hat{G}_{\bm{\beta}}^\dagger(\hat{\mathbb{1}}_M\otimes\hat{D}(\bm{\beta})|\psi_G\rangle\langle \psi_G|\hat{D}^\dagger(\bm{\beta}))\hat{G}_{\bm{\beta}}|\psi\rangle
    =\frac{1}{\pi^L}\langle \psi|\hat{G}_{\bm{\beta}}^\dagger\hat{D}^\dagger(\bm{\beta})(\hat{\mathbb{1}}_M\otimes|\bm{0}\rangle\langle \bm{0}|)\hat{D}(\bm{\beta})\hat{G}_{\bm{\beta}}|\psi\rangle \\ 
    &=\frac{1}{\pi^L}\langle \psi|\hat{D}^\dagger(\bm{\beta}')\hat{G}_{\bm{\beta}}^\dagger(\hat{\mathbb{1}}_M\otimes|\bm{0}\rangle\langle \bm{0}|)\hat{G}_{\bm{\beta}}\hat{D}(\bm{\beta}')|\psi\rangle
    =\frac{1}{\pi^L}\langle \psi'|\hat{G}_{\bm{\beta}}^\dagger(\hat{\mathbb{1}}_M\otimes|\bm{0}\rangle\langle \bm{0}|)\hat{G}_{\bm{\beta}}|\psi'\rangle,
\end{align}
where we used the property of the displacement operator.
\begin{align}
    \hat{D}^\dagger(\bm{\alpha}')=\hat{G}_{\bm{\beta}}^\dagger \hat{D}^\dagger(\bm{\alpha})\hat{G}_{\bm{\beta}}.    
\end{align}
More specifically, when $L=O(1)$, by Lemma~\ref{lemma:main}, we can sample from $p(\bm{\beta})$ in polynomial time.

Finally, we analyze the complexity of computing the random variable $X$:
\begin{align}
    X(\bm{\beta},\bm{\alpha})
    =\frac{\|\hat{O}_A\|_2^2\langle \psi_{\bm{\beta}}|(\hat{D}^\dagger(\bm{\alpha}_A)\otimes \hat{\mathbb{1}}_B\otimes |\bm{\beta}\rangle\langle \bm{\beta}|)|\psi_{\bm{\beta}}\rangle}{\pi^L p(\bm{\beta})\chi_{\hat{O}_A}^*(\bm{\alpha}_A)}.
\end{align}
The difficult part in computing the random variable is
\begin{align}
    &\langle \psi_{\bm{\beta}}|(\hat{D}^\dagger(\bm{\alpha}_A)\otimes \hat{\mathbb{1}}_B\otimes \hat{D}(\bm{\beta})|\psi_G\rangle\langle \psi_G|\hat{D}^\dagger(\bm{\beta}))|\psi_{\bm{\beta}}\rangle
    =\langle \psi|\hat{G}^\dagger(\bm{\beta})(\hat{D}^\dagger(\bm{\alpha}_A)\otimes \hat{\mathbb{1}}_B\otimes \hat{D}(\bm{\beta})|\psi_G\rangle\langle \psi_G|\hat{D}^\dagger(\bm{\beta}))\hat{G}(\bm{\beta})|\psi\rangle \\
    &=\langle \psi|\hat{G}^\dagger(\bm{\beta})(\hat{\mathbb{1}}\otimes \hat{D}(\bm{\beta})|\psi_G\rangle\langle \psi_G|\hat{D}^\dagger(\bm{\beta}))\hat{G}(\bm{\beta})\hat{D}(\bm{\alpha}_A')|\psi\rangle 
    =\langle \psi|\hat{G}^\dagger(\bm{\beta})(\hat{\mathbb{1}}\otimes |\bm{0}\rangle\langle \bm{0}|)\hat{G}(\bm{\beta})|\psi'\rangle,
\end{align}
but this can also be computed by using Lemma~\ref{lemma:main} when $L=O(1)$.
Thus, the total running time is given as
\begin{align}
    O\left(\frac{\poly(M)\|\hat{O}_A\|_2^2 \mathbb{E}_{\bm{\beta}}[\text{Tr}(\hat{\rho}_{\bm{\beta}}^2)] \log(1/\delta)}{\epsilon^2}\right),    
\end{align}
which completes the proof.

\end{proof}

\section{Derivation of $\tilde{G}(k)$}\label{app:MVS}
In this Appendix, we provide a derivation of the following relationship between $\tilde{G}(k)$ and $G(\Omega)$ in Eqs.~\eqref{eq:g_k} and \eqref{eq:g_omega}~\cite{oh2024quantum}:
\begin{align}
    \tilde{G}(k)
    &\equiv \langle \phi|\hat{G}_0^\dagger e^{-ik\theta \hat{\bm{n}}\cdot \bm{\omega}}\hat{G}_0|\psi\rangle, \\ 
    G(\Omega)
    &\equiv \sum_{\bm{n}=\bm{0}}^\infty q(\bm{n})\delta(\Omega-\bm{\omega}\cdot \bm{n})
    =\frac{1}{\Omega_\text{max}+1}\sum_{k=0}^{\Omega_\text{max}}\tilde{G}(k)e^{ik\Omega},
\end{align}
where $\theta=2\pi/(\Omega_\text{max}+1)$ and
\begin{align}
    q(\bm{n})\equiv \langle \phi|\hat{G}_0^\dagger |\bm{n}\rangle\langle \bm{n}|\hat{G}_0|\psi\rangle.
\end{align}

Let us compute the Fourier transform of $G(\Omega)$:
\begin{align}
    &\sum_{\Omega=0}^{\Omega_\text{max}}G(\Omega)e^{-ik\theta\Omega}
    =\sum_{\Omega=0}^{\Omega_\text{max}}\sum_{\bm{n}=\bm{0}}^\infty q(\bm{n})\delta(\Omega-\bm{\omega}\cdot \bm{n}) e^{-ik\theta\Omega}
    =\frac{1}{\Omega_\text{max}+1}\sum_{\Omega=0}^{\Omega_\text{max}}\sum_{\bm{n}=\bm{0}}^\infty q(\bm{n})\sum_{l=0}^{\Omega_\text{max}}e^{il\theta(\Omega-\bm{\omega}\cdot \bm{n})-ik\theta\Omega} \\ 
    &=\frac{1}{\Omega_\text{max}+1}\sum_{l=0}^{\Omega_\text{max}}\sum_{\bm{n}=\bm{0}}^\infty q(\bm{n})e^{-il\theta\bm{\omega}\cdot \bm{n}}\sum_{\Omega=0}^{\Omega_\text{max}}e^{i(l-k)\theta\Omega}
    =\sum_{\bm{n}=\bm{0}}^\infty q(\bm{n})e^{-ik\theta\bm{\omega}\cdot \bm{n}}
    =\sum_{\bm{n}=\bm{0}}^\infty \langle \phi|\hat{G}_0^\dagger |\bm{n}\rangle\langle \bm{n}|\hat{G}_0|\psi\rangle e^{-ik\theta\bm{\omega}\cdot \bm{n}} \\ 
    &=\langle \phi|\hat{G}_0^\dagger e^{-ik\theta \hat{\bm{n}}\cdot \bm{\omega}}\hat{G}_0|\psi\rangle
    =\tilde{G}(k).
\end{align}
Hence, $\tilde{G}(k)$ is the Fourier transformation of $G(\Omega)$, which enables us to compute the marginal probabilities of $q$.

\section{Low-rank structure in Gaussian circuits}\label{app:rank}
In this Appendix, we provide an explicit derivation of the transformation of bosonic annihilation operators through the Gaussian unitary circuit $\hat{G}_0^\dagger\hat{P}(\bm{\varphi})\hat{G}_0$.
To be more specific, we prove Lemma~\ref{lemma:low-rank}:
\begin{lemma}[Lemma~\ref{lemma:low-rank} restated]
    Consider a general Gaussian unitary circuit $\hat{G}_0$ and phase-shifter operator $\hat{P}(\bm{\varphi})$ with $\bm{\varphi}=(\varphi_1,\dots,\varphi_L,0,\dots,0)$. Then, $\hat{G}_0^\dagger\hat{P}(\bm{\varphi})\hat{G}_0$ can be decomposed into a linear-optical circuit, at most $2L$ non-zero squeezing operators, and another linear-optical circuit.
\end{lemma}
\begin{proof}

First of all, the bosonic annihilation operator transforms under $\hat{G}_0=\hat{U}\hat{S}(\bm{s})\hat{V}$ as 
\begin{align}
    \hat{G}_0\hat{a}_i\hat{G}_0^\dagger
    &=\hat{U}\hat{S}\hat{V}\hat{a}_i\hat{V}^\dagger \hat{S}^\dagger\hat{U}^\dagger
    =\hat{U}\hat{S}\left(\sum_{j=1}^M V_{ji}^* \hat{a}_j\right)\hat{S}^\dagger\hat{U}^\dagger
    =\hat{U}\left(\sum_{j=1}^M V_{ji}^* (\cosh s_j\hat{a}_j-\sinh s_j\hat{a}_j^\dagger)\right)\hat{U}^\dagger \\
    &=\sum_{j,k=1}^M V_{ji}^* (\cosh s_j U_{kj}^*\hat{a}_k-\sinh s_jU_{kj}\hat{a}_k^\dagger)
    =\sum_{k=1}^M \left(A_{ik}\hat{a}_k+B_{ik}\hat{a}_k^\dagger\right),
\end{align}
where we used
\begin{align}
    \hat{U}^\dagger \hat{a}_i\hat{U}=\sum_{j=1}^M U_{ij}\hat{a}_j,~~~
    \hat{U} \hat{a}_i\hat{U}^\dagger=\sum_{j=1}^M U_{ji}^*\hat{a}_j,~~~
    \hat{S}^\dagger(s)\hat{a}_i\hat{S}(s)=\cosh s\hat{a}_i+\sinh s\hat{a}_i^\dagger,~~~
    \hat{S}(s)\hat{a}_i\hat{S}^\dagger(s)=\cosh s\hat{a}_i-\sinh s\hat{a}_i^\dagger.
\end{align}
and we defined
\begin{align}
    A_{ik}\equiv \sum_{j=1}^M V_{ji}^*\cosh s_j U_{kj}^*,~~~
    B_{ik}\equiv -\sum_{j=1}^M V_{ji}^*\sinh s_j U_{kj},
\end{align}
which characterize the Gaussian transformation.
Here, we emphasize that the number of non-zero squeezing parameters decides the rank of the matrix $B$.
More precisely, if the number of non-zero squeezing parameters is $L$, then the rank of the matrix $B$ is at most $L$.

For the transformation under $\hat{G}^\dagger_0$, we similarly have
\begin{align}
    \hat{G}_0^\dagger\hat{a}_i\hat{G}_0
    &=\hat{V}^\dagger\hat{S}^\dagger\hat{U}^\dagger \hat{a}_i \hat{U}\hat{S}\hat{V}
    =\hat{V}^\dagger\hat{S}^\dagger\left(\sum_{j=1}^M U_{ij}\hat{a}_j\right)\hat{S}\hat{V}
    =\hat{V}^\dagger\left(\sum_{j=1}^M U_{ij}(\cosh s_j\hat{a}_j+\sinh s_j\hat{a}_j^\dagger) \right)\hat{V} \\ 
    &=\sum_{j,k=1}^M U_{ij}(\cosh s_j V_{jk}\hat{a}_k+\sinh s_jV_{jk}^*\hat{a}_k^\dagger)
    =\sum_{k=1}^M \left((A^\dagger)_{ik}\hat{a}_k-(B^\T)_{ik}\hat{a}_k^\dagger\right).
\end{align}

Because of the unitarity, after the transformation under $\hat{G}_0$ and then $\hat{G}_0^\dagger$, i.e., the identity operation,
\begin{align}
    \hat{a}_i\to \sum_{j=1}^M A_{ij}\hat{a}_j+\sum_{j=1}^M B_{ij}\hat{a}_j^\dagger
    &\to \sum_{j,k=1}^M A_{ij}\left((A^\dagger)_{jk}\hat{a}_k-(B^\T)_{jk}\hat{a}_k^\dagger\right)+\sum_{j,k=1}^M B_{ij}\left((A^\dagger)_{jk}^*\hat{a}_k^\dagger-(B^\T)_{jk}^*\hat{a}_k\right) \\ 
    &=\sum_{j,k=1}^M \left[\left(A_{ij}(A^\dagger)_{jk}-B_{ij}(B^\dagger)_{jk}\right)\hat{a}_k+\left(-A_{ij}(B^\T)_{jk}+B_{ij}(A^\T)_{jk}\right)\hat{a}_k^\dagger\right],
\end{align}
the matrices $A$ and $B$ must satisfy the following conditions:
\begin{align}\label{eq:AB_condition}
    AA^\dagger-BB^\dagger=\mathbb{1},~~~
    AB^\T=BA^\T.
\end{align}

Let $C\equiv A^\dagger$ and $D=-B^\T$.
After the entire process, including the phase shifter $\hat{P}(\bm{\varphi})$, the annihilation opertor transforms as
\begin{align}
    \hat{a}_i
    &\to \sum_{j=1}^M A_{ij}\hat{a}_j+\sum_{j=1}^M B_{ij}\hat{a}_j^\dagger 
    \to \sum_{j=1}^M e^{i\varphi_j}A_{ij}\hat{a}_j+\sum_{j=1}^M e^{-i\varphi_j}B_{ij}\hat{a}_j^\dagger \\
    &\to \sum_{j,k=1}^M e^{i\varphi_j}A_{ij}\left(C_{jk}\hat{a}_k+D_{jk} \hat{a}_k^\dagger\right)+\sum_{j,k=1}^M e^{-i\varphi_j}B_{ij}\left(C_{jk}^*\hat{a}_k^\dagger+D_{jk}^* \hat{a}_k\right) \\
    &=\sum_{j,k=1}^M e^{i\varphi_j}A_{ij}C_{jk}\hat{a}_k+\sum_{j,k=1}^M e^{i\varphi_j}A_{ij}D_{jk} \hat{a}_k^\dagger+\sum_{j,k=1}^M e^{-i\varphi_j}B_{ij}C_{jk}^*\hat{a}_k^\dagger+\sum_{j,k=1}^M e^{-i\varphi_j}B_{ij}D_{jk}^* \hat{a}_k.
\end{align}
Here, by assuming that $\varphi_j=0$ for $j>L$, without loss of generality, and using Eq.~\eqref{eq:AB_condition}, we can simplify the expression as
\begin{align}
    &\sum_{j,k=1}^M e^{i\varphi_j}A_{ij}C_{jk}\hat{a}_k+\sum_{j,k=1}^M e^{i\varphi_j}A_{ij}D_{jk} \hat{a}_k^\dagger+\sum_{j,k=1}^M e^{-i\varphi_j}B_{ij}C_{jk}^*\hat{a}_k^\dagger+\sum_{j,k=1}^M e^{-i\varphi_j}B_{ij}D_{jk}^* \hat{a}_k \\
    &=\sum_{j=1}^M\sum_{k=1}^M \left(A_{ij}C_{jk}\hat{a}_k+B_{ij}D_{jk}^* \hat{a}_k+A_{ij}D_{jk} \hat{a}_k^\dagger+B_{ij}C_{jk}^*\hat{a}_k^\dagger\right) \\ 
    &+\sum_{j=1}^L\sum_{k=1}^M \left((e^{i\varphi_j}-1)A_{ij}C_{jk}\hat{a}_k+(e^{-i\varphi_j}-1)B_{ij}D_{jk}^* \hat{a}_k+(e^{i\varphi_j}-1)A_{ij}D_{jk} \hat{a}_k^\dagger+(e^{-i\varphi_j}-1)B_{ij}C_{jk}^*\hat{a}_k^\dagger\right) \\ 
    &=\hat{a}_i+\sum_{j=1}^L\sum_{k=1}^M \left((e^{i\varphi_j}-1)A_{ij}C_{jk}\hat{a}_k+(e^{-i\varphi_j}-1)B_{ij}D_{jk}^* \hat{a}_k+(e^{i\varphi_j}-1)A_{ij}D_{jk} \hat{a}_k^\dagger+(e^{-i\varphi_j}-1)B_{ij}C_{jk}^*\hat{a}_k^\dagger\right) \\ 
    &=\hat{a}_i+\sum_{k=1}^M \left(W_{ik}\hat{a}_k+Z_{ik}\hat{a}_k^\dagger\right),
\end{align}
where we defined
\begin{align}
    W_{ik}
    &\equiv \sum_{j=1}^L \left((e^{i\varphi_j}-1)A_{ij}C_{jk}+(e^{-i\varphi_j}-1)B_{ij}D_{jk}^*\right)
    =\sum_{j=1}^L \left((e^{i\varphi_j}-1)A_{ij}A_{kj}^*-(e^{-i\varphi_j}-1)B_{ij}B_{kj}^*\right), \\
    Z_{ik}
    &\equiv \sum_{j=1}^L \left((e^{i\varphi_j}-1)A_{ij}D_{jk}+(e^{-i\varphi_j}-1)B_{ij}C_{jk}^*\right)
    =\sum_{j=1}^L \left(-(e^{i\varphi_j}-1)A_{ij}B_{kj}+(e^{-i\varphi_j}-1)B_{ij}A_{kj}\right),
\end{align}
and hence each of the terms has a low-rank structure.
More specifically, $W$ and $Z$'s ranks are at most $2L$.
Thus,
\begin{align}
    W=\sum_{j=1}^{2L}u^{(j)}v^{(j)\T},~~~
    Z=\sum_{j=1}^{2L}w^{(j)}z^{(j)\T}.
\end{align}
Here, note that $W+I$ and $Z$ form a Gaussian transform.
The low-rankness of the matrix $Z$ implies that the Gaussian transform has a structure in which its squeezing parameters are nonzero only for $2L$ modes.
\end{proof}

\section{Generalized Gurvits' second algorithm for linear-optical circuits}\label{app:gurvits_linear}

In this Appendix, we present a generalized version of Gurvits' second algorithm for linear-optical circuits with an arbitrary product input, generalizing it from Fock states:
\begin{align}
    \langle \phi|\hat{V}|\psi\rangle=\sum_{\bm{n},\bm{m}}a_{\bm{m}} b_{\bm{n}}^*\langle\bm{n}|\hat{V}|\bm{m}\rangle,
\end{align}
where $|\psi\rangle$ and $|\phi\rangle$ are arbitrary product states, which are written as
\begin{align}
    |\psi\rangle=
    \bigotimes_{i=1}^M\left(\sum_{m_i=0}^{n_\text{max}}a^{(i)}_{m_i}|m_i\rangle\right),~
    |\phi\rangle=
    \bigotimes_{i=1}^M\left(\sum_{n_i=0}^{n_\text{max}}b^{(i)}_{n_i}|n_i\rangle\right).
\end{align}
Here, $\bm{n}=(n_1,\dots,n_M)$ and $\bm{m}=(m_1,\dots,m_M)$ and we are interested in a linear-optical circuit $\hat{V}$ whose corresponding unitary matrix $V$ has a low-rank structure such that $V=W+I$ with a low-rank matrix $W$.
More precisely, let us assume that the rank of $W$ is $L$, so that it can be written as
\begin{align}
    W_{ij}=\sum_{k=1}^L u_i^{(k)}v_{j}^{(k)}.
\end{align}

To do that, let us first focus on a simpler case
\begin{align}
    \langle \bm{n}|\hat{V}|\bm{m}\rangle,
\end{align}
which is still a more general form than that of Gurvits' algorithm because $|\bm{n}\rangle$ is not equal to $|\bm{m}\rangle$ in general.
Here, because the linear-optical unitary circuit preserves the photon number, we assume that $|\bm{n}|=|\bm{m}|$; otherwise, it becomes zero.
Let us introduce a notation $X_{\bm{n}}$ to represent a multiset corresponding to $\bm{n}$, which is composed of $n_i$ times of $i$'s.
For example, for $\bm{n}=(2,1,0)$, $X_{\bm{n}}=\{1,1,2\}$.

First of all, we have
\begin{align}
    \langle\bm{n}|\hat{V}|\bm{m}\rangle 
    &=\frac{1}{\sqrt{\bm{n}!}}\langle0|\left(\prod_{i\in X_{\bm{n}}}\hat{a}_i\right) \hat{V}|\bm{m}\rangle 
    =\frac{1}{\sqrt{\bm{n}!}}\langle0|\prod_{i\in X_{\bm{n}}}\left(\hat{a}_i+\sum_{k=1}^MW_{ik}\hat{a}_k\right)|\bm{m}\rangle 
    \equiv \frac{1}{\sqrt{\bm{n}!}}\langle0|\prod_{i\in X_{\bm{n}}}\left(\hat{a}_i+\hat{b}_i\right)|\bm{m}\rangle \\ 
    &=\frac{1}{\sqrt{\bm{n}!}}\langle0|\left[\sum_{X\subset X_{\bm{n}}}\prod_{x\in X}\hat{a}_x\prod_{x\in X^c}\hat{b}_x\right]|\bm{m}\rangle 
    =\frac{1}{\sqrt{\bm{n}!}}\sum_{X\subset X_{\bm{n}}}\langle0|\left[\prod_{x\in X}\hat{a}_x\prod_{x\in X^c}\hat{b}_x\right]|\bm{m}\rangle \\ 
    &=\frac{1}{\sqrt{\bm{n}!}}\sum_{\bm{l}\leq \bm{n},\bm{m}}\binom{\bm{n}}{\bm{l}}\frac{\sqrt{\bm{m}!}}{\sqrt{(\bm{m}-\bm{l})!}}\langle0|\left[\prod_{x \in X_{\bm{n}-\bm{l}}}\sum_{j=1}^M W_{xj}\hat{a}_j^{\dagger} \right]|\bm{m}-\bm{l}\rangle \\ 
    &=\frac{\sqrt{\bm{m}!}}{\sqrt{\bm{n}!}}\sum_{\bm{l}\leq \bm{n},\bm{m}}\binom{\bm{n}}{\bm{l}}\frac{1}{(\bm{m}-\bm{l})!}\langle0|\left[\prod_{x \in X_{\bm{n}-\bm{l}}}\sum_{j=1}^M W_{xj}\hat{a}_j^{\dagger} \right]\left[\prod_{x\in X_{\bm{m}-\bm{l}}}\hat{a}_x^\dagger\right]|0\rangle,
\end{align}
where
\begin{align}
    &\langle0|\left[\prod_{x \in X_{\bm{n}-\bm{l}}}\sum_{j=1}^M W_{xj}\hat{a}_j^{\dagger} \right]\left[\prod_{x\in X_{\bm{m}-\bm{l}}}\hat{a}_x^\dagger\right]|0\rangle\\
    &=\sum_{\sigma \in \mathcal{S}_{|X_{\bm{n}-\bm{l}}|}:X_{\bm{n}-\bm{l}}\to X_{\bm{m}-\bm{l}}}\prod_{x\in X_{\bm{n}-\bm{l}}}W_{x,\sigma(x)} \\
    &=\sum_{\sigma \in \mathcal{S}_{|X_{\bm{n}-\bm{l}}|}:X_{\bm{n}-\bm{l}}\to X_{\bm{m}-\bm{l}}}\prod_{x\in X_{\bm{n}-\bm{l}}}\sum_{k=1}^L u_{x}^{(k)}v_{\sigma(x)}^{(k)} \\ 
    &=\sum_{\sigma \in \mathcal{S}_{|X_{\bm{n}-\bm{l}}|}:X_{\bm{n}-\bm{l}}\to X_{\bm{m}-\bm{l}}}\sum_{k_{1},\dots,k_{|X_{\bm{n}-\bm{l}}|}=1}^L\prod_{i=1}^{|X_{\bm{n}-\bm{l}}|} u_{X_{\bm{n}-\bm{l}}(i)}^{(k_i)}v_{X_{\bm{m}-\bm{l}}(i)}^{(k_{\sigma(i)})}  \\
    &=\sum_{s_u(X_{\bm{n}-\bm{l}})}\sum_{s_v(X_{\bm{m}-\bm{l}})}\delta([s_u(X_{\bm{n}-\bm{l}})]=[s_v(X_{\bm{m}-\bm{l}})])\prod_{i=1}^{|X_{\bm{n}-\bm{l}}|} u_{X_{\bm{n}-\bm{l}}(i)}^{(s_u(i))}v_{X_{\bm{m}-\bm{l}}(i)}^{(s_v(i))} \prod_{i=1}^L k_i! \\ 
    &=\sum_{s_u(X_{\bm{n}-\bm{l}})}\sum_{s_v(X_{\bm{m}-\bm{l}})}\delta([s_u(X_{\bm{n}-\bm{l}})]=[s_v(X_{\bm{m}-\bm{l}})])\prod_{x\in X_{\bm{n}-\bm{l}}} u_{x}^{(s_u(x))}\prod_{x\in X_{\bm{m}-\bm{l}}}v_{x}^{(s_v(x))} \prod_{i=1}^L k_i!,
\end{align}
where the summation over $s(X)$ is interpreted as the summation of a tuple from $[L]^{|X|}$ and $[s(X)]$ denotes the multiset of the elements of $s(X)$, and $k_i$'s denote the multiplicity of $i$ in the multiset $[s_u(X)]$, which arises because of the redundancy.

At this moment, let us ignore the delta function and the multiplicity factor $\prod_i k_i!$, which will be taken care of later.
The full expression is then rewritten as
\begin{align}
    &\frac{\sqrt{\bm{m}!}}{\sqrt{\bm{n}!}}\sum_{\bm{l}\leq \bm{n},\bm{m}}\binom{\bm{n}}{\bm{l}}\frac{1}{(\bm{m}-\bm{l})!}\sum_{s_u(X_{\bm{n}-\bm{l}})}\sum_{s_v(X_{\bm{m}-\bm{l}})}\prod_{x\in X_{\bm{n}-\bm{l}}} u_{x}^{(s_u(x))}\prod_{x\in X_{\bm{m}-\bm{l}}}v_{x}^{(s_v(x))} \\ 
    &=\frac{\sqrt{\bm{m}!}}{\sqrt{\bm{n}!}}\sum_{\bm{l}\leq \bm{n},\bm{m}}\binom{\bm{n}}{\bm{l}}\frac{1}{(\bm{m}-\bm{l})!}\left(\prod_{x\in X_{\bm{n}-\bm{l}}}\sum_{s=1}^L u_{x}^{(s)}\right)\left(\prod_{x\in X_{\bm{m}-\bm{l}}} \sum_{s=1}^L v_{x}^{(s)}\right) \\ 
    &=\frac{\sqrt{\bm{m}!}}{\sqrt{\bm{n}!}}\prod_{i=1}^M\sum_{l\leq \min(n_i,m_i)}\binom{n_i}{l}\frac{1}{(m_i-l)!}\left(\sum_{s=1}^L u_{i}^{(s)}\right)^{n_i-l}\left(\sum_{s=1}^L v_{i}^{(s)}\right)^{m_i-l}.
\end{align}

Now, to take care of the delta function, we introduce the indeterminates $\{a_u^{(s)}\}_{s=1}^L$ and $\{a_v^{(s)}\}_{s=1}^L$ and write the expression as
\begin{align}
    \prod_{i=1}^M \frac{\sqrt{m_i!}}{\sqrt{n_i!}}\sum_{l\leq \min(n_i,m_i)}\binom{n_i}{l}\frac{1}{(m_i-l)!}\left(\sum_{s=1}^L a_u^{(s)}u_{x}^{(s)}\right)^{n_i-l}\left(\sum_{s=1}^L a_v^{(s)}v_{i}^{(s)}\right)^{m_i-l}
\end{align}
and keep the coefficients of the monomials in which $a_u^{(s)}$ and $a_v^{(s)}$'s degrees match each other (e.g., $(a_u^{(1)})^2(a_u^{(2)})^3(a_v^{(1)})^2(a_v^{(2)})^3$).
Finally, to take care of the multiplicity factor $\prod_{i}k_i$, we then multiply the factorial of degrees of $a_u^{(s)}$ (or equivalently $a_v^{(s)}$).

Now, for arbitrary product input states $|\psi\rangle$ and $|\phi\rangle$, we have
\begin{align}
    \langle \phi|\hat{V}|\psi\rangle
    &=\sum_{\bm{n},\bm{m}}a_{\bm{m}} b_{\bm{n}}^*\langle\bm{n}|\hat{V}|\bm{m}\rangle \\
    &\to \sum_{\bm{n},\bm{m}}a_{\bm{m}}b_{\bm{n}}^*\prod_{i=1}^M \frac{\sqrt{m_i!}}{\sqrt{n_i!}}\sum_{l\leq \min(n_i,m_i)}\binom{n_i}{l}\frac{1}{(m_i-l)!}\left(\sum_{s=1}^L a_u^{(s)}u_{i}^{(s)}\right)^{n_i-l}\left(\sum_{s=1}^L a_v^{(s)}v_{i}^{(s)}\right)^{m_i-l} \\ 
    &\to \prod_{i=1}^M \sum_{n_i,m_i=0}^{n_\text{max}}a_{m_i}^{(i)}b_{n_i}^{(i)*} \frac{\sqrt{m_i!}}{\sqrt{n_i!}}\sum_{l\leq \min(n_i,m_i)}\binom{n_i}{l}\frac{1}{(m_i-l)!}\left(\sum_{s=1}^L a_u^{(s)}u_{i}^{(s)}\right)^{n_i-l}\left(\sum_{s=1}^L a_v^{(s)}v_{i}^{(s)}\right)^{m_i-l}\label{eq:LON_last}.
\end{align}

Let us now analyze the complexity.
The number of indeterminates we introduce is $2L$, and each indeterminate may have degree at most $Mn_\text{max}$.
Thus, the number of coefficients we must keep track of is $(Mn_\text{max}+1)^{2L}$.
Since there are $M$ product, the total number of operations is $O(M(Mn_\text{max}+1)^{2L})$.

\section{Antinormal ordering}\label{app:anti}
The goal of this Appendix is to find the antinormal ordering of the following term:
\begin{align}
    \prod_{i=1}^N(\hat{A}_i+\hat{B}_i^\dagger),
\end{align}
where the commutation relations of the operators are given as $[\hat{A}_i,\hat{A}_j]=0$,  $[\hat{B}_i^\dagger,\hat{B}_j^\dagger]=0$, and $[\hat{A}_i,\hat{B}_j^\dagger]=C_{ij}$.
This Appendix is based on Ref.~\cite{varvak2005rook}.
To find its antinormal ordering, i.e., express the above by operators in which $\hat{A}_i$'s always appear on the left and $\hat{B}_i^\dagger$'s always appear on the right, we consider each term after expanding it, such as $\hat{A}_1\hat{B}_2^\dagger\hat{A}_3\hat{A}_4\hat{B}_5^\dagger\hat{A}_6$.
Now, we move all the creation operators $\hat{B}_i^\dagger$'s to the right by using the commutation relation.
First, we denote $Y$ as the indices of the annihilation operators $\hat{A}_i$ and $Y^c$ as those of the creation operators $\hat{B}_i^\dagger$.
For example, for $\hat{A}_1\hat{B}_2^\dagger\hat{A}_3\hat{A}_4\hat{B}_5^\dagger\hat{A}_6$, $Y=\{1,3,4,6\}$ and $Y^c=\{2,5\}$.
We then move $\hat{B}_i^\dagger$ for $i\in Y^c$ one by one using the commutation relation.
First of all, note that when we move $\hat{B}_j^\dagger$ to the right, it only interacts with $\hat{A}_i$ with $i>j$.

Let us start with the largest $j\in Y^c$, say $5$ for the above example.
Using $\hat{B}_5^\dagger \hat{A}_6=[\hat{B}_5^\dagger, \hat{A}_6]+\hat{A}_6\hat{B}_5^\dagger=-C_{65}+\hat{A}_6\hat{B}_5^\dagger$, we now have two terms
\begin{align}
    \hat{A}_1\hat{B}_2^\dagger\hat{A}_3\hat{A}_4\hat{B}_5^\dagger\hat{A}_6
    =-C_{65}\hat{A}_1\hat{B}_2^\dagger\hat{A}_3\hat{A}_4+\hat{A}_1\hat{B}_2^\dagger\hat{A}_3\hat{A}_4 \hat{A}_6 \hat{B}_5^\dagger.
\end{align}
Thus, when we use the commutation relation between $\hat{A}_i$ and $\hat{B}_j^\dagger$, the associated term is split into two terms: one that changes $\hat{A}_i\hat{B}_j^\dagger$ to $-C_{ij}$ and the other that simply changes $\hat{A}_i\hat{B}_j^\dagger$ to $\hat{B}_j^\dagger\hat{A}_i$.
Using this property, one can easily see that for a given $Y$ and $Y^c$, its antinormal ordering is simply written as
\begin{align}
    \sum_{\mathcal{M}\in \text{Match}(F_{N};Y\to Y^c)}(-1)^{|\mathcal{M}|}\prod_{(i,j)\in \mathcal{M}}C_{ij}\prod_{i\in Y\setminus A(\mathcal{M})}\hat{A}_i\prod_{j\in Y^c\setminus B(\mathcal{M})}\hat{B}_j^\dagger,
\end{align}
where the Match is the possible matches in $F_N$, which is the Ferrers graph whose edge set is $E=\{(i,j)|1\leq j<i\leq N\}$~\cite{stanley2011enumerative}, and $A(\mathcal{M})$ and $B(\mathcal{M})$ represent the elements in $Y$ and $Y^c$ that are involved in the matching $\mathcal{M}$, respectively.

Since we need to take into account all possible pairs $Y$ and $Y^c$, we have
\begin{align}
    \prod_{i=1}^N(\hat{A}_i+\hat{B}_i^\dagger)
    =\sum_{Y\subset [N]}\sum_{\mathcal{M}\in \text{Match}(F_{N};Y\to Y^c)}(-1)^{|\mathcal{M}|}\prod_{(i,j)\in \mathcal{M}}C_{ij}\prod_{i\in Y\setminus A(\mathcal{M})}\hat{A}_i\prod_{j\in Y^c\setminus B(\mathcal{M})}\hat{B}_j^\dagger.
\end{align}

\section{Generalized Gurvits' second algorithm for Gaussian circuits}\label{app:gurvits_general}

In this Appendix, we provide a classical algorithm that computes the following using a low-rank property of the relevant matrices,
\begin{align}
    \langle \phi|\hat{G}|\psi\rangle=\sum_{\bm{n},\bm{m}}a_{\bm{m}} b_{\bm{n}}^*\langle\bm{n}|\hat{G}|\bm{m}\rangle,
\end{align}
where $|\psi\rangle$ and $|\phi\rangle$ are arbitrary product states, which are written due to the product structure as
\begin{align}
    |\psi\rangle=
    \bigotimes_{i=1}^M\left(\sum_{m_i=0}^{n_\text{max}}a^{(i)}_{m_i}|m_i\rangle\right),~
    |\phi\rangle=
    \bigotimes_{i=1}^M\left(\sum_{n_i=0}^{n_\text{max}}b^{(i)}_{n_i}|n_i\rangle\right).
\end{align}
Also, $\hat{G}$ is a Gaussian unitary circuit, and in particular, we are interested in a Gaussian unitary circuit generated by $\hat{G}_0^\dagger\hat{P}(\bm{\varphi})\hat{G}_0$, where $\hat{G}_0$ is an arbitrary Gaussian unitary circuit and $\hat{P}(\bm{\varphi})$ is a phase shifter with phases $\bm{\varphi}$ whose nonzero elements are a few, i.e., $\bm{\varphi}=(\varphi_1,\dots,\varphi_L,0,\dots,0)$.
Hence, we prove Lemma~\ref{lemma:MVS}:
\begin{lemma}[Lemma~\ref{lemma:MVS} restated]
    Consider an $M$-mode bosonic system. Let $\hat{G}_0$ be a Gaussian unitary circuit and $\hat{P}(\bm{\varphi})$ be the phase shifter with phase $\bm{\varphi}$ and $|\psi\rangle$ and $|\phi\rangle$ be product states.
    If the number of non-zero elements in $\bm{\varphi}$ is $L=O(1)$,
    \begin{align}
        \langle \phi |\hat{G}_0^\dagger\hat{P}(\bm{\varphi})\hat{G}_0|\psi\rangle
    \end{align}
    can be exactly computed in polynomial time in $M$.
\end{lemma}

\begin{proof}
Before we derive the full expression and the classical algorithm, let us first focus on $\langle\bm{n}|\hat{G}|\bm{m}\rangle$ and then generalize, where $\hat{G}\equiv \hat{G}_0^\dagger\hat{P}(\bm{\varphi})\hat{G}_0$.
Let $X_{\bm{n}}$ be the multiset representation of $\bm{n}$, which is composed of $n_i$ times of $i$.
For example, when $\bm{n}=(2,1,4)$, $X_{\bm{n}}=\{1,1,2,3,3,3,3\}$.
$X_{\bm{m}}$ is similar.
Let $\hat{G}=\hat{U}\hat{S}(\bm{r})\hat{V}$ be the Bloch-Messiah decomposition of $\hat{G}$, following Appendix~\ref{app:rank}.
Then,
\begin{align}
    \langle\bm{n}|\hat{G}|\bm{m}\rangle
    &=\frac{1}{\sqrt{\bm{n}!}}\langle0|\hat{G}\hat{G}^\dagger\left(\prod_{i\in X_{\bm{n}}}\hat{a}_i\right)\hat{G}|\bm{m}\rangle 
    =\frac{1}{\sqrt{\bm{n}!}}\langle0|\hat{G} \left[\prod_{i\in X_{\bm{n}}}\left(\hat{a}_i+\sum_{k=1}^M\left(W_{ik}\hat{a}_k+Z_{ik}\hat{a}_k^\dagger\right)\right)\right]|\bm{m}\rangle \\ 
    &=\frac{1}{\sqrt{\bm{n}!}}\langle\bm{r}|\hat{V} \left[\prod_{i\in X_{\bm{n}}}\left(\hat{a}_i+\sum_{k=1}^M\left(W_{ik}\hat{a}_k+Z_{ik}\hat{a}_k^\dagger\right)\right)\right]|\bm{m}\rangle 
    =\frac{1}{\sqrt{\bm{n}!}}\langle\bm{r}|\hat{V} \left[\prod_{i\in X_{\bm{n}}}\left(\hat{A}_i+\hat{B}_i^\dagger\right)\right]|\bm{m}\rangle,    
\end{align}
where we defined
\begin{align}
    \hat{A}_i=\hat{a}_i+\sum_{k=1}^MW_{ik}\hat{a}_k,~~~\hat{B}_i^\dagger=\sum_{k=1}^M Z_{ik}\hat{a}_k^\dagger,~~~
    C_{ij}=[\hat{A}_i,\hat{B}_j^\dagger]=Z_{ji}+\sum_{k=1}^M W_{ik}Z_{jk}.
\end{align}
Note that, as shown in Appendix~\ref{app:rank}, $W$, $Z$, and $C$'s ranks are at most $2L$ and that only at most $2L$ $r_i$'s are nonzero. Without loss of generality, we set $r_i=0$ for $i>2L$.

By using the antinormal ordering~\cite{varvak2005rook} (see Appendix~\ref{app:anti}), let us expand it and move all the creation operators to the right, 
\begin{align}
    \frac{1}{\sqrt{\bm{n}!\bm{m}!}}\sum_{Y\subset X_{\bm{n}}}\sum_{\mathcal{M}\in \text{Match}(F_{X_{\bm{n}}};Y\to Y^c)}(-1)^{|\mathcal{M}|}\left(\prod_{(i,j)\in \mathcal{M}}C_{ij}\right)\langle \bm{r}|\hat{V}\left(\prod_{i\in Y\setminus A(\mathcal{M})}\hat{A}_i\right)\left(\prod_{i\in Y^c\setminus B(\mathcal{M})}\hat{B}^\dagger_i\right)\prod_{i\in X_{\bm{m}}}\hat{a}_i^\dagger|0\rangle,
\end{align}
where $Y$ and $Y^c\equiv X_{\bm{n}}\setminus Y$ are multisets that are subsets of $X_{\bm{n}}$ and the Match is the possible matches in $F_{X_{\bm{n}}}$, which is the Ferrers graph whose edge set is $E=\{(i,j)\in X_{\bm{n}}\times X_{\bm{n}}|j<i\}$.
Also, $A(\mathcal{M})$ and $B(\mathcal{M})$ are the elements that are involved in $M$ from $Y$ and $Y^c$, respectively.
That means that $Y\setminus A(\mathcal{M})$ and $Y^c\setminus B(\mathcal{M})$ are vertices that are not involved in the matching $\mathcal{M}$.
Here,
\begin{align}
    &\langle \bm{r}|\hat{V}\left(\prod_{i\in Y\setminus A(\mathcal{M})}\hat{A}_i\right)\left(\prod_{i\in Y^c\setminus B(\mathcal{M})}\hat{B}^\dagger_i\right)\prod_{i\in X_{\bm{m}}}\hat{a}_i^\dagger|0\rangle \\ 
    &=\langle \bm{r}|\hat{V}\sum_{X\subset Y\setminus A(\mathcal{M})}\prod_{i\in X}\hat{a}_i\left(\prod_{i\in Y\setminus A(\mathcal{M})\setminus X}\sum_{k=1}^M W_{ik}\hat{a}_k\right)\left(\prod_{i\in Y^c\setminus B(\mathcal{M})}\sum_{k=1}^M Z_{jk}\hat{a}_k^\dagger\right)\left(\prod_{i\in X_{\bm{m}}}\hat{a}_i^\dagger\right)|0\rangle.
\end{align}
Thus, let us first simplify the following:
\begin{align}
    \langle \bm{r}|\hat{V} \prod_{i\in A}\hat{a}_i\prod_{i\in B}\hat{b}_i\prod_{i\in C}\hat{c}_i^\dagger \prod_{i\in D} \hat{a}_i^\dagger |0\rangle,
\end{align}
where $A\iff X, B\iff Y\setminus A(\mathcal{M})\setminus X, C\iff Y^c\setminus B(\mathcal{M}), D\iff X_{\bm{m}}$.
Here,
\begin{align}
    &\prod_{i\in A}\hat{a}_i\prod_{i\in B}\hat{b}_i\prod_{i\in C}\hat{c}_i^\dagger \prod_{i\in D} \hat{a}_i^\dagger |0\rangle \\ 
    &=\prod_{i\in A}\hat{a}_i\prod_{i\in B}\sum_{j=1}^MW_{ij}\hat{a}_j\prod_{i\in C}\sum_{j=1}^MZ_{ij}\hat{a}_j^\dagger \prod_{i\in D} \hat{a}_i^\dagger |0\rangle \\ 
    &=\sum_{j_{B(1)},\dots,j_{B(|B|)}=1}^M \sum_{j'_{C(1)},\dots,j'_{C(|C|)}=1}^M \prod_{i\in B}W_{ij_i}\prod_{i\in C}Z_{ij'_i} \prod_{i\in A}\hat{a}_{i}\prod_{i\in B}\hat{a}_{j_i}\prod_{i\in C}\hat{a}_{j'_i}^\dagger \prod_{i\in D} \hat{a}_i^\dagger |0\rangle.
\end{align}
To simplify this expression, we first need to match all $A$ and $B$ to $C$ and $D$ first, and then the remaining elements in $C$ and $D$ will be matched with the squeezed state $\langle \bm{r}|\hat{V}$.
Here, $a\subset A$ matches with $a=d_a\subset D$ and $b\subset B$ with $d_b\subset D\setminus a$, and $A\setminus a$ with $c_a\subset C$ and $B\setminus b$ with $c_b$.
Thus,
\begin{align}
    &\prod_{i\in A}\hat{a}_{i}\prod_{i\in B}\hat{a}_{j_i}\prod_{i\in C}\hat{a}_{j'_i}^\dagger \prod_{i\in D} \hat{a}_i^\dagger |0\rangle \\
    &=\sum_{a\subset A\cap D}\sum_{b\subset B}\sum_{c_a,c_b\subset C}\sum_{d_b\subset D\setminus d_a}\sum_{\sigma\in\mathcal{S}_{|a|}:a\to d_a}\prod_{i\in a}\delta_{i,\sigma(i)}\sum_{\sigma\in\mathcal{S}_{|a^c|}:a^c\to c_a}\prod_{i\in a^c}\delta_{i,j'_{\sigma(i)}}\sum_{\sigma\in\mathcal{S}_{|b|}:b\to d_b}\prod_{i\in b}\delta_{j_i,\sigma(i)}\sum_{\sigma\in \mathcal{S}_{|B\setminus b|}:B\setminus b\to c_b}\prod_{i\in B\setminus b}\delta_{j_i,j'_{\sigma(i)}} \nonumber \\ 
    &\times \prod_{i\in C\setminus (c_a\cup c_b)}\hat{a}_{j_i'}^\dagger\prod_{i\in D\setminus(d_a\cup d_b)}\hat{a}_{i}^\dagger|0\rangle.
\end{align}
Hence, we have
\begin{align}
    &\langle \bm{r}|\hat{V} \sum_{j_{B(1)},\dots,j_{B(|B|)}=1}^M \sum_{j'_{C(1)},\dots,j'_{C(|C|)}=1}^M \prod_{i\in B}W_{ij_i}\prod_{i\in C}Z_{ij'_i} \prod_{i\in A}\hat{a}_{i}\prod_{i\in B}\hat{a}_{j_i}\prod_{i\in C}\hat{a}_{j'_i}^\dagger \prod_{i\in D} \hat{a}_i^\dagger |0\rangle \\ 
    &=\sum_{a\subset A\cap D}\sum_{b\subset B}\sum_{c_a,c_b\subset C}\sum_{d_b\subset D\setminus d_a}\sum_{\sigma\in\mathcal{S}_{|a^c|}:a^c\to c_a}\prod_{i\in a^c}Z_{i,\sigma(i)}\sum_{\sigma\in\mathcal{S}_{|b|}:b\to d_b}\prod_{i\in b}W_{i,\sigma(i)}\sum_{\sigma\in \mathcal{S}_{|B\setminus b|}:B\setminus b\to c_b}\prod_{i\in B\setminus b}(WZ^\T)_{i,\sigma(i)} \nonumber \\ 
    &\times \langle \bm{r}|\hat{V}\prod_{i\in C\setminus (c_a\cup c_b)}\sum_{j_i'=1}^MZ_{ij_i'}\hat{a}_{j_i'}^\dagger\prod_{i\in D\setminus(d_a\cup d_b)}\hat{a}_{i}^\dagger|0\rangle \\ 
    &=\sum_{a\subset A\cap D}\sum_{b\subset B}\sum_{c_a,c_b\subset C}\sum_{d_b\subset D\setminus d_a}\sum_{\sigma\in\mathcal{S}_{|a^c|}:a^c\to c_a}\prod_{i\in a^c}Z_{i,\sigma(i)}\sum_{\sigma\in\mathcal{S}_{|b|}:b\to d_b}\prod_{i\in b}W_{i,\sigma(i)}\sum_{\sigma\in \mathcal{S}_{|B\setminus b|}:B\setminus b\to c_b}\prod_{i\in B\setminus b}(WZ^\T)_{i,\sigma(i)} \nonumber \\ 
    &\times \langle \bm{r}|\prod_{i\in C\setminus (c_a\cup c_b)}\sum_{j_i',k=1}^MZ_{ij_i'}V_{k,j_i'}\hat{a}_{k}^\dagger\prod_{i\in D\setminus(d_a\cup d_b)}\sum_{k=1}^MV_{k,i}\hat{a}_{k}^\dagger|0\rangle \\ 
    &=\sum_{a\subset A\cap D}\sum_{b\subset B}\sum_{c_a,c_b\subset C}\sum_{d_b\subset D\setminus d_a}\sum_{\sigma\in\mathcal{S}_{|a^c|}:a^c\to c_a}\prod_{i\in a^c}Z_{i,\sigma(i)}\sum_{\sigma\in\mathcal{S}_{|b|}:b\to d_b}\prod_{i\in b}W_{i,\sigma(i)}\sum_{\sigma\in \mathcal{S}_{|B\setminus b|}:B\setminus b\to c_b}\prod_{i\in B\setminus b}(WZ^\T)_{i,\sigma(i)} \nonumber \\ 
    &\times \langle \bm{r}|\prod_{i\in C\setminus (c_a\cup c_b)}\sum_{k=1}^M(ZV^\T)_{i,k}\hat{a}_{k}^\dagger\prod_{i\in D\setminus(d_a\cup d_b)}\sum_{k=1}^M(V^\T)_{i,k}\hat{a}_{k}^\dagger|0\rangle,
\end{align}
where $c_a\cap c_b=\emptyset$ and $d_a\cap d_b=\emptyset$.
Here, using the low-rank structure of $W$ and $Z$, we simplify the expression further as
\begin{align}
    &\sum_{\sigma\in\mathcal{S}_{|a^c|}:a^c\to c_a}\prod_{i\in a^c}Z_{i,\sigma(i)}\sum_{\sigma\in\mathcal{S}_{|b|}:b\to d_b}\prod_{i\in b}W_{i,\sigma(i)}\sum_{\sigma\in \mathcal{S}_{|B\setminus b|}:B\setminus b\to c_b}\prod_{i\in B\setminus b}(WZ^\T)_{i,\sigma(i)} \\ 
    &=\sum_{\sigma\in\mathcal{S}_{|a^c|}:a^c\to c_a}\prod_{i\in a^c}\sum_{l=1}^{2L} u_{i}^{(l)}v_{\sigma(i)}^{(l)}\sum_{\rho\in\mathcal{S}_{|b|}:b\to d_b}\prod_{i\in b}\sum_{k=1}^{2L} w_{i}^{(k)}x_{\rho(i)}^{(k)}\sum_{\omega\in \mathcal{S}_{|B\setminus b|}:B\setminus b\to c_b}\prod_{i\in B\setminus b}\sum_{m=1}^{2L} y_i^{(m)}z_{\omega(i)}^{(m)} \\ 
    &=\sum_{\sigma\in\mathcal{S}_{|a^c|}:a^c\to c_a}\sum_{l_{a^c(1)},\dots,l_{a^c(|a^c|)}=1}^{2L}\prod_{i\in a^c} u_i^{(l_i)}v_{\sigma(i)}^{(l_i)}\sum_{\rho\in\mathcal{S}_{|b|}:b\to d_b}\sum_{k_{b(1)},\dots,k_{b(|b|)}=1}^{2L}\prod_{i\in b} w_{i}^{(k_i)}x_{\rho(i)}^{(k_i)}  \nonumber\\ 
    &\times \sum_{\omega\in \mathcal{S}_{|B\setminus b|}:B\setminus b\to c_b}\sum_{m_{(B\setminus b)(1)},\dots,m_{(B\setminus b)(|B\setminus b|)}=1}^{2L}\prod_{i\in B\setminus b} y_i^{(m_i)}z_{\omega(i)}^{(m_i)} \\ 
    &=\sum_{\sigma\in\mathcal{S}_{|a^c|}:a^c\to c_a}\sum_{\rho\in\mathcal{S}_{|b|}:b\to d_b}\sum_{\omega\in \mathcal{S}_{|B\setminus b|}:B\setminus b\to c_b}\sum_{l_{1},\dots,l_{|a^c|}=1}^{2L}\sum_{k_{1},\dots,k_{|b|}=1}^{2L}\sum_{m_{1},\dots,m_{|B\setminus b|}=1}^{2L}\prod_{i\in [a^c]} u_{c_a(i)}^{(l_{i})}v_{a^c(i)}^{(l_{\sigma(i)})}\prod_{i\in [b]} w_{b(i)}^{(k_{\rho(i)})}x_{d_b(i)}^{(k_{i})} \nonumber \\ 
    &\times \prod_{i\in [B\setminus b]} y_{c_b(i)}^{(m_i)}z_{(B\setminus b)(i)}^{(m_{\omega(i)})} \\ 
    &=\sum_{s_u(c_a),s_v(a^c)}\sum_{s_w(b),s_x(d_b)}\sum_{s_{z}(B\setminus b),s_y(c_b)}\delta([s_u(c_a)]=[s_v(a^c)])\delta([s_w(b)]=[s_x(d_b)])\delta([s_z(B\setminus b)]=[s_y(c_b)]) \nonumber \\ 
    &\times \prod_{i\in [a^c]} u_{c_a(i)}^{(l_{s_u(i)})}v_{a^c(i)}^{(s_v(i))}\prod_{i\in [b]} w_{b(i)}^{(s_w(i))}x_{d_b(i)}^{(k_{s_x(i)})}\prod_{i\in [B\setminus b]} y_{c_b(i)}^{(s_z(i))}z_{(B\setminus b)(i)}^{(s_y(i))}\prod_{i=1}^L(k_i!k_i'!k_i''!),
\end{align}
where the summation over $s(X)$ is interpreted as the summation of a tuple from $[L]^{|X|}$ and $k_i,k_i',k_i''$ represent the multiplicity of $u,w,y$, respectively.
Thus, ignoring the delta functions and the multiplicity factors, which will be taken care of later, we can rewrite the expression as
\begin{align}
    &\sum_{a\subset A\cap D}\sum_{b\subset B}\sum_{c_a,c_b\subset C}\sum_{d_b\subset D\setminus a}\sum_{s_u(c_a),s_v(a^c)}\sum_{s_w(b),s_x(d_b)}\sum_{s_{z}(B\setminus b),s_y(c_b)}\prod_{i\in [a^c]} u_{c_a(i)}^{(l_{s_u(i)})}v_{a^c(i)}^{(s_v(i))}\prod_{i\in [b]} w_{b(i)}^{(s_w(i))}x_{d_b(i)}^{(k_{s_x(i)})}\prod_{i\in [B\setminus b]} y_{c_b(i)}^{(s_z(i))}z_{(B\setminus b)(i)}^{(s_y(i))} \\ 
    &=\sum_{a\subset A\cap D}\sum_{b\subset B}\sum_{c_a,c_b\subset C}\sum_{d_b\subset D\setminus a}\prod_{i \in c_a}\sum_{s=1}^{2L}u_{i}^{(s)}a_u^{(s)}\prod_{i\in a^c}\sum_{s=1}^{2L}v_{i}^{(s)}a_v^{(s)} \prod_{i\in b}\sum_{s=1}^{2L} w_{i}^{(s)}a_w^{(s)}\prod_{i\in d_b}\sum_{s=1}^{2L}x_{i}^{(s)}a_x^{(s)}\prod_{i\in c_b}\sum_{s=1}^{2L}y_{i}^{(s)}a_y^{(s)}\prod_{i\in B\setminus b} \sum_{s=1}^{2L}z_{i}^{(s)}a_z^{(s)} \\ 
    &=\sum_{a\subset A\cap D}\sum_{b\subset B}\sum_{c_a,c_b\subset C}\sum_{d_b\subset D\setminus a}\prod_{i \in c_a}f_1(i)\prod_{i\in a^c}f_2(i) \prod_{i\in b}f_3(i)\prod_{i\in d_b}f_4(i)\prod_{i\in c_b}f_5(i)\prod_{i\in B\setminus b} f_6(i),
\end{align}
where we introduced the indeterminates $a_u^{(s)},a_v^{(s)},a_w^{(s)},a_x^{(s)},a_y^{(s)},a_z^{(s)}$ to take care of the delta functions and defined
\begin{align}
    f_1(i)=\sum_{s=1}^{2L}u_i^{(s)}a_u^{(s)},
    f_2(i)=\sum_{s=1}^{2L}v_i^{(s)}a_v^{(s)},
    f_3(i)=\sum_{s=1}^{2L}w_i^{(s)}a_w^{(s)},
    f_4(i)=\sum_{s=1}^{2L}x_i^{(s)}a_x^{(s)},
    f_5(i)=\sum_{s=1}^{2L}y_i^{(s)}a_y^{(s)},
    f_6(i)=\sum_{s=1}^{2L}z_i^{(s)}a_z^{(s)}.
\end{align}

In addition, to take care of the remaining creation operators, we replace them with indeterminates $x_k$'s as in Sec.~\ref{sec:low-mode}:
\begin{align}
    &\sum_{a\subset A\cap D}\sum_{b\subset B}\sum_{c_a,c_b\subset C}\sum_{d_b\subset D\setminus a}\prod_{i \in c_a}f_1(i) \prod_{i\in a^c}f_2(i)\prod_{i\in b}f_3(i)\prod_{i\in d_b}f_4(i)\prod_{i\in c_b} f_5(i)\prod_{i\in B\setminus b}f_6(i) \nonumber \\ 
    &\times \langle\bm{r}|\prod_{i\in C\setminus (c_a\cup c_b)}\sum_{k=1}^M(ZV^\T)_{i,k}\hat{a}_{k}^\dagger\prod_{i\in D\setminus(d_a\cup d_b)}\sum_{k=1}^M(V^\T)_{i,k}\hat{a}_{k}^\dagger|0\rangle  \\ 
    &\to \sum_{a\subset A\cap D}\sum_{b\subset B}\sum_{c_a,c_b\subset C}\sum_{d_b\subset D\setminus a}\prod_{i \in c_a}f_1(i) \prod_{i\in a^c}f_2(i)\prod_{i\in b}f_3(i)\prod_{i\in d_b}f_4(i)\prod_{i\in c_b} f_5(i)\prod_{i\in B\setminus b}f_6(i) \nonumber \\ 
    &\times \prod_{i\in C\setminus (c_a\cup c_b)}\sum_{k=1}^{2L}(ZV^\T)_{i,k}x_k\prod_{i\in D\setminus(d_a\cup d_b)}\sum_{k=1}^{2L}(V^\T)_{i,k}x_k \\ 
    &\to \sum_{a\subset A\cap D}\sum_{b\subset B}\sum_{c_a,c_b\subset C}\sum_{d_b\subset D\setminus a}\prod_{i \in c_a}f_1(i) \prod_{i\in a^c}f_2(i)\prod_{i\in b}f_3(i)\prod_{i\in d_b}f_4(i)\prod_{i\in c_b} f_5(i)\prod_{i\in B\setminus b}f_6(i)\prod_{i\in C\setminus (c_a\cup c_b)}g_1(i)\prod_{i\in D\setminus(d_a\cup d_b)}g_2(i),
\end{align}
where we defined
\begin{align}
    g_1(i)=\sum_{k=1}^{2L}(ZV^\T)_{i,k}x_k,~~~
    g_2(i)=\sum_{k=1}^{2L}(V^\T)_{i,k}x_k.
\end{align}
The squeezed states' coefficients will be taken care of later by using the indeterminates $x_k$'s.

Now we simplify the expression.
Using
\begin{align}
    &\sum_{b\subset B}\left(\prod_{i\in b} f_3(i)\prod_{i\in B\setminus b}f_6(i)\right)
    =\prod_{i\in B}(f_3(i)+f_6(i)), \\ 
    &\sum_{c_a,c_b\subset C}\left(\prod_{i\in c_a}f_1(i)\prod_{i\in c_b}f_5(i)\prod_{i\in C\setminus (c_a\cup c_b)}g_1(i)\right)
    =\prod_{i\in C}(f_1(i)+f_5(i)+g_1(i)), \\ 
    &\sum_{a\subset A\cap D}\sum_{d_b\subset D\setminus a}\left(\prod_{i\in A\setminus a}f_2(i)\prod_{i\in d_b}f_4(i)\prod_{i\in (D\setminus a)\setminus d_b}g_2(i)\right) \\
    &=\prod_{i\in A\setminus D}f_2(i)\sum_{a\subset A\cap D}\prod_{i\in (A\cap D)\setminus a}f_2(i)\sum_{d_b\subset D\setminus a}\left(\prod_{i\in d_b}f_4(i)\prod_{i\in (D\setminus a)\setminus d_b}g_2(i)\right) \\
    &=\prod_{i\in A\setminus D}f_2(i)\sum_{a\subset A\cap D}\prod_{i\in (A\cap D)\setminus a}f_2(i)\prod_{i\in D\setminus a}\left(f_4(i)+g_2(i)\right) \\
    &=\prod_{i\in A\setminus D}f_2(i) \prod_{i\in D\setminus A}(f_4(i)+g_2(i)) \sum_{a\subset A\cap D}\prod_{i\in (A\cap D)\setminus a}f_2(i)\prod_{i\in (A\cap D)\setminus a}\left(f_4(i)+g_2(i)\right) \\ 
    &=\prod_{i\in A\setminus D}f_2(i) \prod_{i\in D\setminus A}(f_4(i)+g_2(i))\prod_{i\in A\cap D}[1+f_2(i)(f_4(i)+g_2(i))],
\end{align}
the expression is simplified as 
\begin{align}
    \prod_{i\in A\setminus D}f_2(i) \prod_{i\in D\setminus A}(f_4(i)+g_2(i))\prod_{i\in A\cap D}[1+f_2(i)(f_4(i)+g_2(i))]\prod_{i\in B}(f_3(i)+f_6(i))\prod_{i\in C}(f_1(i)+f_5(i)+g_1(i)).
\end{align}
And recalling 
\begin{align}
    A\iff X,\quad
    B\iff Y\setminus A(\mathcal{M})\setminus X,\quad
    C\iff Y^c\setminus B(\mathcal{M}),\quad
    D\iff X_{\bm{m}},
\end{align}
the full expression becomes
\begin{align}
    \prod_{i\in X\setminus X_{\bm{m}}}F_i \prod_{i\in X_{\bm{m}}\setminus X}E_i\prod_{i\in X\cap X_{\bm{m}}}D_i\prod_{i\in Y\setminus A(\mathcal{M})\setminus X}A_i\prod_{i\in Y^c\setminus B(\mathcal{M})}B_i,
\end{align}
where we defined
\begin{align}
    A_i=f_3(i)+f_6(i),~
    B_i=f_1(i)+f_5(i)+g_1(i),~
    D_i=1+f_2(i)(f_4(i)+g_2(i)),~
    E_i=f_4(i)+g_2(i),~
    F_i=f_2(i).
\end{align}

Thus, we can simplify the expression as a loop hafnian:
\begin{align}
    &\sum_{Y\subset X_{\bm{n}}}\sum_{\mathcal{M}\in \text{Match}(F_{X_{\bm{n}}};Y\to Y^c)}\sum_{X\subset Y\setminus A(\mathcal{M})} \nonumber \\ 
    &\times (-1)^{|\mathcal{M}|}\left(\prod_{(i,j)\in \mathcal{M}}C_{ij}\right)\left(\prod_{i\in X\setminus X_{\bm{m}}}F_i\right)\left(\prod_{i\in X_{\bm{m}}\setminus X}E_i\right)\left(\prod_{i\in X\cap X_{\bm{m}}}D_i\right)\left(\prod_{i\in Y\setminus A(\mathcal{M})\setminus X}A_i\right)\left(\prod_{i\in Y^c\setminus B(\mathcal{M})}B_i\right) \\ 
    &=\sum_{Y\subset X_{\bm{n}}}\sum_{\mathcal{M}\in \text{Match}(F_{X_{\bm{n}}};Y\to Y^c)}\prod_{(i,j)\in \mathcal{M}}(-C_{ij})\prod_{i\in Y\setminus A(\mathcal{M})\setminus X_{\bm{m}}}(A_i+F_i)\prod_{i\in X_m}(E_iA_i+D_i)\prod_{i\in Y^c\setminus B(\mathcal{M})}B_i \\ 
    &=\prod_{i\in X_m}(E_iA_i+D_i)\lhaf(\tilde{K}),
\end{align}
where $\tilde{K}$ is a $|X_{\bm{n}}|\times |X_{\bm{n}}|$ matrix whose elements are given as
\begin{align}
    \tilde{K}_{ij}=
    \begin{cases}
        -C_{ij} & \text{if}~i<j~\text{and} \\ 
        A_i+B_i+F_i & \text{if}~i=j~\text{and}~i\notin X_{\bm{m}} \\ 
        B_i+1 & \text{if}~i=j~\text{and}~i\in X_{\bm{m}}
    \end{cases},
\end{align}
where $i,j\in X_{\bm{n}}$.
Hence, the low-rank structure of $C$ may be used.
Let us also define a matrix $K$ as the symmetrized version of $\tilde{K}$, i.e.,
\begin{align}
    K_{ij}=
    \begin{cases}
        -C_{ij} & \text{if}~i\neq j~\text{and} \\ 
        A_i+B_i+F_i & \text{if}~i=j~\text{and}~i\notin X_{\bm{m}} \\ 
        B_i+1 & \text{if}~i=j~\text{and}~i\in X_{\bm{m}}
    \end{cases},
\end{align}
and note that $\lhaf(K)=\lhaf(\tilde{K})$ if we define the loop hafnian to depend only on the upper triangular matrix~(see Appendix~\ref{app:lhaf}).
To exploit this low-rank structure, we can then invoke the following lemma:
\begin{lemma}
    [Loop hafnian of a matrix with low-rank structure~\cite{oh2024quantum}] 
    Let $\Sigma$ be an $n\times n$ symmetric matrix with rank $r$ and $A$ be a matrix constructed by filling $\Sigma$'s diagonal matrix with arbitrary elements.
    Then the loop hafnian of the matrix $A$ can be computed in $O(n\binom{2n+r-1}{r-1})$. Thus, when $r$ is fixed, the cost is polynomial in $n$ and when $r=O(\log n)$, the cost is quasipolynomial in $n$.
\end{lemma}

As an example, let us first consider the case that $\bm{n}=\bm{m}=(1,\dots,1)$, i.e., $X_{\bm{n}}=X_{\bm{m}}=[M]$.
Then, the matrix $K$ can be treated as a matrix constructed by replacing the diagonal matrix of matrix $C$ with $B_i+1$.
Hence, we can use the lemma, and the complexity is $O(M\binom{2M+2L-1}{2L-1})$. (See Appendix~\ref{app:lhaf} for the algorithm.)
More specifically, since $C$ has a rank $2L$, we find a decomposition $C=GG^\T$, where $G$ is an $M \times 2L$ matrix and then construct the polynomial
\begin{align}
    g(y_1,\dots,y_{2L})
    =\prod_{i=1}^{M}\left(\sum_{j=1}^{2L} G_{ij}y_j+B_i+1\right)(E_iA_i+D_i),
\end{align}
and multiply appropriate prefactors to obtain the loop hafnian.

Now, let us consider more general $\bm{n}$ and $\bm{m}$.
In this case, the matrix $K$ is constructed by repeating rows and columns of the matrix $C$ according to $X_{\bm{n}}$ and filling the matrix's the diagonal elements according to the definition of $K$.
Hence, this matrix still has a low-rank structure that enables us to invoke the above lemma.
One may then observe that if we let $l_i\equiv \min(n_i,m_i)$, define a polynomial
\begin{align}
    g(y_1,\dots,y_{2L})
    =\frac{1}{\sqrt{\bm{n}!\bm{m}!}}\prod_{i=1}^{M}\left[\left(\sum_{j=1}^{2L} G_{ij}y_j+B_i+1\right)^{l_i}\left(\sum_{j=1}^{2L} G_{ij}y_j+A_i+B_i+F_i\right)^{\max(0,n_i-l_i)}(E_iA_i+D_i)^{m_i}\right],
\end{align}
and multiply appropriate prefactors, we obtain the loop hafnian.
The final expression is a polynomial of $\{x_i\}_{i=1}^{2L}$ and $\{y_i\}_{i=1}^{2L}$ and $\{a_u^{(i)}\}_{i=1}^{2L}$, $\{a_v^{(i)}\}_{i=1}^{2L}$, $\{a_w^{(i)}\}_{i=1}^{2L}$, $\{a_x^{(i)}\}_{i=1}^{2L}$, $\{a_y^{(i)}\}_{i=1}^{2L}$, and $\{a_z^{(i)}\}_{i=1}^{2L}$.
Finally, to deal with the squeezed states, we replace $x_i$ with the corresponding coefficients of the squeezed states.
To be more explicit, $x_i$ is replaced with $\cosh^{-1/2}r \tanh^{i/2}r\sqrt{i!}/(2^{i/2} (i/2)!)$ for even $i$ and with $0$ for odd $i$, where $r$ is the squeezing parameter for each.
Then, $y_i$'s are taken care of using the method in Appendix~\ref{app:lhaf}, and $\{a_u^{(i)}\}_{i=1}^{2L}$, $\{a_v^{(i)}\}_{i=1}^{2L}$, $\{a_w^{(i)}\}_{i=1}^{2L}$, $\{a_x^{(i)}\}_{i=1}^{2L}$, $\{a_y^{(i)}\}_{i=1}^{2L}$, and $\{a_z^{(i)}\}_{i=1}^{2L}$ are taken care of as explained above.

Finally, let us now consider the general case:
\begin{align}
    \langle \phi|\hat{G}|\psi\rangle=\sum_{\bm{n},\bm{m}}a_{\bm{m}} b_{\bm{n}}^*\langle\bm{n}|\hat{G}|\bm{m}\rangle,
\end{align}
which is the goal of this Appendix.
Then, we define a polynomial
\begin{align}
    g(y_1,\dots,y_{2L})
    =\prod_{i=1}^{M}\sum_{n_i,m_i=0}^{n_\text{max}}\frac{a_{m_i}^{(i)}b_{n_i}^{(i)*}}{\sqrt{n_i!m_i!}}\left[\left(\sum_{j=1}^{2L} G_{ij}y_j+B_i+1\right)^{l_i}\left(\sum_{j=1}^{2L} G_{ij}y_j+A_i+B_i+F_i\right)^{\max(0,n_i-l_i)}(E_iA_i+D_i)^{m_i}\right].
\end{align}
Then, by multiplying appropriate prefactors as above for squeezed states' coefficients, we obtain $\langle \phi|\hat{G}|\psi\rangle$ by appropriately selecting coefficients of the monomials as above.

Let us now analyze the complexity.
Note that we have introduced indeterminates $\{x_i\}_{i=1}^{2L}$ and $\{y_i\}_{i=1}^{2L}$ and $\{a_u^{(i)}\}_{i=1}^{2L}$, $\{a_v^{(i)}\}_{i=1}^{2L}$, $\{a_w^{(i)}\}_{i=1}^{2L}$, $\{a_x^{(i)}\}_{i=1}^{2L}$, $\{a_y^{(i)}\}_{i=1}^{2L}$, and $\{a_z^{(i)}\}_{i=1}^{2L}$.
Thus, the number of indeterminates we introduced is $16L$.
All these indeterminates can have a degree at most $3Mn_\text{max}$ due to the quadratic form of $E_iA_i$ in the expression and $D_i=1+f_2(i)(f_4(i)+g_2(i))$.
Thus, the number of coefficients we need to keep track of is $(3Mn_\text{max}+1)^{16L}$.
Therefore, when $L=O(1)$, it can be computed in polynomial time.

\end{proof}

\section{Low-rank loop hafnian}\label{app:lhaf}
For completeness, we provide the details of a classical algorithm that computes the loop hafnian of a matrix with a low-rank structure, which was proposed in Ref.~\cite{oh2024quantum} by generalizing the algorithm for hafnian~\cite{bjorklund2019faster}.

Let us first recall the definition of the loop hafnian of an $n\times n$ matrix $A$.
Consider the graph $G$ whose adjacency matrix is $A$.
Here, the diagonal elements of $A$ represent the loops from a vertex to itself.
Then, the loop hafnian can be written as
\begin{align}
    \lhaf(A)
    \equiv \sum_{\mathcal{M} \in \text{SPM}}\prod_{(i,j)\in \mathcal{M}}A_{ij},
\end{align}
where $SPM$ is the set of all possible matchings of a graph $G$, including loops.
In many cases, as the hafnian, we assume $A$ to be a symmetric matrix because of the undirectedness of the graph.
However, one may define the loop hafnian of a matrix that is not symmetric by always considering a match $(i,j)$ as either $i\leq j$ or $i\geq j$.
Then, the loop hafnian depends only on the upper or lower triangular part, respectively.
This property was used in Ref.~\cite{bjorklund2019faster} and is used below to employ the low-rank structure for computing the loop hafnian.

Before we present the algorithm, let us define the loop hafnian in a different way.
Let $X\equiv [n]$.
Then, the loop hafnian can be written as
\begin{align}
    \lhaf(A)
    =\sum_{Y\subset X} \sum_{\mathcal{M}\in \text{Match}(F_X;Y\to Y^c)}\prod_{(i,j)\in \mathcal{M}}A_{ij}\prod_{i\in Y\setminus A(\mathcal{M})}A_{ii}/2 \prod_{i\in Y^c\setminus B(\mathcal{M})}A_{ii}/2.
\end{align}
Here, $F_X$ is the Ferrers graph, i.e., its edges are $E=\{(i,j):j<i\}$ and $1/2$ factor is introduced to cancel out the redundancy, and $A(\mathcal{M})$ and $B(\mathcal{M})$ represent the elements in $Y$ and $Y^c$ that are involved in the matching $\mathcal{M}$, respectively.
Hence, this is consistent with the argument regarding the lower triangular matrix above.
Thus, the first product corresponds to the matches between distinct vertices, and the last two products correspond to the loops.
We may generalize this expression to 
\begin{align}
    \sum_{Y\subset X} \sum_{\mathcal{M}\in \text{Match}(F_X;Y\to Y^c)}\prod_{(i,j)\in \mathcal{M}}C_{ij}\prod_{i\in Y\setminus A(\mathcal{M})}A_i \prod_{i\in Y^c\setminus B(\mathcal{M})}B_i,
\end{align}
and this is equivalent to
\begin{align}
    \lhaf(W),
\end{align}
where the matrix $W$'s off-diagonal elements are equal to $C_{ij}$ and diagonal elements are $A_i+B_i$.

We now provide the classical algorithm that takes advantage of the low-rank structure~\cite{oh2024quantum}.
Consider a symmetric matrix $\Sigma$ and another symmetric matrix $A$ which is obtained by replacing $\Sigma$'s diagonal elements by $\mu$.
When $\Sigma$ has rank $r$, it can be decomposed as $\Sigma=GG^\T$, where $G$ is an $n\times r$ matrix.
Then, we can compute the loop hafnian of the matrix as
\begin{align}
    \lhaf(A)
    &=\sum_{\mathcal{M} \in \text{SPM}}\prod_{(i,j)\in \mathcal{M}}A_{ij} 
    =\sum_{a\subset [n]:|a|:\text{even}}\haf(A_{a,a})\prod_{i\in a^c}\mu_i 
    =\sum_{\sigma \in \text{PMP}(|a|)}\prod_{i=1}^{|a|} A_{a(\sigma(2i-1)),a(\sigma(2i))}\prod_{i\in a^c}\mu_i \\ 
    &\to \prod_{i=1}^n \left(\sum_{j=1}^r G_{ij}x_j+\mu_i\right), \label{eq:lhaf_poly}
\end{align}
where $\text{PMP}$ is the set of perfect matchings without loops, and we introduced a polynomial with indeterminates $x_i$'s.
Let $\mathcal{P}_{2s,r}$ be the set of integer $r$-partitions of $2s$, i.e., tuples $(p_1,p_2,\dots,p_r)$ such that $p_i\geq 0$ for all $i$ and $\sum_i p_i=2s$, and let $\mathcal{E}_{2s,r}$ be the subset of $\mathcal{P}_{s,r}$ such that $p_i$ is even for all $i$.
We then expand the polynomial and identify the coefficients $\lambda_p$ as
\begin{align}\label{eq:lhaf_poly2}
    q(x_1,\dots,x_r)=\sum_{s=0}^{[n/2]}\sum_{p=(p_1,\dots,p_{2s})\in \mathcal{P}_{2s,r}}\lambda_p\prod_{i=1}^r x_i^{p_i},
\end{align}
where we dropped a polynomial composed of an odd number of $x_i$'s from the polynomial, Eq.~\eqref{eq:lhaf_poly}, because it is not relevant to the loop hafnian.
Since
\begin{align}
    |\mathcal{P}_{2s,r}|=\binom{2s+r-1}{r-1},~~~
    |\mathcal{E}_{2s,r}|=\binom{s+r-1}{r-1},
\end{align}
the coefficients $\lambda_p$ can be found in $O(n|\mathcal{P}_{2n,r}|)$ running time because there are $|\mathcal{P}_{2n,r}|$ monomials whose coefficients are kept track of and we have $n$ multiplications.
Meanwhile, the loop hafnian can be written as
\begin{align}
    \lhaf(A)=\sum_{s=0}^{[n/2]}\sum_{e\in \mathcal{E}_{2s,r}}\lambda_e \prod_{i=1}^r(e_i-1)!!.
\end{align}
Therefore, by expanding the polynomial~\eqref{eq:lhaf_poly2} and identifying the coefficients, we can compute the loop hafnian in $O(n|\mathcal{P}_{2n,r}|)$ and when $r$ is fixed, in $O(n^r)$.


\end{document}